\documentclass[11pt]{article}
\usepackage{graphicx}
\usepackage{color}
\usepackage{amsmath, amssymb, amsthm, bm, comment, enumerate, ascmac}
\usepackage{latexsym, mathrsfs, mathtools, psfrag, setspace, lscape, subcaption}
\usepackage{newtxtext}
\usepackage{natbib}
\usepackage[stable]{footmisc}
\usepackage{multirow}
\usepackage[colorlinks = true, linkcolor = blue, filecolor = blue, urlcolor = blue, citecolor = blue, backref = true]{hyperref}
\usepackage[margin = 2.3cm]{geometry}
\usepackage{multibib}
\usepackage{here, moresize}
\usepackage{clipboard}
\newclipboard{manuscript}

\newcites{main}{References}
\newcites{appendix}{References}

\DeclareMathOperator*{\argmax}{argmax}

\newcommand{\bE}{\mathbb E}

\newcommand{\convd}{\stackrel{d}{\longrightarrow}}

\renewcommand{\hat}{\widehat}
\renewcommand{\tilde}{\widetilde}

\allowdisplaybreaks[2]

\theoremstyle{definition}
\newtheorem{theorem}{Theorem}[section]
\newtheorem{assumption}{Assumption}[section]

\newtheorem{lemma}{Lemma}[section]

\newtheorem{remark}{Remark}[section]

\numberwithin{equation}{section}

\title{Treatment Effect Models with Strategic Interaction \\ in Treatment Decisions}

\author{Tadao Hoshino\thanks{Corresponding author: School of Political Science and Economics, Waseda University, 1-6-1 Nishi-waseda, Shinjuku, Tokyo 169-8050, Japan. Email: \href{mailto:thoshino@waseda.jp}{thoshino@waseda.jp}.} \
	 and Takahide Yanagi\thanks{Graduate School of Economics, Kyoto University, Yoshida Honmachi, Sakyo, Kyoto, 606-8501, Japan. Email: \href{mailto:yanagi@econ.kyoto-u.ac.jp}{yanagi@econ.kyoto-u.ac.jp}.}}

\date{This version: February 2023 \hspace{1cm} First version: October 2018}

\begin{document}

\onehalfspacing

\maketitle
\vspace{-1cm}
\begin{abstract}
	This study considers treatment effect models in which others' treatment decisions can affect both one's own treatment and outcome.
	Focusing on the case of two-player interactions, we formulate treatment decision behavior as a complete information game with multiple equilibria.
	Using a latent index framework and assuming a stochastic equilibrium selection, we prove that the marginal treatment effect from one's own treatment and that from the partner are identifiable on the conditional supports of certain threshold variables determined through the game model.
	Based on our constructive identification results, we propose a two-step semiparametric procedure for estimating the marginal treatment effects using series approximation.
	We show that the proposed estimator is uniformly consistent and asymptotically normally distributed.
	As an empirical illustration, we investigate the impacts of risky behaviors on adolescents' academic performance.
	
	
\end{abstract}

\newpage 

\section{Introduction}\label{sec:intro}

Estimating the marginal treatment effects (MTE) is essential in treatment evaluation.
The MTE can provide rich information on how treatment effects vary across economic agents in terms of their observed and unobserved characteristics.
Furthermore, after estimating the MTE, researchers can identify many treatment parameters of interest, such as the average treatment effects (ATE), local ATE (LATE), and policy-relevant treatment effects (PRTE) as some weighted averages of the MTE (\citealpmain{heckman1999local, heckman2005structural}).
Prior studies clearly demonstrate the usefulness of MTE methods in various empirical fields (e.g., \citealpmain{basu2007use}; \citealpmain{carneiro2011estimating}; \citealpmain{cornelissen2016late}; \citealpmain{felfe2018does}).

An important but often neglected issue in studies of treatment effects is the presence of potential interference between agents.
For example, when evaluating the effect of smoking behavior on health outcomes for couples, one partner's smoking behavior would affect the health outcome of the other.
That is, the stable unit treatment value assumption (SUTVA) does not hold due to the ``treatment spillover''.
In addition, it is natural to expect that one partner's smoking behavior interacts ``strategically'' with the other partner's smoking behavior, in the sense that one's action can directly affect the utility of smoking of the other, and vice versa.
This example suggests the presence of two different types of interference that need to be addressed: (i) treatment spillover and (ii) strategic interaction in the treatment decisions.

In reality, the co-existence of treatment spillover and strategic interaction should be common.
\phantomsection\label{page:AE-2}\Copy{AE-2}{
	Here, we provide three empirical examples that would fit into our framework.
	First, many previous studies empirically demonstrate that strategic/social interaction between friends is a primary cause of delinquency.
	It is also often observed that delinquent activities among close friends have significant impacts on students' academic performance.
	Understanding such spillover effects, as well as direct effects, is important in the education literature.
	Second, consider two competing airlines deciding whether to introduce a direct flight between a given pair of cities.
	Then, it would be interesting to investigate to what extent the number of passengers using one airline's transfer flight to travel between that city pair drops when the other introduces a direct flight.
	Third, suppose there are two candidates in a mayoral election.
	The outcome of interest is the vote share of each candidate.
	Here, we would be interested in the effect of the candidates' political position (e.g., pro-choice or pro-life) or their main campaign strategy (e.g., online or grassroots) on their outcomes.
}
\bigskip

This study aims to develop identification and estimation procedures for MTE models that allow both treatment spillover and strategic interaction in the treatment decisions.
\phantomsection\label{page:AE-3-1}\Copy{AE-3-1}{
	In particular, we focus on an empirically relevant setup in which interactions occur only within a pair of agents (e.g., couples, best friends, duopoly firms, or incumbent and challenger), and there are no treatment spillovers across the pairs.
	We postulate that they decide their treatment status simultaneously in a binary game of complete information.
	Within this framework, we formulate a set of sufficient conditions under which it is possible to point-identify the MTE parameters of interest.
	Note that if we consider each pair of players as a single observation unit (ignoring the strategic interaction), such data do not violate SUTVA, and we may be able to apply an existing causal inference method with multiple treatments to identify some causal parameters.
	However, we do not adopt such approach because doing so would make it difficult to uncover the nature of the interaction structure and, more importantly, may not allow us to perform a precise player-level treatment evaluation.
}

To achieve our goal, we need to address the following two issues.
The first issue is the non-applicability of {\it (unordered) monotonicity} (\citealpmain{imbens1994identification}; \citealpmain{heckman1999local, heckman2005structural}; \citealpmain{heckman2018unordered}).
The monotonicity conditions require that shifts in instrumental variables (IVs) determine the direction of changes in the treatment choices uniformly for all agents.
\phantomsection\label{page:AE-4-1}\Copy{AE-4-1}{
	\citetmain{vytlacil2002independence} and \citetmain{heckman2018unordered} demonstrate that these monotonicity conditions are equivalent to assuming that each treatment realization is characterized by a single threshold-crossing equation in which the IVs and an unobserved error term are weakly separable.
	However, in our case, an agent's treatment choice may complexly depend on that of another agent through strategic interaction. 
	As a result, each player's treatment choice cannot be expressed as a weakly separable threshold crossing model, implying the failure of the monotonicity.
}
The second issue is the possibility of multiple equilibria in the treatment decisions. 
The presence of multiple equilibria leads to an {\it incomplete} econometric model (e.g., \citealpmain{tamer2003incomplete}; \citealpmain{lewbel2007coherency}; \citealpmain{ciliberto2009market}; \citealpmain{chesher2020structural}) in the sense that the model-consistent treatment assignment is not unique.
The issue of incompleteness is common in the literature on game model estimation; however, it is not yet well understood in the context of treatment evaluation.

Our identification strategy solves these two issues simultaneously by combining the local IV (LIV) method in \citemain{heckman1999local, heckman2005structural} and a stochastic equilibrium selection rule in the treatment decision game.
The key idea is to use local variations of player-specific continuous IVs that alter players' treatment status, but do not directly affect their outcomes.
Although this is a natural extension of the LIV method to a multi-dimensional space, it is still insufficient to point-identify the MTEs due to the presence of multiple equilibria.
We overcome this issue by explicitly introducing an equilibrium selection rule.

To identify the MTE parameters, we first need to identify the parameters of the treatment decision game.
We model the treatment decision game as a binary game of complete information.
In particular, for the identification of MTE parameters, the model needs to have strategic interaction effects that are functions of the IV.
This would be empirically plausible under many different situations.
\phantomsection\label{page:R1-6-1}\Copy{R1-6-1}{
	We then provide a new identification result for our game model based on a large support condition and a particular dependence property for a parametric copula function with a scalar correlation parameter (cf. \citealpmain{han2017identification}).
}
We also show that the large support condition can be mitigated when we introduce additional parametric model restrictions.
\phantomsection\label{page:R1-6-2}\Copy{R1-6-2}{
	Although treatment evaluation is our main concern, these identification results provide independent contributions to the literature; to the best of our knowledge, we would be the first to address the identification of game models where the joint distribution of the unobserved payoff components is characterized by a certain class of parametric copulas.
}

Given the identification of the treatment decision game, we present the following identification results for treatment parameters.
First, we can point-identify the following MTE parameters: the {\it direct} MTE, in which only one agent's treatment status switches from untreated to treated, whereas the partner's remains unchanged; the {\it indirect} MTE, in which only the partner's treatment status switches from untreated to treated, and the focal agent's status remains unchanged;
and the {\it total} MTE, in which the treatment status of both players switches from untreated to treated.
In contrast to the conventional MTE framework, our MTEs can reveal the treatment effect heterogeneity in terms of the pair of unobservables. 
We show that the regions in which the MTEs are identifiable are determined by the conditional supports of appropriately transformed threshold variables in the game model.
This result extends \citemain{heckman1999local,heckman2005structural} and \citemain{carneiro2009estimating}, who prove the identification of conventional MTEs in terms of the supports of propensity scores.
\phantomsection\label{page:AE-6-1}\Copy{AE-6-1}{
	Second, we demonstrate that the MTE parameters are over-identified under our identification conditions, which is essentially a consequence of the stochastic equilibrium selection assumption.
	This over-identification result provides an opportunity to improve the efficiency of MTE estimation.
}
Finally, we present the identification of several other treatment parameters, including the LATE and PRTE (these results are relegated to Appendix \ref{sec:severalparameter}).

Our identification is constructive in that we can estimate the MTEs directly by following the identification strategy.
We propose a two-step semiparametric procedure for estimating the MTEs.
In the first step, we estimate the parameters in the treatment decision game using a maximum likelihood (ML) approach based on a fully parametric model specification.
Using the ML estimates, we estimate the MTEs in the second step by employing semiparametric series (sieve) technique.
The proposed estimator is uniformly consistent with the optimal convergence rate and is asymptotically normally distributed.
In addition, the estimator possesses an oracle property; that is, its limiting distribution is the same as that of the infeasible estimator where the parameters in the treatment decision game are known.

To illustrate our methods empirically, we investigate the impacts of the delinquency (e.g., smoking and drinking) of an opposite-gender best friend on the academic outcomes of an adolescent, i.e., their grade point average (GPA).
Following the literature (e.g., \citealpmain{card2013peer}), we model the decision to participate in risky activities as a complete information game.
Our method revealed the following empirical evidence: (i) the susceptibility to peer effects varies with the personality of the students, (ii) the direct treatment effect of risky behaviors on the GPA is significantly negative for both male and female students, and (iii) the total treatment effect is even larger than the direct effect for both genders.
The third finding implies that the delinquency of a best friend has a negative causal impact on the academic performance of the students.
These results partially overlap with those obtained in previous empirical studies, but the current study would be the first to report such results by formally integrating social interactions in risky activities and the resulting causal effects on academic performance in a single framework.
We also demonstrate that if we ignore the strategic interaction between each pair of students and estimate the treatment model simply as a bivariate probit model, the resulting MTE estimates can differ significantly from those obtained in our framework.
This indicates the importance of correctly addressing the interaction structure in the treatment choices, not just as correlated multiple treatments.

\bigskip

\phantomsection\label{page:R2-1-2}\Copy{R2-1-2}{
	This study has a clear connection to both the treatment effect and the structural estimation literature.
	By formulating the treatment choice as a structural game model within the MTE framework, we can take advantage of both reduced-form causal inference and structural estimation.
	In particular, unlike a structural approach that fully specifies the outcome and treatment choice equations, our approach provides some robustness to the problem of model misspecification for the outcome equation.
	At the same time, compared to the conventional reduced-form analysis, we can perform a certain class of interesting counterfactual exercises through the structural estimation of treatment decision game.
}

To the best of our knowledge, few studies have addressed treatment evaluation in the presence of strategic interactions modeled explicitly as games.
One important exception is \citetmain{balat2020multiple}.
Their approach is more general than ours in that it allows for nonparametric strategic interactions with more than two players and does not require point identification of the game parameters.
To overcome the multiple equilibria issue, \citetmain{balat2020multiple} employ nonparametric shape restrictions and variations of possibly discrete IVs.
However, because of such generality, the treatment parameter of interest in their study, ATE, is only partially identified in general.
In contrast, we focus only on two-player games with a stochastic equilibrium selection to establish point identification for the treatment decision game and MTE.
While this loses some generality, we can obtain rich information on the unobserved heterogeneity in the treatment effects that cannot be learned by estimating ATE alone.
Thus, these two studies complement each other from different perspectives.

Another closely related study is \citetmain{lee2018identifying}, who showed that it is possible to identify MTE parameters by modeling the treatment selection with a set of threshold crossing rules.
Their identification strategy is essentially the same as ours.
They characterize a two-player binary game by combining five threshold crossing rules and identify the MTE parameters by computing the changes in the expected outcome with respect to local variations of all these threshold variables (see Appendix C of \citealpmain{lee2018identifying}).
\phantomsection\label{page:AE-7-1}\Copy{AE-7-1}{
	Compared to their model, each treatment realization in our model can be characterized by a smaller set of threshold crossing equations and consequently the target MTE parameters are not completely identical between these studies.
	This difference is crucial from a practical perspective because our MTE parameter can be estimated reasonably well even with moderate sample size (see Remark \ref{remark:leesalanie} for more details).
}
\phantomsection\label{page:AE-8-1}\Copy{AE-8-1}{
	In addition, they do not discuss in depth how to cope with the identification problem in the presence of multiple equilibria, whereas we formally prove the identification of both the game model and the treatment parameters under multiplicity.
}

\bigskip

The rest of the paper is organized as follows.
In Section \ref{sec:model}, we introduce our model and review the incompleteness problem for discrete game models.
Section \ref{sec:identification} presents the identification analysis.
We develop our MTE estimators and study their asymptotic properties in Section \ref{sec:estimation}. 
Section \ref{sec:numeric} provides the numerical illustrations, including Monte Carlo simulations and the empirical analysis of adolescents' academic performance.
Section \ref{sec:conclusion} concludes the paper.
The proofs of all theorems are provided in Appendix \ref{sec:proof}, and the other supplementary technical results in Appendix \ref{appendix:supptech}.
In Appendix \ref{sec:severalparameter}, we discuss the identification of several treatment parameters besides the MTE.
Appendix \ref{sec:MC} presents detailed information about the Monte Carlo experiments in Section \ref{sec:numeric}.
Finally, we provide supplementary material for the empirical analysis in Appendix \ref{sec:supp}.

\section{Model} \label{sec:model}

In this section, we introduce our treatment effect model with strategic interaction.
We denote a player by $j \in \{1, 2\}$ and his/her partner (or opponent) by $-j$. 
We aim to evaluate the effects of player $j$'s treatment $D_j \in \{0, 1\}$ and/or partner's treatment $D_{-j} \in \{0, 1\}$ on player $j$'s outcome $Y_j$ and/or the partner's outcome $Y_{-j}$.
The outcomes may or may not be common to both players; that is, we allow both $Y_j = Y_{-j}$ and $Y_j \neq Y_{-j}$.
For $(d_j, d_{-j}) \in \{ 0, 1 \}^2$, let $Y_j^{(d_j, d_{-j})}$ be the potential outcome when $j$'s own treatment status is $D_j = d_j$ and the partner's is $D_{-j} = d_{-j}$.
Then, the observed outcome can be written as $Y_j = \sum_{d_j}\sum_{d_{-j}} I^{(d_j, d_{-j})} Y_j^{(d_j, d_{-j})}$, where $I^{(d_j, d_{-j})} \coloneqq \mathbf{1}\{(D_j, D_{-j}) = (d_j, d_{-j})\}$.
\phantomsection\label{page:AE-9-1}\Copy{AE-9-1}{
	Suppose that player $j$'s potential outcome can be written as
	\begin{align}\label{eq:model-Y1}
		Y_j^{(d_j, d_{-j})} = \mu_j^{(d_j, d_{-j})} \left( X_j, U_j^{(d_j, d_{-j})} \right),
	\end{align}
	where $X_j \in \mathbb{R}^{\mathrm{dim}(X)}$ is a vector of observable covariates, $U_j^{(d_j, d_{-j})}$ is an unobserved random variable with arbitrary dimension, and $\mu_j^{(d_j, d_{-j})}$ is an unknown structural function.
}
The covariates $X_1$ and $X_2$ may contain common elements as well as some player-specific elements.
For simplicity, we assume that the dimensions of $X_1$ and $X_2$ are both equal to $\mathrm{dim}(X)$, and the same assumption will be made for the other variables.

\phantomsection\label{page:AE-18-1}\Copy{AE-18-1}{
	Note that we do not explicitly consider the stochastic process on who becomes whose partner, but we treat the pair formation as a given object.
	In other words, all subsequent analyses are conditioned on the existing pair formation and may not be generalizable to other pair data.
	While this could be a shortcoming of this study, such an approach has been adopted quite commonly in the social interaction literature to circumvent the notorious endogenous network formation problem (with a few exceptions: e.g., \citealpmain{goldsmith2013social,hsieh2016social,johnsson2021estimation}).
}

Hereinafter, for a generic player-specific variable, say $A_j$, we write $A$ without a subscript to denote the union $A = (A_1, A_2)$.
For example, $Y = (Y_1, Y_2)$, $X = (X_1, X_2)$, and $U^{(d_1, d_2)} = ( U_1^{(d_1, d_2)}, U_2^{(d_2, d_1)} )$.
In addition, let $F_{A | E}$ denote the cumulative distribution function (CDF) of $A$ conditional on $E$.
When $A$ is continuously distributed given $E$, we write its conditional probability density function (PDF) as $f_{A|E}$.

\subsection{Strategic treatment decision}\label{subsec:treatment}

Suppose that each player $j$ obtains the payoff $\pi_j(D_{-j}, W_j) - \varepsilon_j$ when $D_j = 1$, and he/she obtains $0$ (for normalization) when $D_j = 0$, where $W_j \coloneqq (X_j^\top, Z_j^\top)^\top$ is a vector that includes $X_j$ and the instruments $Z_j \in \mathbb{R}^{\mathrm{dim}(Z)}$, $\varepsilon_j \in \mathbb{R}$ is an unobserved continuous variable, and $\pi_j$ is an unknown function.
Then, for $j = 1,2$, the payoff function can be written as follows:
\begin{align*}
	u_j(d_j, d_{-j}) \coloneqq d_j \left[ \pi_j^0(W_j) + d_{-j} \cdot \Delta_j(W_j) - \varepsilon_j \right],
\end{align*}
where $\Delta_j(W_j) \coloneqq \pi_j^1(W_j) - \pi_j^0(W_j)$ corresponds to the strategic interaction effect with $\pi_j^0(W_j) \coloneqq \pi_j(0, W_j)$ and $\pi_j^1(W_j) \coloneqq \pi_j(1, W_j)$.\footnote{
	This expression clearly indicates that the additive separability imposed on the payoff function implies the interaction effect to be a function of only the observables.
	In our empirical setting, this requires that the impact of friends' delinquency is independent of own unobserved factors, such as his/her latent attitude towards risky activities.
	With few exceptions (e.g., \citealpmain{kline2015identification}), models where the interaction effects can depend on some unobservables have not been studied in detail in the literature.
	\label{foot:additivity}
 	}
Whereas the strategic interaction effect is often assumed to be constant in the game econometrics literature, we require it to be a non-constant function of the instruments for identification of the MTE parameters.

Based on the payoff-maximization principle, the best response is $D_j(d_{-j}) = \argmax_{d_j} u_j(d_j, d_{-j})$.
Here, we assume that $(W, \varepsilon)$ is common knowledge for both players (i.e., a complete information setup).
Furthermore, assume that the set of realized treatments $(D_j, D_{-j})$ is characterized as a pure strategy Nash equilibrium.\footnote{
	Most of the studies in the literature focus on pure strategy Nash equilibria as the solution concept for $2 \times 2$ games, except when no Nash equilibria in pure strategies exist, for example, because of asymmetric strategic interaction (e.g., \citealpmain{bjorn1984simultaneous}; \citealpmain{tamer2003incomplete}; \citealpmain{depaula2013econometric}).
	For the motivations of this choice, see, for example, Section 6 of \citemain{bjorn1984simultaneous}.
}
Then, the players' treatment decisions follow the simultaneous binary response model:
\begin{align}\label{eq:model-D1}
	\begin{split}
		D_j 
		& = \mathbf{1} \left\{ \pi_j (D_{-j}, W_j) \ge \varepsilon_j \right\} \\
		& = \mathbf{1} \left\{ \pi_j^0(W_j) + D_{-j} \cdot \Delta_j(W_j) \ge \varepsilon_j \right\}.
	\end{split}
\end{align}
The first line of \eqref{eq:model-D1} shows that our treatment decision model can be viewed as a direct extension of the latent index model of \citetmain{heckman1999local, heckman2005structural}.\footnote{\label{page:R2-1-e2}\Copy{R2-1-e2}{
	We can see that the additive separability in the treatment payoff function essentially rules out unobserved preference heterogeneity in the treatment choice. 
	One possible way to relax this restriction is to use a finite-mixture model to permit the preference heterogeneity across pairs, as in \citetmain{hoshino2022estimating}, but we leave this task for a future study.
}
}
\phantomsection\label{page:AE-10-1}\Copy{AE-10-1}{
	As we will discuss in Remark \ref{remark:monotonicity}, the model in \eqref{eq:model-D1} does not satisfy the IV monotonicity.
	Studies in the MTE literature that consider a ``non-monotonic'' treatment model similar to ours include \citemain{klein2010heterogeneous} and \citemain{lee2018identifying}.
}

In the following, we assume that the treatment decisions are known to be strategic complements (to be in line with our empirical application).

\begin{assumption}\label{as:complement}
	\hfil
	\begin{enumerate}[(i)]
		\item For both $j = 1, 2$, $\Delta_j(W_j) > 0$ almost surely (a.s.).
		\item \phantomsection\label{page:AE-11-1}\Copy{AE-11-1}{
			$(\varepsilon_1, \varepsilon_2)$ are conditionally independent of $Z$ given $X$ and continuously distributed with strictly increasing marginal CDFs $F_{\varepsilon_1 | X}$ and $F_{\varepsilon_2 | X}$, respectively.
		}
	\end{enumerate}
\end{assumption}

Note that Assumption \ref{as:complement}(i) is empirically testable.\footnote{
	For example, \citetmain{aradillas2019nonparametric} has developed nonparametric tests for the presence and direction of strategic interaction effects in $2 \times 2$ games of complete information.
	For another example, one may use a Vuong-type model selection test for non-nested alternatives (e.g., \citealpmain{hsu2017model}; \citealpmain{schennach2017simple}).
	}
We also note that our analysis can be easily converted to the situation with strategic substitutes where $\Delta_j(W_j) < 0$ a.s. for both $j = 1, 2$.
We introduce Assumption \ref{as:complement}(ii) not only for the identification of the game model but also for that of MTE.

Define
\begin{align}\label{eq:definition-P}
	V_j \coloneqq F_{\varepsilon_j | X} ( \varepsilon_j ),
	\qquad 
	P_j^0 \coloneqq F_{\varepsilon_j | X} ( \pi_j^0(W_j) ),
	\qquad
	P_j^1 \coloneqq F_{\varepsilon_j | X} ( \pi_j^1(W_j) ).
\end{align}
By construction, $V_j $ is distributed as $\text{Uniform}[0,1]$ conditional on $X$.
Then, we can rewrite \eqref{eq:model-D1} as
\begin{align*}\label{eq:model-D2}
	D_j = \mathbf{1} \left\{ P_j^0 + D_{-j} (P_j^1 - P_j^0) \ge V_j \right\}.
\end{align*}
By Assumption \ref{as:complement}, $P_j^1 - P_j^0 > 0$ a.s. for both $j = 1, 2$.

\subsection{Incompleteness} \label{subsec:incomplete}

A major difficulty in our model is the {\it incompleteness} of the treatment decision model.
\phantomsection\label{page:AE-11-2}\Copy{AE-11-2}{
	We have the following relationships:
	\begin{equation*} 
		\begin{array}{lcllllcl}
			D = (1, 0) & \Longleftrightarrow & V_1 \le P_1^0, \; V_2 > P_2^1, & & & 
			D = (0, 1) & \Longleftrightarrow & V_1 > P_1^1, \; V_2 \le P_2^0, \\
			D = (1, 1) & \Longrightarrow & V_1 \le P_1^1, \; V_2 \le P_2^1, & & &
			D = (0, 0) & \Longrightarrow & V_1 > P_1^0, \; V_2 > P_2^0.
		\end{array}
	\end{equation*}
}
Figure \ref{fig:nash-comp1} visually summarizes these relationships.
As shown in the figure, the space of $(V_1, V_2)$ cannot be partitioned into non-overlapping regions associated with the four alternative realizations of $D$.
Both $D = (1, 1)$ and $D = (0, 0)$ can occur when $P_1^0 < V_1 \le P_1^1$ and $P_2^0 < V_2 \le P_2^1$ (i.e., multiple equilibria).
This non-uniqueness of model-consistent decisions is called incompleteness and has been extensively studied in the literature on simultaneous equation models for discrete outcomes.

\begin{figure}[!h]
	\centering
	\fbox{\includegraphics[width = 8cm, bb = 0 0 720 540]{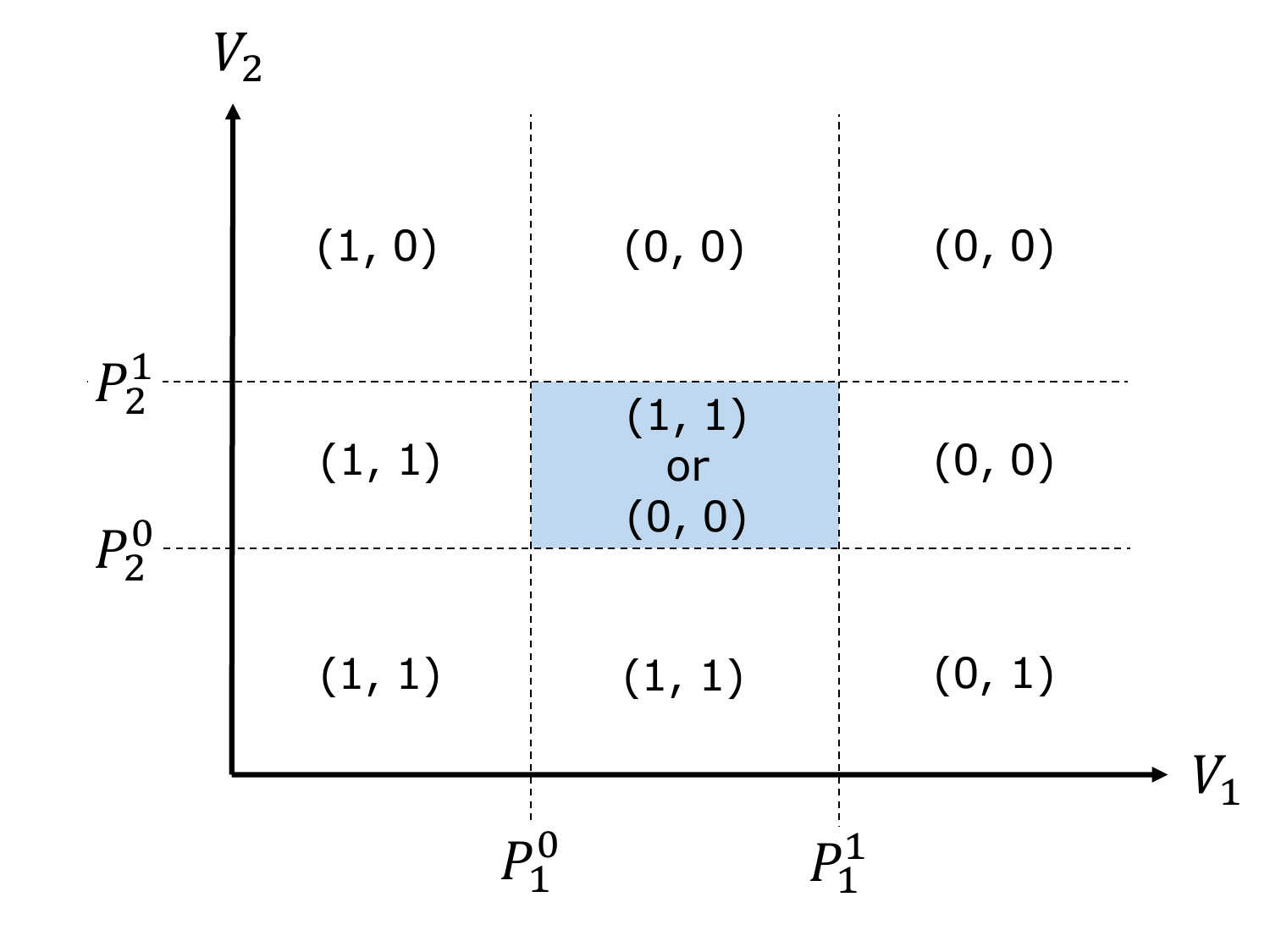}}
	\caption{Nash equilibrium under strategic complementarity.}
	\label{fig:nash-comp1}
\end{figure}

In the game econometrics literature, there are several approaches to handle the incompleteness problem to achieve point identification.\footnote{
	See \citetmain{depaula2013econometric} for an excellent survey on this topic.
} 
One major approach is to explicitly introduce a stochastic (or possibly deterministic) equilibrium selection mechanism (e.g., \citealpmain{bjorn1984simultaneous}; \citealpmain{kooreman1994estimation}; \citealpmain{soetevent2007discrete}; \citealpmain{bajari2010identification}; \citealpmain{card2013peer}).
This approach allows us to predict the choice behavior in the region of multiplicity and to perform counterfactual exercises, which is essential to establish the identification of the MTE parameters, and thus is adopted in this study.

\begin{figure}[!h]
	\centering
	\fbox{\includegraphics[width = 8cm, bb = 0 0 720 540]{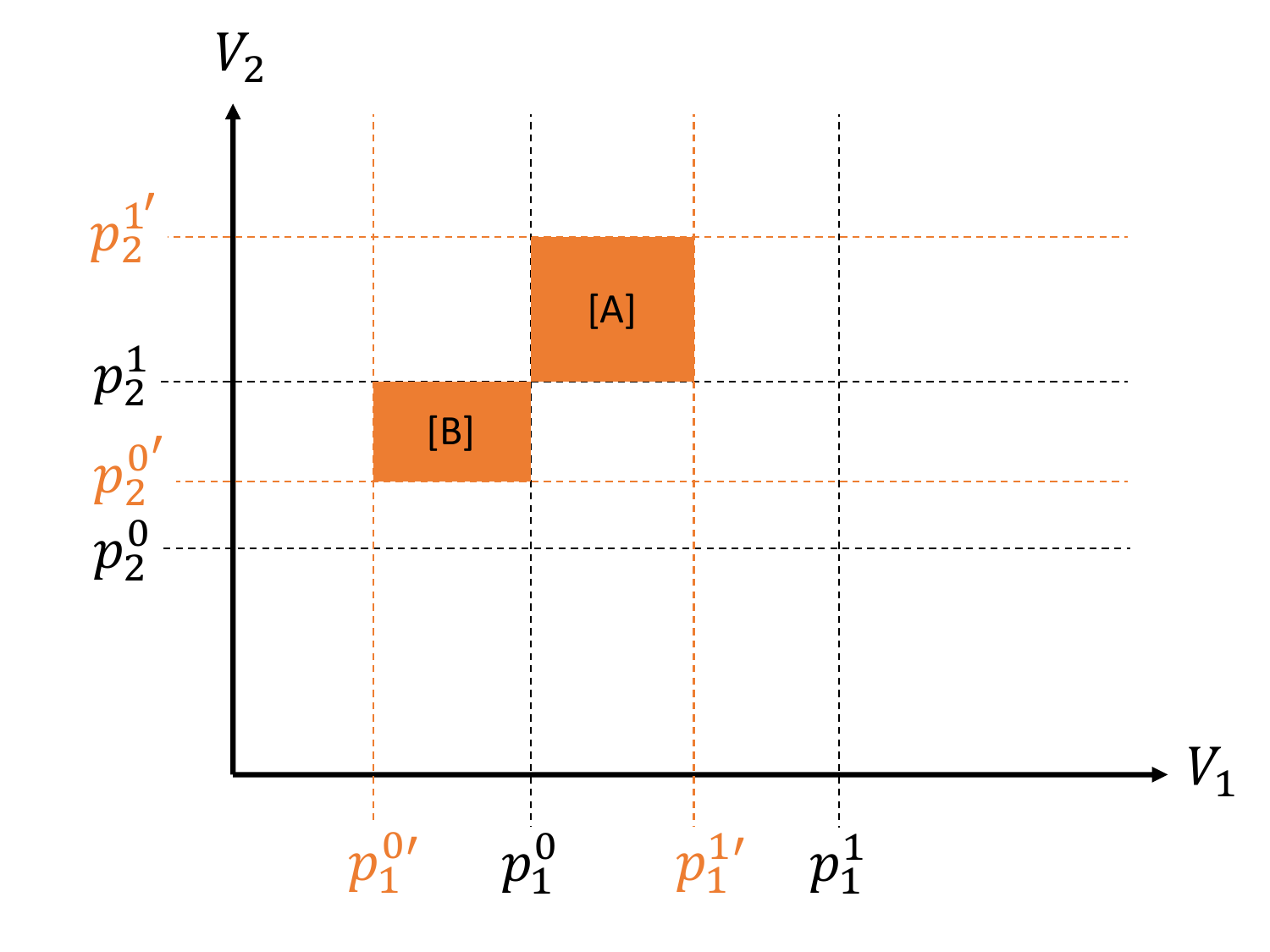}}
	\caption{Failure of the monotonicity conditions.}
	\label{fig:monotonicity}
\end{figure}

\begin{remark}[Monotonicity] \label{remark:monotonicity}
	To see why the monotonicity conditions are not applicable to our situation, let $D_j(w)$ be the potential treatment status of player $j$ when $W = w$, and $D(w) = (D_1(w), D_2(w))$.
	The (ordered) monotonicity in \citetmain{imbens1994identification} is that, for any $w, w' \in \text{supp}[W]$ and $j = 1, 2$, either $\Pr[ D_j(w) \ge D_j(w') ] = 1$ or $\Pr[ D_j(w) \le D_j(w') ] = 1$ should hold.
	The unordered monotonicity in \citetmain{heckman2018unordered} is that, for any $w, w' \in \text{supp}[W]$ and $(d_1, d_2) \in \{ 0, 1 \}^2$, either $\Pr[ \mathbf{1}\{ D(w) = (d_1, d_2) \} \ge \mathbf{1}\{ D(w') = (d_1, d_2) \} ] = 1$ or $\Pr[ \mathbf{1}\{ D(w) = (d_1, d_2) \} \le \mathbf{1}\{ D(w') = (d_1, d_2) \} ] = 1$ should hold.
	Here, consider two groups of pairs: the one in the region of $D = (0, 0)$ and the other in the region of $D = (1, 1)$ for a given $W = w$.
	Suppose that shifting $W$ from $w$ to $w'$ induces both groups to the multiple equilibria.
	Figure \ref{fig:monotonicity} illustrates such a situation where the pairs in the former group and those in the latter are respectively located in regions [A] and [B].
	Here, $(p_1^0, p_1^1, p_2^0, p_2^1)$ and $(p_1^{0'}, p_1^{1'}, p_2^{0'}, p_2^{1'})$ correspond the values of $(P_1^0, P_1^1, P_2^0, P_2^1)$ when $W = w$ and $W = w'$, respectively.
	In this situation, some of the former pairs would switch their treatment statuses from $D = (0, 0)$ to $(1, 1)$, while some of the latter pairs would switch from $D = (1, 1)$ to $(0, 0)$.
	Thus, the impact of $W$ on the treatment choice may be non-monotonic.
\end{remark}

\section{Identification}\label{sec:identification}

In this section, we present the identification results for the treatment decision game and the MTE parameters.
We refer to the treatment decision game \eqref{eq:model-D1} as the ``first stage'' and the realization of the outcome \eqref{eq:model-Y1} as the ``second stage''.
Because the identification of MTE relies on the knowledge of the first-stage parameters, we first examine the identification of them in Subsection \ref{subsec:first}, and then present the identification results for the second-stage parameters in Subsection \ref{subsec:second}.

\subsection{Identification of the first-stage parameters} \label{subsec:first}

We first introduce the following assumptions on the joint distribution of $(\varepsilon_1, \varepsilon_2)$.

\begin{assumption}
	\hfil 
	\label{as:copula}
	\begin{enumerate}[(i)]
		\item 
		The joint distribution of $(\varepsilon_1, \varepsilon_2)$ given $X = x$ is represented by the copula function $H_{\rho_x}$ such that $\Pr[\varepsilon_1 \le a_1, \varepsilon_2 \le a_2 | X = x] = H_{\rho_x}(F_{\varepsilon_1 | X = x}(a_1), F_{\varepsilon_2 | X = x}(a_2))$, where $\rho_x \in (\underline{c}, \bar c)$ is a scalar correlation parameter and $\underline{c}$ and $\bar c$ are real numbers whose values depend on the choice of the copula function.
		\item $H_{\rho_x}(\cdot, \cdot)$ is twice differentiable in its arguments and ${\rho_x}$.
		\item $H_{\rho_x}(\cdot, \cdot)$ is strictly more stochastically increasing in joint distribution with respect to ${\rho_x}$ (see Definition 3.3 of \citetmain{han2017identification}).
	\end{enumerate}
\end{assumption}
\phantomsection\label{page:AE-12-1}\Copy{AE-12-1}{
	Assumption \ref{as:copula}(iii) restricts the dependence ordering of the copula function in terms of stochastic monotonicity.
	\citemain{han2017identification} use this property to identify generalized bivariate probit models, and they demonstrate that various commonly used copula functions satisfy it.
	Note that this dependence property is only for the class of copulas with a scalar correlation parameter.
	Investigating the identification with a more general multi-parameter copula is beyond the scope of our study.
}
A typical example satisfying this assumption is a standard bivariate normal distribution, in which $H_{\rho_x}$ corresponds to the Gaussian copula $H_{\rho_x}(v_1, v_2) = \Phi_2(\Phi^{-1}(v_1), \Phi^{-1}(v_2); \rho_x)$, with $\Phi_2(\cdot, \cdot; \rho_x)$ and $\Phi(\cdot)$ being the standard bivariate normal CDF with correlation $\rho_x$ and the standard normal CDF, respectively.
For another example, the assumption is also satisfied with the Farlie--Gumbel--Morgenstern (FGM) copula $H_{\rho_x}(v_1, v_2) = v_1v_2[1 + \rho_x(1 - v_1)(1 - v_2)]$.
For other examples and further discussions on the dependence ordering properties of copula functions, see \citemain{han2017identification}.

\begin{assumption} \label{as:multiple}
	In the case of multiple equilibria for a given $X = x$, $D = (0, 0)$ occurs if and only if $\epsilon \le \lambda_x$, where $\lambda_x \in [0, 1]$ is a constant and $\epsilon$ is a random variable distributed as $\text{Uniform}[0, 1]$ independent of $(Z, \varepsilon, U^{(d_1, d_2)})$ given $X$ for $(d_1,d_2)\in\{0,1\}^2$.
\end{assumption}

This assumption states that, given $X = x$, $D = (0, 0)$ is observed with probability $\lambda_x$ in the multiple equilibria situation.
For example, some authors assume that under multiple equilibria, one of model-consistent actions is selected uniformly at random (e.g., \citealpmain{bjorn1984simultaneous}; \citealpmain{soetevent2007discrete}; \citealpmain{card2013peer}).
If we adopt the same assumption, we can set $\lambda_x = 0.5$ a priori.
As another example, one may assume that the realized treatment status corresponds to the ``largest'' Nash equilibrium, as in \citemain{xu2015estimation}.
In this case, because $u_j(1,1) > u_j(0,0)$ in the multiple equilibria region for both players, $D = (1,1)$ holds almost surely, i.e., $\lambda_x = 0$.
Assumption \ref{as:multiple} allows for a more general equilibrium selection in that $\lambda_x$ can be unknown and can depend on $X$.
\phantomsection\label{page:R2-2-1}\Copy{R2-2-1}{
	However, note that the IV $Z$ should not affect equilibrium selection.
	Such an assumption is not necessary for identifying the game model  but is required for nonparametrically identifying MTE (see also Remark \ref{remark:brinch}).
}

In addition, Assumption \ref{as:multiple} rules out the possibility of ``endogenous'' equilibrium selection such that the selection probability depends on unobservables $(\varepsilon, U^{(d_1,d_2)})$.
In our empirical setting, this requires that the latent academic abilities and attitudes toward risky activities do not affect which is selected between the two equilibria.
Although this requirement might be restrictive for certain empirical applications, it would be challenging to identify such a game model (cf. \citealpmain{jun2020counterfactual}).

For any given $x \in \text{supp}[X]$ and $w \in \text{supp}[W | X = x]$, we write the conditional probability of $D = (d_1, d_2)$ as $\mathcal{L}^{(d_1, d_2)}(w) \coloneqq \Pr[D = (d_1, d_2) | W = w]$ and the value of $P_j^d$ as $p_j^d = F_{\varepsilon_j | X = x}(\pi_j^d(w_j))$.
Then, the above assumptions give
\begin{equation} \label{eq:moments}
	\begin{array}{ll}
		\mathcal{L}^{(1, 0)}(w) = p_1^0 - H_{\rho_x}(p_1^0, p_2^1), & \quad \mathcal{L}^{(1, 1)}(w) = H_{\rho_x}(p_1^1, p_2^1) - \lambda_x \cdot \mathcal{L}_{\text{mul}}(w), \\
		\mathcal{L}^{(0, 1)}(w) = p_2^0 - H_{\rho_x}(p_1^1, p_2^0), & \quad \mathcal{L}^{(0, 0)}(w) = 1 - \sum_{(d_1,d_2) \neq (0,0)} \mathcal{L}^{(d_1, d_2)}(w),
	\end{array}
\end{equation}
where $\mathcal{L}_{\text{mul}}(w) \coloneqq H_{\rho_x}(p_1^1, p_2^1) - H_{\rho_x}(p_1^1, p_2^0) - H_{\rho_x}(p_1^0, p_2^1) + H_{\rho_x}(p_1^0, p_2^0)$ is the probability that the pair of players is in the multiple equilibria region.
There are six unknown parameters $(\mathbf{p}, \rho_x, \lambda_x)$ in the moment equations in \eqref{eq:moments}, where $\mathbf{p} = (p_1^0, p_1^1, p_2^0, p_2^1)$.

\begin{assumption}
	\hfil 
	\label{as:infinity}
	\begin{enumerate}[(i)]
		\item For both $j = 1, 2$, there exists a player-specific continuous random variable, say $W_{j,1}$, whose PDF is everywhere positive on $\mathbb{R}$ given $(W_{j,-1}, W_{-j})$, where $W_j = (W_{j, 1}, W_{j, -1})$, so that each of $\pi_j^0(W_j)$ and $\pi_j^1(W_j)$ is non-degenerate and continuously distributed conditional on $(W_{j,-1}, W_{-j})$.
		\item For both $j = 1, 2$, $\pi_j^1(w_{j, 1}, w_{j, -1}) \to -\infty$ as $w_{j, 1} \to -\infty$ for any $w_{j, -1} \in \text{supp}[W_{j, -1} | X = x]$.
	\end{enumerate}
\end{assumption}

Assumption \ref{as:infinity}(i) requires that at least one player-specific element of $W_j$ can tend to $-\infty$ and $\infty$, and Assumption \ref{as:infinity}(ii) says that this player-specific variable has a significant impact on the payoff function $\pi_j^1$.
We note that this assumption can be replaced by $\pi_j^1(w_{j, 1}, w_{j, -1}) \to -\infty$ as $w_{j, 1} \to \infty$.
Some prior studies have utilized assumptions similar to ours (e.g., \citealp{tamer2003incomplete}; \citealp{kline2015identification}).
Importantly, this assumption is empirically verifiable because almost all pairs with sufficiently small $w_{1, 1}$ and $w_{2, 1}$ should choose $D = (0, 0)$ if it is true.

The next theorem presents identification for a general parameter value $(\mathbf{p}, \rho_x, \lambda_x)$, for any given $x \in \text{supp}[X]$ and $w \in \text{supp}[W | X = x]$.
To the best of our knowledge, the identification result for the correlation parameter $\rho_x$ is novel in the literature of discrete games with complete information.

\begin{theorem}\label{thm:gameiden}
	\phantomsection\label{page:AE-11-3}\Copy{AE-11-3}{
		Fix arbitrary $x \in \text{supp}[X]$ and $w \in \text{supp}[W | X = x]$.
		Let $p_j^d = F_{\varepsilon_j | X = x}(\pi_j^d(w_j))$ with $d = 0, 1$ and $j = 1, 2$.
		\begin{enumerate}[(i)]
			\item Suppose that Assumptions \ref{as:complement}, \ref{as:copula}(i), and \ref{as:infinity} hold.
			Then, $p_j^0$ is identified for $j = 1, 2$.
			\item In addition, suppose that Assumptions \ref{as:copula}(ii)--(iii) hold.
			If $\text{supp}[P_j^1 | X = x] \times (\underline{c}, \bar c)$ is a simply connected set, then $(p_j^1, \rho_x)$ are identified for $j = 1,2$.
			\item If Assumption \ref{as:multiple} additionally holds, $\lambda_x$ is identified.
		\end{enumerate}
	}
\end{theorem}

Note that $\mathbf{p} = (p_1^0, p_1^1, p_2^0, p_2^1)$ can be identified without assuming any particular form of equilibrium selection.
Indeed, the equilibrium selection assumption is imposed not to simplify the identification of the game model but to achieve the identification of the MTE parameters.

\begin{remark}[Identification without the large support condition]\label{remark:infinity}
	The large support condition in Assumption \ref{as:infinity} can be replaced by other identification conditions.
	For example, consider $(w_1, w_2)$, $(w_1', w_2)$, $(w_1, w_2')$, $(w_1', w_2') \in \text{supp}[W | X = x]$ such that $w_j \neq w_j'$ for $j = 1, 2$.
	Define $\vartheta_{w, w'} \coloneqq (\mathbf{p}, \mathbf{p}', \rho_x, \lambda_x)$ and
	\begin{align*}
		\mathcal{G}(w_1, w_2; \vartheta_{w, w'})
		\coloneqq 
		\begin{pmatrix}
			p_1^0 - H_{\rho_x}(p_1^0, p_2^1), 
			\;\; p_2^0 - H_{\rho_x}(p_1^1, p_2^0), 
			\;\; H_{\rho_x}(p_1^1, p_2^1) - \lambda_x \cdot \mathcal{L}_{\text{mul}}(w)
		\end{pmatrix}
		,
	\end{align*}
	where $\mathbf{p}' = (p_1^{0'}, p_1^{1'}, p_2^{0'}, p_2^{1'})$ with $p_j^{d'} = F_{\varepsilon_j | X = x} (\pi_j^d(w_j'))$.
	Further, let $\mathcal{G}(\vartheta_{w, w'}) \coloneqq [\mathcal{G}(w_1, w_2; \vartheta_{w, w'})$, $\mathcal{G}(w_1', w_2; \vartheta_{w, w'}), \mathcal{G}(w_1, w_2'; \vartheta_{w, w'}), \mathcal{G}(w_1', w_2'; \vartheta_{w, w'})]^\top$.
	Then, $\mathcal{G}(\vartheta_{w, w'})$ is a $12 \times 1$ vector whose value is knowable from data in view of \eqref{eq:moments}, while $\vartheta_{w, w'}$ contains 10 unknown parameters.
	Thus, we can locally identify $\vartheta_{w, w'}$ if and only if the rank of the Jacobian matrix $\partial \mathcal{G}(\vartheta_{w, w'}) / \partial \vartheta_{w, w'}$ is equal to 10 (Theorem 6, \citealpmain{rothenberg1971identification}).
	Moreover, we can achieve the global identification if we introduce additional conditions on the structure of the Jacobian matrix.\footnote{
		By Lemma 4.2 of \citetmain{han2017identification}, $\vartheta_{w, w'}$ is globally identified if there exists a $10 \times 1$ sub-vector $\mathcal{G}'$ of $\mathcal{G}$ such that (i) $\mathcal{G}'$ is proper, (ii) the Jacobian of $\mathcal{G}'$ vanishes nowhere, and (iii) the range of $\mathcal{G}'$ is simply connected.
	}
	In general, proving the full-rankness of the Jacobian matrix with easy-to-check conditions is difficult unless we introduce specific parametric specifications.
	As such an example, we consider a particular linear index form for the payoff function and the equilibrium selection in Appendix \ref{subsec:ident_wo_inf}, and demonstrate that the parameters in this game model can be identified globally without the large support condition.
	Another approach to circumvent the large support condition is to employ the strategy of \citetmain{kline2016empirical}, which is based on the unimodality of the error distribution and the linearity of the payoff function.
\end{remark}

Finally, we note that any counterfactual analysis requires that the structure of the game model, including the equilibrium selection mechanism, should remain invariant under the new environment.
Although this requirement sounds restrictive, it cannot be omitted in principle, except for a partial identification analysis. 
For a more specific discussion, see \citemain{jun2020counterfactual} and Appendix \ref{subsec:PRTE} of this paper.

\subsection{Identification of the second-stage parameters} \label{subsec:second}

Throughout this subsection, we denote the conditional joint CDF and PDF of $(V_1, V_2)$ given $X = x$ as $H(\cdot, \cdot | x)$ and $h(\cdot, \cdot | x)$, respectively.
Theorem \ref{thm:gameiden}(ii) ensures that these functions are identified via the copula function $H_{\rho_x}$.
Including these, all the first-stage components that are identifiable through Theorem \ref{thm:gameiden} are treated as observable information in the following analysis.

\begin{assumption}
	\hfil
	\label{as:IV}
	\begin{enumerate}[(i)]
		\item $Z$ is excluded from the structural functions in \eqref{eq:model-Y1} and is independent of the unobservables $(\varepsilon, U^{(d_1, d_2)})$ given $X$ for $(d_1, d_2) \in \{0, 1\}^2$.
		\item \phantomsection\label{page:AE-14-1}\Copy{AE-14-1}{
			For all $(d_1, d_1', d_2, d_2') \in \{ 0, 1 \}^4$ such that $d_1 \neq d_1'$ and $d_2 \neq d_2'$, $(P_1^{d_1}, P_2^{d_2})$ is continuously distributed conditional on $(X, P_1^{d_1'}, P_2^{d_2'})$.
		}
	\end{enumerate}
\end{assumption}

Assumption \ref{as:IV}(i) requires $Z$ not to directly affect $Y$ and to be conditionally independent of the error terms.
By construction, the transformed error $V$ is also conditionally independent of $Z$ given $X$.
\phantomsection\label{page:AE-15}\Copy{AE-15}{
	Assumption \ref{as:IV}(ii) can be read as the exclusion restriction assumption that is often required for nonparametric identification.
	That is, for this to hold in a nonparametric model, both player- and action-specific continuous IVs must exist.
	However, note that when considering a fully parametric game model with a linear index form (i.e., the one given in Assumption \ref{as:gamemodel}), the existence of action-specific IVs is not strictly required to maintain Assumption \ref{as:IV}(ii), and the assumption can be satisfied if some IVs are player-specific and continuously distributed.
}

\phantomsection\label{page:R2-1-1}\Copy{R2-1-1}{
	As an example of IV $Z$ satisfying these conditions, in the empirical application of this study, with $Y$ being students' academic performance and $D$ being their delinquency, we consider employing their non-best friends' parental characteristics.
	For another example, in the case of an airline market with two competitors, we may estimate the effect of introducing a direct flight ($D$) on the number of total passengers for each company ($Y$) by using cost variables as $Z$ (cf. \citealpmain{ciliberto2009market}).	
	As the last example, when we are interested in measuring how the choice of electoral campaign strategy ($D$) affects the candidate's vote share ($Y$) in a mayoral election, the amount spent in the previous election may be a good candidate for $Z$ (in a similar sense to \citealpmain{gerber1998estimating}).
	Then, once the MTE is obtained, we can estimate some interesting PRTEs; for example, what if there had been a cap on the election budget imposed by the government?
}

\bigskip

We provide below a series of identification results only for player 1 (the results for player 2 are symmetric and thus omitted). 
Define
\begin{align}\label{eq:m}
	m^{(d_1, d_2)}(x, p_1, p_2) 
	\coloneqq \bE [Y_1^{(d_1, d_2)} | X = x, V_1 = p_1, V_2 = p_2],
\end{align}
for a given $x \in \text{supp}[X]$ and $(p_1, p_2) \in [0, 1]^2$.
We call the function $m^{(d_1, d_2)}$ the marginal treatment response (MTR) function, as in \citetmain{mogstad2018using}.
From the viewpoint of player 1, the parameters of interest are the \textit{direct} MTE, \textit{indirect} MTE, and \textit{total} MTE:
\begin{align*}
	\text{MTE}_{\text{direct}}^{(d_2)}(x, p_1, p_2) &\coloneqq m^{(1, d_2)}(x, p_1, p_2) - m^{(0, d_2)}(x, p_1, p_2)   \;\; \text{ for $d_2 \in \{0,1\}$},\\
	\text{MTE}_{\text{indirect}}^{(d_1)}(x, p_1, p_2) &\coloneqq m^{(d_1, 1)}(x, p_1, p_2) - m^{(d_1, 0)}(x, p_1, p_2)    \;\; \text{ for $d_1 \in \{0,1\}$},\\
	\text{MTE}_{\text{total}}(x, p_1, p_2) &\coloneqq m^{(1, 1)}(x, p_1, p_2) - m^{(0, 0)}(x, p_1, p_2).
\end{align*}
Notably, these MTE parameters are informative about the variation of the treatment effects in terms of the pair of unobservables $V_1$ and $V_2$, in contrast to the conventional MTE framework.
The identification of the MTE parameters is straightforward once the MTR function for each $(d_1,d_2)$ is identified.

\subsubsection{Identification of the marginal treatment response functions}\label{subsec:condmean}

We first investigate the identification of the MTRs $m^{(1, 0)}(x, p_1, p_2)$ and $m^{(0, 1)}(x, p_1, p_2)$.
Since $D = (1, 0)$ and $D = (0, 1)$ are unique equilibria, the presence of multiple equilibria does not cause problems in these cases.
Indeed, we can achieve point identification of MTRs with a natural extension of the conventional LIV method without using the equilibrium selection rule in Assumption \ref{as:multiple}.\footnote{
	For this study, when it is stated that MTR $m^{(d_1,d_2)}(x, p_1, p_2)$ is ``(point) identified'', the value of the MTR at this specific $(p_1, p_2)$ is identified, but the whole functional form with respect to $(p_1, p_2)$ is not necessarily identified.
	When the latter is true, we say, for example, the MTR function $m^{(d_1,d_2)}(x, \cdot, \cdot)$ is ``fully'' (point) identified on $[0,1]^2$.
	The same expression applies to the identification of MTEs.
}

For the subsequent analysis, we introduce a standard overlap condition.

\begin{assumption}\label{as:overlap}
	\phantomsection\label{page:R1-1-1}\Copy{R1-1-1}{
		For all $(d_1, d_2) \in \{ 0, 1 \}^2$, $0 < \Pr[ D = (d_1, d_2) | W ] < 1$ a.s.
	}
\end{assumption}

\phantomsection\label{page:R1-1-2}\Copy{R1-1-2}{
	Define
	\begin{align*}
		\begin{array}{ll}
			\psi^{(1, 0)}(x, p_1, p_2) \coloneqq \bE [ I^{(1, 0)} Y_1| X = x, P_1^0 = p_1, P_2^1 = p_2],\\
			\psi^{(0, 1)}(x, p_1, p_2) \coloneqq \bE [ I^{(0, 1)} Y_1| X = x, P_1^1 = p_1, P_2^0 = p_2].
		\end{array}
	\end{align*}
	Under Assumption \ref{as:overlap}, these quantities are identified on $\mathcal{S}_x^{(1,0)} \coloneqq \text{supp}[P_1^0, P_2^1 | X = x]$ and $\mathcal{S}_x^{(0,1)} \coloneqq \text{supp}[P_1^1, P_2^0 | X = x]$, respectively.
	The next theorem demonstrates that the MTR functions $m^{(1,0)}(x, \cdot, \cdot)$ and $m^{(0,1)}(x, \cdot, \cdot)$ are also identified on these supports.
}
Hereinafter, we use the following differential operator notations: $\partial_{p_j} \coloneqq \partial / (\partial p_j)$ and $\partial_{p_1 p_2} \coloneqq \partial^2 / (\partial p_1 \partial p_2)$.

\begin{theorem}\label{thm:unique}
	Suppose that Assumptions \ref{as:complement}, \ref{as:copula}, \ref{as:infinity}, \ref{as:IV}, and \ref{as:overlap} hold.
	\phantomsection\label{page:R1-3-1}\Copy{R1-3-1}{
		Under these assumptions, $\mathbf{P} = (P_1^0, P_1^1, P_2^0, P_2^1)$ and $h(\cdot, \cdot | x)$ are identified from Theorem \ref{thm:gameiden}(i)--(ii).
	}
	If $m^{(1, 0)}(x, \cdot, \cdot)$, $m^{(0,1)}(x, \cdot, \cdot)$, and $h(\cdot, \cdot | x)$ are continuous, the MTRs are identified in the following way:
	\begin{align*}
		& m^{(1, 0)}(x, p_1, p_2) 
		= - \frac{\partial_{p_1 p_2} [\psi^{(1,0)}(x, p_1, p_2)]}{h(p_1, p_2 | x)} \quad \text{for} \;\; (p_1, p_2) \in \mathcal{S}_x^{(1,0)}, \\
		& m^{(0, 1)}(x, p_1, p_2)
		= - \frac{\partial_{p_1 p_2} [\psi^{(0,1)}(x, p_1, p_2)]}{h(p_1, p_2 | x)} \quad \text{for} \;\; (p_1, p_2) \in \mathcal{S}_x^{(0,1)}.
	\end{align*}
\end{theorem}

The main idea behind Theorem \ref{thm:unique} can be visually understood by Figure \ref{fig:nash-comp2}. 
We consider the pairs of players with $X = x$ and $\mathbf{P} = \mathbf{p}$, where $\mathbf{p} = (p_1^0, p_1^1, p_2^0, p_2^1)$.
Because $(V_1, V_2)$ are continuously distributed on $[0, 1]^2$, a certain proportion of the pairs are in the neighborhood of $(p_1^0, p_2^1)$, point A, at the margin of $D = (1, 0)$, $(0, 0)$, and $(1, 1)$.
Hence, for a small fluctuation in $(P_1^0, P_2^1)$ at this point, some pairs would switch their treatment status between $D = (1, 0)$ and $D = (0, 0)$ or $(1, 1)$.
As $\mathbf{P}$ is exogenous under Assumption \ref{as:IV}(i), such variation in the treatment status is generated exogenously, which in turn can be used to identify $m^{(1, 0)}(x, p_1^0, p_2^1)$.
Meanwhile, since fluctuations of $(P_1^1, P_2^1)$ around point B, that of $(P_1^0, P_2^0)$ around C, and of $(P_1^1, P_2^0)$ around D do not generate any shift of the treatment status from/to $D = (1, 0)$, they do not possess identification power for $m^{(1, 0)}(x, p_1^0, p_2^1)$.
The analogous argument applies to the identification of $m^{(0, 1)}(x, p_1^1, p_2^0)$ at point D.

\begin{figure}[!ht]
	\centering
	\fbox{\includegraphics[width = 8cm, bb = 0 0 720 540]{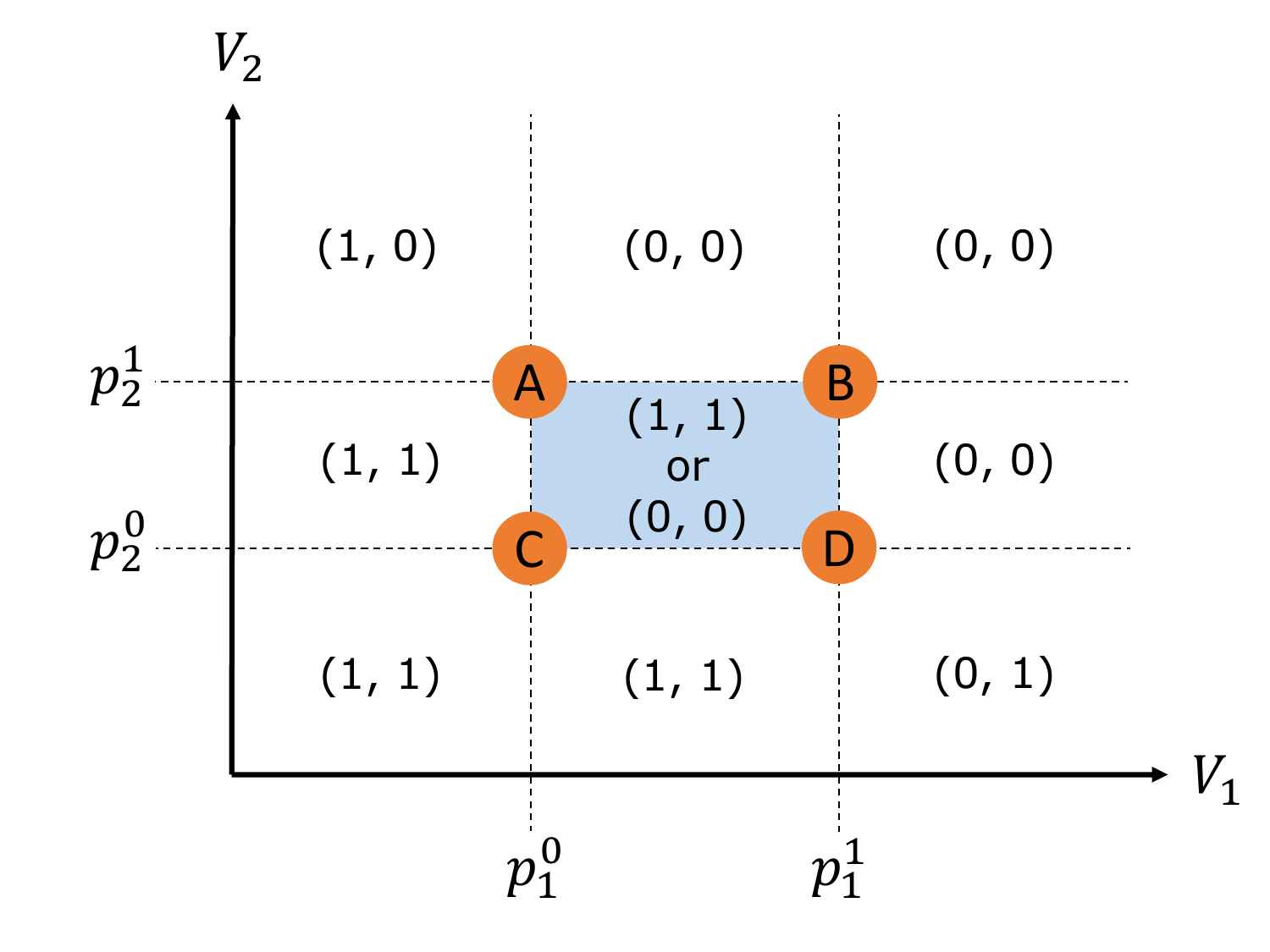}}
	\caption{Points at which the MTRs are identified.}
	\label{fig:nash-comp2} 
\end{figure}

Presenting a sketch of the proof for $m^{(1, 0)}(x, p_1, p_2)$ would be helpful.
It can be shown that
\begin{align} \label{eq:sketch}
\begin{split}
	\psi^{(1, 0)}(x, p_1, p_2)
	&= \bE [Y_1^{(1, 0)}| X = x, V_1 \le p_1, V_2 > p_2] \cdot \Pr[V_1 \le p_1, V_2 > p_2 | X = x]\\
	&= \int_{p_2}^1 \int_0^{p_1} m^{(1, 0)}(x, v_1, v_2) h(v_1, v_2 | x) \mathrm{d}v_1 \mathrm{d}v_2.
\end{split}
\end{align}
The first line follows from Assumption \ref{as:IV}(i) and the fact that the treatment status $D = (1, 0)$ is uniquely linked with the set of threshold crossing rules $V_1 \le P_1^0$ and $V_2 > P_2^1$.
Since $(P_1^0, P_2^1)$ are jointly continuously distributed under Assumption \ref{as:IV}(ii), we can take the partial derivatives of both sides with respect to $p_1$ and $p_2$, which yields the desired result.
Note that Assumptions \ref{as:copula} and \ref{as:infinity} are not used to show \eqref{eq:sketch}, but they are used to recover $\mathbf{P}$ and $h$.

In view of \eqref{eq:sketch}, we can see that Assumption \ref{as:IV}(ii) is in fact stronger than necessary for Theorem \ref{thm:unique}.
Indeed, the partial differentiation of \eqref{eq:sketch} with respect to $(p_1, p_2)$ is well-defined as long as only $(P_1^0, P_2^1)$ are continuously distributed conditional on $X = x$.

\bigskip

We move on to the identification of $m^{(0, 0)}(x, p_1, p_2)$ and $m^{(1, 1)}(x, p_1, p_2)$.
Due to the presence of multiple equilibria, the treatment statuses $D = (0, 0)$ and $D = (1, 1)$ are not uniquely determined by the threshold crossing rules regarding $V$ only.
Consequently, in general, we cannot achieve point identification of the MTR for $(d_1, d_2) \in  \{ (0, 0), (1, 1) \}$ with no assumptions on the equilibrium selection.
In the following, we demonstrate that Assumption \ref{as:multiple} is sufficient to overcome this issue, and enables us to utilize the same identification strategy as in Theorem \ref{thm:unique}.

\phantomsection\label{page:R1-1-3}\Copy{R1-1-3}{
	To state the next theorem, define
	\begin{align*}
		\psi^{(d_1, d_2)} (x, \mathbf{p})
		\coloneqq \bE [ I^{(d_1, d_2)} Y_1 | X = x, \mathbf{P} = \mathbf{p}],
	\end{align*}
	which is identified on $\text{supp}[\mathbf{P} | X = x]$.
}

\begin{theorem}\label{thm:multiple}
	Suppose that Assumptions \ref{as:complement} and \ref{as:copula}--\ref{as:overlap} hold.
	\phantomsection\label{page:R1-3-2}\Copy{R1-3-2}{
		Under these assumptions, $\mathbf{P} = (P_1^0, P_1^1, P_2^0, P_2^1)$, $h(\cdot, \cdot | x)$, and $\lambda_x$ are identified from Theorem \ref{thm:gameiden}(i)--(iii).
	}
	If $m^{(0, 0)}(x, \cdot, \cdot)$, $m^{(1, 1)}(x, \cdot, \cdot)$, and $h(\cdot, \cdot | x)$ are continuous and $0 < \lambda_x < 1$, the MTRs are identified in the following way: for $\mathbf{p} = (p_1^0, p_1^1, p_2^0, p_2^1) \in \text{supp}[\mathbf{P} | X = x]$,
	\begin{alignat*}{3}
		m^{(0, 0)}(x, p_1^0, p_2^0) 
		& = \frac{\partial_{p_1^0 p_2^0} [\psi^{(0, 0)}(x, \mathbf{p})]}{\lambda_x h(p_1^0, p_2^0 | x)}, 
		&\quad
		m^{(0, 0)}(x, p_1^1, p_2^1) 
		& = - \frac{\partial_{p_1^1 p_2^1} [\psi^{(0, 0)}(x, \mathbf{p})]}{(1 -\lambda_x) h(p_1^1, p_2^1|x)},\\ 
		m^{(0, 0)}(x, p_1^1, p_2^0) 
		& = \frac{\partial_{p_1^1 p_2^0} [\psi^{(0, 0)}(x, \mathbf{p})]}{(1 - \lambda_x) h(p_1^1, p_2^0|x)},
		&\quad
		m^{(0, 0)}(x, p_1^0, p_2^1) 
		& = \frac{\partial_{p_1^0 p_2^1} [\psi^{(0, 0)}(x, \mathbf{p})]}{(1 - \lambda_x) h(p_1^0, p_2^1|x)},
	\end{alignat*}
	and
	\begin{alignat*}{3}
		m^{(1, 1)}(x, p_1^0, p_2^0) 
		& = - \frac{\partial_{p_1^0 p_2^0} [\psi^{(1, 1)}(x, \mathbf{p})]}{\lambda_x h(p_1^0, p_2^0 | x)},
		&\quad
		m^{(1, 1)}(x, p_1^1, p_2^1) 
		& = \frac{\partial_{p_1^1 p_2^1} [\psi^{(1, 1)}(x, \mathbf{p})]}{(1 - \lambda_x) h(p_1^1, p_2^1|x)},\\
		m^{(1, 1)}(x, p_1^1, p_2^0) 
		& = \frac{\partial_{p_1^1 p_2^0} [\psi^{(1, 1)}(x, \mathbf{p})]}{\lambda_x h(p_1^1, p_2^0|x)},
		&\quad
		m^{(1, 1)}(x, p_1^0, p_2^1) 
		& = \frac{\partial_{p_1^0 p_2^1} [\psi^{(1, 1)}(x, \mathbf{p})]}{\lambda_x h(p_1^0, p_2^1|x)}.
	\end{alignat*}
\end{theorem}

The proof of Theorem \ref{thm:multiple} is straightforward from the following fact:
\begin{align}\label{eq:psi_00}
	\begin{split}
\psi^{(0, 0)} (x, \mathbf{p} )
	&= \int_{p_2^0}^1 \int_{p_1^0}^1 m^{(0, 0)}(x, v_1, v_2) h(v_1, v_2 | x) \mathrm{d}v_1 \mathrm{d}v_2 \\
	& \quad  - (1 - \lambda_x) \int_{p_2^0}^{p_2^1} \int_{p_1^0}^{p_1^1}  m^{(0, 0)}(x, v_1, v_2) h(v_1, v_2 | x) \mathrm{d}v_1 \mathrm{d}v_2.
	\end{split}
\end{align}
A similar representation can be obtained for $\psi^{(1, 1)} (x, \mathbf{p} )$.
Again, these identification results can be visually understood in Figure \ref{fig:nash-comp2}.
For example, a small fluctuation of $(P_1^0, P_2^0)$ around point C generates an exogenous shock to change the treatment status of a certain portion of pairs between $D = (0, 0)$ and $D = (1,1)$, provided that $\lambda_x > 0$.
This local exogenous variation in the treatment choice can be used to recover $m^{(0, 0)}(x, p_1^0, p_2^0)$.
Note that, if $\lambda_x = 0$, the identification at point C fails because $D = (0, 0)$ is never chosen at this point.
Similarly, if $\lambda_x < 1$, a variation in $(P_1^1, P_2^1)$ in the neighborhood of point B also causes an exogenous treatment shift between $D = (0, 0)$ and $D = (1, 1)$, which enables us to identify $m^{(0, 0)}(x, p_1^1, p_2^1)$.

Theorems \ref{thm:unique} and \ref{thm:multiple} differ in two important respects.
First, Assumption \ref{as:IV}(ii) is necessary for Theorem \ref{thm:multiple}.
This point should be obvious because the partial differentiation with respect to $(p_1^{d_1}, p_2^{d_2})$ with $(p_1^{d_1'}, p_2^{d_2'})$ being fixed ($d_1 \neq d_1'$, $d_2 \neq d_2'$) is not well-defined without Assumption \ref{as:IV}(ii).
Second, recall that Theorem \ref{thm:unique} requires only two out of $(P_1^0, P_1^1, P_2^0, P_2^1)$ as the conditioning variables because the point at which identification is achieved is well characterized as the upper-left or lower-right corner in the space of $V$.
By contrast, to characterize the multiple equilibria region, we should fix the values of all four items $(P_1^0, P_1^1, P_2^0, P_2^1)$, of which two are marginal probabilities of each player's counterfactual treatment.
This is the main reason why we need to identify all of the payoff functions, the equilibrium selection probability, and the correlation parameter in the first-stage game (as kindly pointed out by a referee).

\phantomsection\label{page:R1-1-4}\Copy{R1-1-4}{
	The results in Theorem \ref{thm:multiple} enable us to recover the MTR $m^{(0, 0)}(x, \cdot, \cdot)$ on the union of the supports of $(P_1^{d_1}, P_2^{d_2})$ for $(d_1, d_2) \in \{ 0, 1 \}^2$: 
	\begin{align*}
		\overline{\mathcal{S}}_x 
		\coloneqq \bigcup_{(d_1, d_2) \in \{ 0, 1 \}^2} \text{supp}[P_1^{d_1}, P_2^{d_2} | X = x].
	\end{align*}
	To observe this, note that, for arbitrary $(p_1^0, p_2^0) \in \text{supp}[P_1^0, P_2^0 | X = x]$, some $p_1^1$ and $p_2^1$ exist such that $(p_1^0, p_1^1, p_2^0, p_2^1) \in \text{supp}[\mathbf{P} | X = x]$.
	Therefore, $m^{(0, 0)}(x, p_1^0, p_2^0)$ is identified for any $(p_1^0, p_2^0) \in \text{supp}[P_1^0, P_2^0 | X = x]$ as in Theorem \ref{thm:multiple}.
	The same argument holds for the other cases.
	Consequently, $m^{(0, 0)}(x, p_1, p_2)$ is identifiable for any $(p_1, p_2) \in \overline{\mathcal{S}}_x$.
	Similarly, $m^{(1, 1)}(x, \cdot, \cdot)$ is identified on $\overline{\mathcal{S}}_x$.
}

\begin{remark}[Identifiable regions] \label{remark:region}
	\phantomsection\label{page:R1-3-3}\Copy{R1-3-3}{
		The identification regions in Theorems \ref{thm:unique} and \ref{thm:multiple}, namely, $\mathcal{S}_x^{(1,0)}$, $\mathcal{S}_x^{(0,1)}$, and $\overline{\mathcal{S}}_x$, are characterized indirectly through the parameters in the first-stage treatment decision game as $P_j^d = F_{\varepsilon_j | X = x}(\pi_j^d(W_j))$, not directly from the standard propensity score, say $\text{PS}_j \coloneqq \bE[D_j | W_j]$.
		In the canonical single agent setting, the MTR parameter for each agent $j$ can be identified on the conditional support of $\text{PS}_j$ (cf. \citealpmain{heckman1999local, heckman2005structural}).
		Even in the presence of treatment spillovers within each pair of players, if there is no strategic interaction in the treatment decision, it is straightforward to show that the MTR parameter can be identified on the joint support of $(\text{PS}_1, \text{PS}_2)$ (see Appendix \ref{subsec:nointeraction} for more details).
		Once a strategic interaction comes into play, such simple identification arguments no longer hold.
	}
\end{remark}

\begin{remark}[Role of the large support condition]\label{remark:large}
	\phantomsection\label{page:R1-5-1}\Copy{R1-5-1}{
		In Theorems \ref{thm:unique} and \ref{thm:multiple}, the large support condition in Assumption \ref{as:infinity}(i) is used only for recovering the parameters in the treatment decision game.
		If the first-stage parameters can be identified under alternative assumptions as in Remark \ref{remark:infinity}, the identification of the MTR at an interior point of the support of $(P_1^{d_1}, P_2^{d_2})$ can be established without the large support condition.
		However, note that, if the goal is to identify other treatment parameters, such as the LATE and PRTE, we generally need the large support condition to recover the MTR fully on $[0,1]^2$.
		See Appendix \ref{sec:severalparameter} for the identification of these parameters.
	}
\end{remark}

\begin{remark}[Imposing a parametric restriction on MTR]\label{remark:brinch}
	\phantomsection\label{page:R2-3}\Copy{R2-3}{
		Theorems \ref{thm:unique} and \ref{thm:multiple} are uninformative about the value of the MTRs for $(p_1, p_2)$ outside the support of $(P_1^{d_1}, P_2^{d_2})$.
		This is due to the nonparametric nature of the identification strategy.
		Alternatively, if we impose an explicit functional form on $m^{(d_1,d_2)}(x, \cdot, \cdot)$, such as polynomials, there is a possibility to identify $m^{(d_1,d_2)}(x, \cdot, \cdot)$ by interpolating from the other identified MTR values.
		This approach is analogous to the identification strategy discussed in \citetmain{brinch2017beyond} in the conventional MTE framework.
		If we adopt such an approach, even when the equilibrium selection probability $\lambda_x$ also depends on $Z = z$, we would be able to identify $m^{(0, 0)}(x, \cdot, \cdot)$ by directly solving the integral equation \eqref{eq:psi_00}.
		The same is true for $m^{(1, 1)}(x, \cdot, \cdot)$.
		In this sense, for identification of $m^{(0, 0)}(x, \cdot, \cdot)$ and $m^{(1, 1)}(x, \cdot, \cdot)$, there is a trade-off between their nonparametric identifiability and the exclusion restriction on the equilibrium selection.
	}
\end{remark}

\begin{remark}[\citemain{lee2018identifying}]\label{remark:leesalanie}
	A major distinction between \citemain{lee2018identifying} and our study is that they allow the strategic effect to depend on unobservable factors by assuming that the error term $V_j$ has the form of $V_j = V_j^0 + D_{-j} (V_j^1 - V_j^0)$, where $V_j^0$ denotes the unobserved payoff determinant when only player $j$ is treated, and $V_j^1 - V_j^0$ corresponds to the unobservable strategic interaction effect.
	To deal with the multiple equilibria, similarly to us, they assume that one of the equilibria is selected if an unobserved variable $\bar V$ does not exceed a threshold $\bar P$.
	All possible treatment assignments are then characterized uniquely by the set of threshold crossing rules determined by $\mathbf{V} = (V_1^0, V_1^1, V_2^0, V_2^1, \bar V)$ and the threshold variables $\mathbf{P} = (P_1^0, P_1^1, P_2^0, P_2^1, \bar P)$.
	Then, \citemain{lee2018identifying} showed that
	\begin{align*}
		\bE [Y_1^{(d_1, d_2)} | X=x, \mathbf{V} = \mathbf{p}]
		= \frac{\partial_{\mathbf{p}} \bE [I^{(d_1, d_2)} Y_1 | X = x , \mathbf{P} = \mathbf{p}]}{\partial_{\mathbf{p}} \Pr[D = (d_1, d_2) | X=x, \mathbf{P} = \mathbf{p}]},
	\end{align*}
	where $\partial_{\mathbf{p}} \coloneqq \partial^5 / (\partial p_1^0 \partial p_1^1 \partial p_2^0 \partial p_2^1 \partial \bar p)$.
	This MTR parameter includes richer information related to unobserved heterogeneity than ours in that it can reveal unobserved heterogeneity with respect to the unobservable interaction effect, preference on equilibrium selection, and payoff determinants.
	To estimate this MTR, $\mathbf{P}$ must be jointly continuously distributed, which is more demanding than Assumption \ref{as:IV}(ii).
	In addition, the estimator resulting from this identification result would suffer from the curse of dimensionality, whereas our estimator does not at the cost of less-flexible interaction structure.
	Indeed, \citemain{lee2018identifying} is concerned mainly with the identification of MTR and not its estimation, whereas we focus on both.
\end{remark}

\subsubsection{Identification of the marginal treatment effects} \label{subsec:MTE}

\phantomsection\label{page:R1-2-1}\Copy{R1-2-1}{
	Given Theorems \ref{thm:unique} and \ref{thm:multiple}, identification of MTE is straightforward. 
	For example, $\text{MTE}_{\text{direct}}^{(0)}(x, \cdot, \cdot) = m^{(1,0)}(x, \cdot, \cdot) - m^{(0, 0)}(x, \cdot, \cdot)$ is identifiable on the intersection of $\mathcal{S}_x^{(1,0)}$ and $\overline{\mathcal{S}}_x$ because $m^{(1, 0)}(x, \cdot, \cdot)$ and $m^{(0, 0)}(x, \cdot, \cdot)$ are identified on the respective regions, meaning that the identifiable region of $\text{MTE}_{\text{direct}}^{(0)}(x, \cdot, \cdot)$ reduces to $\mathcal{S}_x^{(1,0)}$ as $\mathcal{S}_x^{(1,0)} \subseteq \overline{\mathcal{S}}_x$.
	As another example, the identifiable region of $\text{MTE}_{\text{total}}(x, \cdot, \cdot) = m^{(1, 1)}(x, \cdot, \cdot) - m^{(0, 0)}(x, \cdot, \cdot)$ corresponds to $\overline{\mathcal{S}}_x$, on which both MTR functions are identified.
	We can characterize the identifiable regions of the other MTE parameters analogously.
}

\phantomsection\label{page:R1-1-5}\Copy{R1-1-5}{
	As a result of Theorem \ref{thm:multiple}, $m^{(d_1, d_2)}(x, p_1, p_2)$ for $(d_1, d_2) \in \{ (0, 0), (1, 1) \}$ can be ``over-identified'' if $(p_1, p_2) \in \underline{\mathcal{S}}_x$, where
	\begin{align*}
		\underline{\mathcal{S}}_x 
		\coloneqq \bigcap_{(d_1, d_2) \in \{ 0, 1 \}^2} \text{supp}[P_1^{d_1}, P_2^{d_2} | X = x].
	\end{align*}
	Thus, the MTE parameters are also over-identifiable if $\underline{\mathcal{S}}_x$ is non-empty, and this over-identification result can be used to improve the estimation efficiency (see Subsection \ref{subsec:overidentification:est}).
}

\begin{remark}[Partial identification of MTE without equilibrium selection]\label{remark:partial}
	As shown above, the MTR functions for $(d_1, d_2) \in \{ (1, 0), (0, 1) \}$ can be identified without using the equilibrium selection assumption.
	When no assumptions are imposed on the equilibrium selection, the equilibrium selection probability can take any value on $[0,1]$.
	Thus, as can be seen in Theorem \ref{thm:multiple}, the identified sets of the MTRs for $(d_1, d_2) \in \{ (0, 0), (1, 1) \}$ are typically unbounded.
	However, because the equilibrium selection probability is non-negative, we can still identify the signs of these MTR functions.
	Therefore, for example, if $m^{(1,0)}(x, p_1, p_2)$ has a positive value and the sign of $m^{(0,0)}(x, p_1, p_2)$ is non-positive, we can obtain an informative lower bound of $\text{MTE}_{\text{direct}}^{(0)}(x, p_1, p_2)$ based on the inequality $\text{MTE}_{\text{direct}}^{(0)}(x, p_1, p_2) \ge m^{(1,0)}(x, p_1, p_2)$.
	A more promising approach would be to introduce some shape restrictions, as in \citemain{balat2020multiple}, to derive informative bounds on $m^{(0,0)}(x, p_1, p_2)$.
\end{remark}

\section{Estimation and Asymptotics} \label{sec:estimation}

In this section, we propose a two-step semiparametric procedure for estimating the MTE parameters given the data $\{\{(Y_{ji}, D_{ji}, W_{ji})\}_{j=1}^2\}_{i=1}^n$ are observed.
In the subsequent analysis, we consider the following parametric treatment decision model:
\begin{assumption}
	\hfil
	\label{as:gamemodel}
	\begin{enumerate}[(i)]
		\item \phantomsection\label{page:R2-4-1}\Copy{R2-4-1}{
			For each $j = 1, 2$, $D_{ji} = \mathbf{1}\{ W_{ji}^\top \gamma_0 + D_{-j, i} \cdot \Delta(W_{ji}^\top \gamma_1) \ge \varepsilon_{ji} \}$, where $\gamma = (\gamma_0^\top, \gamma_1^\top)^\top \in \mathbb{R}^{2\mathrm{dim}(W)}$ is a vector of unknown parameters such that $\gamma_0 \neq \gamma_1$, and $\Delta$ is a known positive function.
		}
		\item $(\varepsilon_1, \varepsilon_2)$ are independent of $W$ and continuously distributed with known strictly increasing marginal CDFs $F_{\varepsilon_1}$ and $F_{\varepsilon_2}$, respectively.
		Their joint distribution is given by $\Pr[\varepsilon_1 \le a_1, \varepsilon_2 \le a_2] = H_\rho(F_{\varepsilon_1}(a_1), F_{\varepsilon_2}(a_2))$, where $\rho \in (\underline{c}, \bar c)$ is an unknown parameter.
		The copula $H_\rho$ has a density function $h_\rho$.
		\item Assumption \ref{as:multiple} holds with $\lambda_x = \Lambda(\tilde x^\top \lambda)$, where $\Lambda$ is a known strictly increasing CDF, $\tilde X$ is a linearly independent subset of $X$, and $\lambda \in \mathbb{R}^{\mathrm{dim}(\tilde X)}$ is a vector of unknown parameters.
	\end{enumerate}
\end{assumption}

In Assumption \ref{as:gamemodel}(i), we assume $\Delta > 0$ to ensure strategic complementarity. 
We also assume that the coefficients $\gamma$ are common to both players.
\phantomsection\label{page:R2-4-2}\Copy{R2-4-2}{
	Note that the assumption $\gamma_0 \neq \gamma_1$ is necessary to identify the interaction effect.
	If this does not hold, the value of $P_j^1$ is automatically determined once $P_j^0$ is fixed, violating Assumption \ref{as:IV}(ii).
}
Assumption \ref{as:gamemodel}(ii) strengthens Assumptions \ref{as:complement}(ii) and \ref{as:copula}(i) by requiring a known marginal CDF of $\varepsilon_j$ and full independence between $\varepsilon$ and $W$.
Assuming a known distribution for the unobservable is common in the literature, as nonparametrically identifying the payoff function and the error distribution simultaneously is generally impossible (cf. Section 4.1 of \citealpmain{bajari2010identification}).
Furthermore, as shown by \citetmain{khan2018information}, when the marginal CDFs of the errors are unknown, the interaction effect cannot be estimated at the parametric rate in general.
Estimating the game parameters at the parametric rate is important for the estimation of the MTE parameters.
Assumption \ref{as:gamemodel}(iii) introduces a parametric specification on the equilibrium selection probability.
Identification of this parametric treatment decision model is discussed in detail in Appendix \ref{subsec:ident_wo_inf}.

For the potential outcome, we assume the following linear model.

\begin{assumption}
	\hfil
	\label{as:linear}
	\begin{enumerate}[(i)]
		\item \phantomsection\label{page:AE-9-2}\Copy{AE-9-2}{
			For each $j = 1, 2$ and $(d_j, d_{-j}) \in \{ 0, 1 \}^2$, $Y_{ji}^{(d_j, d_{-j})} = X_{ji}^\top \beta_j^{(d_j, d_{-j})} + U_{ji}^{(d_j, d_{-j})}$, where $\beta_j^{(d_j, d_{-j})} \in \mathbb{R}^{\mathrm{dim}(X)}$ is a vector of unknown parameters and $U_{ji}^{(d_j, d_{-j})} \in \mathbb{R}$ is a scalar error term.
		}
		\item $(\varepsilon, U^{(d_1,d_2)})$ are independent of $W$ for $(d_1, d_2) \in \{0, 1\}^2$.
	\end{enumerate}
\end{assumption}

\subsection{Two-step estimation} \label{subsec:procedure}

\paragraph{First step: Estimation of the treatment decision game.}

In accordance with \eqref{eq:definition-P}, we write $V_{ji} = F_{\varepsilon_j}(\varepsilon_{ji})$, $P_{ji}^0(\gamma) = F_{\varepsilon_j}(W_{ji}^\top \gamma_0)$, and $P_{ji}^1(\gamma) = F_{\varepsilon_j}(W_{ji}^\top \gamma_0 + \Delta(W_{ji}^\top \gamma_1))$.
Further, let $\theta^* = (\gamma^{* \top}, \lambda^{* \top}, \rho^*)^\top$ be the true value of $\theta = (\gamma^\top, \lambda^\top, \rho)^\top$.
For a given $\theta$, the conditional probability that the $i$-th pair of players is in the multiple equilibria region is given by 
\begin{align*}
	\mathcal{L}_{\text{mul},i}(\theta) \coloneqq H_\rho(P_{1i}^1(\gamma), P_{2i}^1(\gamma)) - H_\rho(P_{1i}^1(\gamma), P_{2i}^0(\gamma)) - H_\rho(P_{1i}^0(\gamma), P_{2i}^1(\gamma)) + H_\rho(P_{1i}^0(\gamma), P_{2i}^0(\gamma)).
\end{align*}
Then, letting $\mathcal{L}_i^{(d_1, d_2)}(\theta)$ be the probability that they choose an action $D_i = (d_1, d_2)$, we have
\begin{equation*}
\begin{array}{ll}
	\mathcal{L}_i^{(1, 0)}(\theta) = P_{1i}^0(\gamma) - H_\rho(P_{1i}^0(\gamma), P_{2i}^1(\gamma)), &\quad \mathcal{L}_i^{(1, 1)}(\theta) = H_\rho(P_{1i}^1(\gamma),P_{2i}^1(\gamma)) - \Lambda(\tilde X_i^\top \lambda) \cdot \mathcal{L}_{\text{mul},i}(\theta),\\
	\mathcal{L}_i^{(0, 1)}(\theta) = P_{2i}^0(\gamma) - H_\rho(P_{1i}^1(\gamma), P_{2i}^0(\gamma)), &\quad \mathcal{L}_i^{(0, 0)}(\theta) = 1 - \sum_{(d_1,d_2) \neq (0,0)}\mathcal{L}_i^{(d_1, d_2)}(\theta).
\end{array}
\end{equation*}
The ML estimator $\hat \theta$ of $\theta^*$ is obtained by maximizing the log-likelihood $\sum_{i}\sum_{d_1, d_2} I_i^{(d_1, d_2)} \log \mathcal{L}_i^{(d_1, d_2)}(\theta)$ with respect to $\theta$.

Given this $\hat \theta$, we can estimate $P_{ji}^0 = P_{ji}^0(\gamma^*)$ and $P_{ji}^1 = P_{ji}^1(\gamma^*)$ by $\hat P_{ji}^0 = P_{ji}^0(\hat \gamma)$ and $\hat P_{ji}^1 = P_{ji}^1(\hat \gamma)$, respectively.
We denote the estimator of $\mathbf{P} = (P_1^0, P_1^1, P_2^0, P_2^1)$ by $\hat{\mathbf{P}} = (\hat P_1^0, \hat P_1^1, \hat P_2^0, \hat P_2^1)$.
Similarly, the true equilibrium selection probability $\lambda_x^* = \Lambda(\tilde x^\top \lambda^*)$ and the true joint CDF and density of $(V_1, V_2)$, which we denote by $H = H_{\rho^*}$ and $h = h_{\rho^*}$, can be respectively estimated by $\hat \lambda_x = \Lambda(\tilde x^\top \hat \lambda)$, $\hat H = H_{ \hat \rho}$, and $\hat h = h_{ \hat \rho}$.
Moreover, we can estimate the true choice probability $\mathcal{L}_i^{(d_1, d_2)} = \mathcal{L}_i^{(d_1, d_2)}(\theta^*)$ by $\hat{\mathcal{L}}_i^{(d_1, d_2)} = \mathcal{L}_i^{(d_1, d_2)}(\hat \theta)$.

\begin{remark}[Constant equilibrium selection probability] \label{remark:constant}
	\phantomsection\label{page:R2-2-2}\Copy{R2-2-2}{
		Since the equilibrium selection probability is estimated using only the observations in multiple equilibria, the estimation of $\lambda_x$ is very challenging in practice when $\text{dim}(\tilde X)$ is not small or the strategic effects are so weak that only a small number of observations are in the multiple equilibria region.
		Indeed, when we tried to estimate the function $\lambda_x$ in our empirical application, we could not obtain meaningful estimate.
		To circumvent this issue in practical situations with moderate sample size, we suggest simplifying Assumption \ref{as:gamemodel}(iii) such that $\lambda_x = \lambda$ for a scalar parameter $\lambda \in [0, 1]$.
	}
\end{remark}

\paragraph{Second step: Estimation of the MTE.}

As discussed in Subsection \ref{subsec:second}, a variety of treatment effect parameters can be identified.
We here specifically discuss the estimation of $\text{MTE}_{\text{direct}}^{(0)}(x, p_1^0, p_2^1)$ at point A in Figure \ref{fig:nash-comp2} (the estimation at the other points is analogous).
The estimation of the total MTE is relegated to Appendix \ref{subsec:totalMTE}.

Recall that $\text{MTE}_{\text{direct}}^{(0)}(x, p_1^0, p_2^1) = m^{(1, 0)}(x, p_1^0, p_2^1) - m^{(0, 0)}(x, p_1^0, p_2^1)$.
First, we discuss how to estimate $m^{(1, 0)}(x, p_1^0, p_2^1)$.
Assumption \ref{as:linear} implies that
\begin{align*}
	m^{(1, 0)}(x, p_1^0, p_2^1)
	= x_1^\top \beta_1^{(1, 0)} + \bE [U_1^{(1, 0)} | V_1 = p_1^0, V_2 = p_2^1] 
	= x_1^\top \beta_1^{(1, 0)} - \frac{\partial_{p_1^0 p_2^1} [g^{(1, 0)}(p_1^0, p_2^1)]}{h(p_1^0, p_2^1)},
\end{align*}
where $g^{(1, 0)}(p_1^0, p_2^1) \coloneqq \int_{0}^{p_1^0} \int_{p_2^1}^1 \bE [U_{1}^{(1, 0)}| V_1= v_1, V_2= v_2] h(v_1, v_2) \mathrm{d}v_1 \mathrm{d}v_2$.
Further, observe that
\begin{align*}
	\bE [ I^{(1, 0)} Y_1 | I^{(1, 0)}, X, P_1^0, P_2^1 ] 
	& = I^{(1, 0)} X_{1}^\top \beta_1^{(1, 0)} + I^{(1, 0)} \bE [ U_{1}^{(1, 0)}| V_1\le P_1^0, V_2> P_2^1].
\end{align*}
A similar argument to Theorem \ref{thm:unique} gives $\bE [ U_{1}^{(1, 0)}| V_1\le p_1^0, V_2> p_2^1] = g^{(1, 0)}(p_1^0, p_2^1) / \mathcal{L}^{(1,0)}(p_1^0, p_2^1)$, where $\mathcal{L}^{(1, 0)}(p_1^0, p_2^1) = p_1^0 - H(p_1^0, p_2^1)$.
Then, we have the following partially linear regression model:
\begin{align}\label{eq:plm1}
	I^{(1, 0)} Y_1 = I^{(1, 0)} X_{1}^\top \beta_1^{(1, 0)} + T^{(1, 0)} g^{(1, 0)}(P_1^0, P_2^1) + e^{(1, 0)},
\end{align}
where $T^{(1, 0)} \coloneqq I^{(1, 0)} / \mathcal{L}^{(1, 0)}$, and $\bE [e^{(1, 0)} | I^{(1, 0)}, X, P_1^0, P_2^1] = 0$ by construction.
Here, note that since $\mathcal{L}^{(1, 0)}$ may take an arbitrarily small value close to zero, the weight term $T^{(1, 0)}$ can be extremely large for some observations.
To deal with this issue, we exclude such observations from the analysis.
Specifically, we introduce a non-negative smoothed indicator function $\tau_\varpi(\mathcal{L}^{(1, 0)})$  such that (i) $\tau_\varpi:[0, 1] \to [0, 1]$ is a non-decreasing function and $\tau_\varpi(a) = 0$ if and only if $a < \varpi$ for a small constant $\varpi > 0$, and (ii) $\tau_\varpi$ is continuously differentiable with a bounded derivative.
Multiplying both sides of \eqref{eq:plm1} by $\tau_\varpi(\mathcal{L}^{(1, 0)})$ gives
\begin{align}\label{eq:plm2}
	\tilde I^{(1, 0)} Y_1 = \tilde I^{(1, 0)} X_{1}^\top \beta_1^{(1, 0)} + \tilde T^{(1, 0)} g^{(1, 0)}(P_1^0, P_2^1) + \tilde e^{(1, 0)},
\end{align}
where $\tilde I^{(1, 0)} \coloneqq \tau_\varpi(\mathcal{L}^{(1, 0)}) I^{(1, 0)}$, $\tilde T^{(1, 0)} \coloneqq \tilde I^{(1, 0)} / \mathcal{L}^{(1,0)}$, and $\tilde e^{(1, 0)} \coloneqq \tau_\varpi(\mathcal{L}^{(1, 0)}) e^{(1, 0)}$.
Since $\mathcal{L}^{(1, 0)}$ is a function of $P_1^0$ and $P_2^1$, $\bE [\tilde e^{(1, 0)} | I^{(1, 0)}, X, P_1^0, P_2^1] = 0$ still holds.

Based on \eqref{eq:plm2}, we consider estimating $\beta_1^{(1, 0)}$ and $g^{(1,0)}$ using the series (sieve) method.
Let $b_K(\cdot, \cdot) = (b_{1K}(\cdot, \cdot), \dots, b_{KK}(\cdot, \cdot))^\top$ be a $K \times 1$ vector of bivariate basis functions.
We assume that $g^{(1, 0)}$ can be well-approximated by $g^{(1, 0)}(\cdot, \cdot) \approx b_K(\cdot, \cdot)^\top \alpha^{(1,0)}$ for some coefficient vector $\alpha^{(1, 0)}$ with sufficiently large $K$.
Then, letting $\hat I_i^{(1, 0)} \coloneqq \tau_\varpi(\hat{\mathcal{L}}_i^{(1, 0)}) I_i^{(1, 0)}$ and $\hat T_i^{(1, 0)} \coloneqq \hat I_i^{(1, 0)} / \hat{\mathcal{L}}_i^{(1,0)}$, we have
\begin{align*}
	\hat I_i^{(1, 0)} Y_{1i} \approx \hat I_i^{(1, 0)}  X_{1i}^\top \beta_1^{(1, 0)} + \hat T_i^{(1, 0)} b_K(\hat P_{1i}^0, \hat P_{2i}^1)^\top \alpha^{(1,0)} + \tilde e_i^{(1, 0)}.
\end{align*}
Let $(\hat \beta_1^{(1,0)}, \hat \alpha^{(1,0)})$ be the least squares (LS) estimator of $(\beta_1^{(1,0)}, \alpha^{(1,0)})$ obtained by regressing $\hat I_i^{(1, 0)} Y_{1i}$ on $(\hat I_i^{(1, 0)}  X_{1i}, \hat T_i^{(1, 0)} b_K(\hat P_{1i}^0, \hat P_{2i}^1) )$.
Then, the estimator of $g^{(1, 0)}(p_1^0, p_2^1)$ can be obtained by $\hat g^{(1, 0)}(p_1^0, p_2^1) \coloneqq b_K(p_1^0, p_2^1)^\top \hat \alpha^{(1, 0)}$, and we can estimate $\bE [U_1^{(1, 0)} | V_1 = p_1^0, V_2 = p_2^1]$ by
\begin{align*}
\begin{split}
	\hat \bE_n[U_1^{(1, 0)} | V_1 = p_1^0, V_2 = p_2^1]
	& \coloneqq - \frac{\partial_{p_1^0 p_2^1} [\hat g^{(1, 0)}(p_1^0, p_2^1)]}{\hat h(p_1^0, p_2^1)} = - \frac{\ddot{b}_K(p_1^0, p_2^1)^\top \hat \alpha^{(1, 0)}}{\hat h(p_1^0, p_2^1)} ,
\end{split}
\end{align*}
where $\ddot{b}_K(p_1, p_2) \coloneqq \partial_{p_1 p_2} [ b_K(p_1, p_2) ]$.
Finally, we can estimate $m^{(1, 0)}(x, p_1^0, p_2^1)$ by
\begin{align}\label{eq:m10}
	\hat m^{(1, 0)}(x, p_1^0, p_2^1) \coloneqq x_1^\top \hat \beta_1^{(1, 0)} + \hat \bE_n[U_1^{(1, 0)} | V_1 = p_1^0, V_2 = p_2^1].
\end{align}

Next, we describe the estimation of $m^{(0, 0)}(x, p_1^0, p_2^1) = x_1^\top \beta_1^{(0, 0)} + \bE [U_1^{(0, 0)} | V_1 = p_1^0, V_2 = p_2^1]$.
First, observe that
\begin{align*}
	\bE [ I^{(0, 0)} Y_1| I^{(0, 0)}, X, \mathbf{P}] 
	& = I^{(0, 0)} X_{1}^\top \beta_1^{(0, 0)} + I^{(0, 0)} \bE [U_{1}^{(0, 0)} | I^{(0, 0)} = 1, X, \mathbf{P}].
\end{align*}
The same argument as in Theorem \ref{thm:multiple} gives
\begin{align}\label{eq:additive}
\begin{split}
	\bE [U_{1}^{(0, 0)}| I^{(0, 0)} = 1, X, \mathbf{P}]  
	& = \frac{\int_{P_1^0}^{1} \int_{P_2^0}^1 \bE [U_{1}^{(0, 0)}| V_1= v_1, V_2= v_2] h(v_1, v_2) \mathrm{d}v_1 \mathrm{d}v_2}{\mathcal{L}^{(0, 0)}(\mathbf{P})} \\
	& \quad - \frac{(1 - \lambda_X^*) \int_{P_1^0}^{P_1^1} \int_{P_2^0}^{P_2^1} \bE [U_{1}^{(0, 0)}| V_1= v_1, V_2= v_2] h(v_1, v_2) \mathrm{d}v_1 \mathrm{d}v_2}{\mathcal{L}^{(0, 0)}(\mathbf{P})} \\
	& \hspace{-128pt} = \frac{\lambda_X^* g_1^{(0, 0)}(P_1^0, P_2^0) + (1 - \lambda_X^*) g_2^{(0, 0)}(P_1^1, P_2^0) + (1 - \lambda_X^*) g_3^{(0, 0)}(P_1^0, P_2^1) + (1 - \lambda_X^*) g_4^{(0, 0)}(P_1^1, P_2^1)}{\mathcal{L}^{(0, 0)}(\mathbf{P})},
\end{split}
\end{align}
where $\mathcal{L}^{(0, 0)}(\mathbf{P})$ is the conditional probability of $D = (0,0)$ given $\mathbf{P}$, and $g_l^{(0,0)}$'s, $l = 1, \ldots, 4$, are bivariate real-valued functions.\footnote{\label{foot:gfunc}
	Specifically, $g_1^{(0,0)}(p_1^0, p_2^0) = \int_{p_1^0}^1 \int_{p_2^0}^1 \bE [U_{1}^{(0, 0)}| V_1= v_1, V_2= v_2] h(v_1, v_2) \mathrm{d}v_1 \mathrm{d}v_2$, and $g_2^{(0,0)}(p_1^1, p_2^0)$ and $g_3^{(0,0)}(p_1^0, p_2^1)$ are obtained by replacing $(p_1^0, p_2^0)$ in the right-hand side with $(p_1^1, p_2^0)$ and $(p_1^0, p_2^1)$, respectively.
	Moreover, $g_4^{(0,0)}(p_1^1, p_2^1) = -\int_{p_1^1}^1 \int_{p_2^1}^1 \bE [U_{1}^{(0, 0)}| V_1= v_1, V_2= v_2] h(v_1, v_2) \mathrm{d}v_1 \mathrm{d}v_2$.
}
Hence, similarly to \eqref{eq:plm2}, we obtain the following regression model:
\begin{align}\label{eq:plm00}
\begin{split}
	\tilde I^{(0, 0)} Y_1 = \tilde I^{(0, 0)} X_{1}^\top \beta_1^{(0, 0)} 
	& + \tilde T^{(0, 0)} \lambda_X^* g_1^{(0, 0)}(P_1^0, P_2^0) + \tilde T^{(0, 0)} (1 - \lambda_X^*) g_2^{(0, 0)}(P_1^1, P_2^0) \\
	& + \tilde T^{(0, 0)} (1 - \lambda_X^*) g_3^{(0, 0)}(P_1^0, P_2^1) + \tilde T^{(0, 0)} (1 - \lambda_X^*) g_4^{(0, 0)}(P_1^1, P_2^1) + \tilde e^{(0, 0)},
\end{split}
\end{align}
where $\tilde I^{(0, 0)} \coloneqq \tau_\varpi(\mathcal{L}^{(0, 0)}) I^{(0, 0)}$, $\tilde T^{(0, 0)} \coloneqq \tilde I^{(0, 0)} / \mathcal{L}^{(0,0)}$, and $\bE [\tilde e^{(0, 0)} | I^{(0, 0)}, X, \mathbf{P}] = 0$.
Assuming again that each $g_l^{(0,0)}$, $l = 1, \ldots , 4$, can be approximated by $g_l^{(0,0)}(\cdot, \cdot) \approx b_K(\cdot, \cdot)^\top \alpha_l^{(0,0)}$ and replacing $\tilde I^{(0,0)}$, $\tilde T^{(0,0)}$, $\lambda_X^*$, and $\mathbf{P}$ with their estimators $\hat I^{(0,0)} \coloneqq \tau_\varpi(\hat{\mathcal{L}}^{(0, 0)}) I^{(0, 0)}$, $\hat T^{(0, 0)} \coloneqq \hat I^{(0, 0)} / \hat{\mathcal{L}}^{(0,0)}$, $\hat \lambda_X$, and $\hat{\mathbf{P}}$, respectively, $\beta_1^{(0, 0)}$ and $\alpha_l^{(0,0)}$'s can be estimated by LS regression.\footnote{
	It is possible to use different orders of basis terms to approximate each component of the functions $g_l^{(0,0)}$'s for $l = 1, \ldots, 4$, but we use the same order $K$ for all, for simplicity.
	Also note that the ``locations'' of the functions $g_l^{(0,0)}$'s are not identified without further restrictions.
	To simplify our presentation, we postulate that an appropriate location normalization is made implicitly.
}
Let $\hat \beta_1^{(0, 0)}$ and $\hat \alpha_l^{(0,0)}$ be the resulting LS estimators.
Then, each $g_l^{(0,0)}(p_1, p_2)$ can be estimated by $\hat g_l^{(0, 0)}(p_1, p_2) \coloneqq b_K(p_1, p_2)^\top \hat \alpha_l^{(0, 0)}$.
Moreover, \eqref{eq:additive} implies that 
\begin{align*}
	\bE [U_1^{(0, 0)} | V_1 = p_1^0, V_2 = p_2^1] 
	=\frac{\partial_{p_1^0 p_2^1}[g_3^{(0, 0)}(p_1^0, p_2^1)]}{h(p_1^0, p_2^1)} ,
\end{align*}
and thus we can estimate this by
\begin{align*}
\begin{split}
	\hat \bE_n[U_1^{(0, 0)} | V_1 = p_1^0, V_2 = p_2^1] 
	\coloneqq \frac{\partial_{p_1^0 p_2^1} [\hat g_3^{(0, 0)}(p_1^0, p_2^1)]}{\hat h(p_1^0, p_2^1)} 
	= \frac{\ddot{b}_K(p_1^0, p_2^1)^\top \hat \alpha_3^{(0, 0)}}{\hat h(p_1^0, p_2^1)}.
\end{split}
\end{align*}
Consequently, $m^{(0, 0)}(x, p_1^0, p_2^1)$ can be estimated by
\begin{align}\label{eq:m00}
	\hat m^{(0, 0)}(x, p_1^0, p_2^1) \coloneqq x_1^\top \hat \beta_1^{(0, 0)} + \hat \bE_n[U_1^{(0, 0)} | V_1 = p_1^0, V_2 = p_2^1].
\end{align}
Finally, $\text{MTE}_{\text{direct}}^{(0)}(x, p^0_1, p^1_2)$ can be estimated by
\begin{align*}\label{eq:estim-direct}
	\widehat{\text{MTE}}_{\text{direct}}^{(0)}(x, p^0_1, p^1_2) \coloneqq \hat m^{(1, 0)}(x, p_1^0, p_2^1) - \hat m^{(0, 0)}(x, p_1^0, p_2^1).
\end{align*}

\subsection{Asymptotics}\label{subsec:asymptotics}

In this subsection, we discuss the asymptotic properties of the proposed estimators.
We mainly investigate the MTR estimator for $D = (1,0)$.
Analogous arguments apply to the other cases.

For simplicity of presentation, all the technical assumptions are summarized in Appendix \ref{subsubsec:prep}.
Let us comment on some of the high-level assumptions employed here.
\phantomsection\label{page:R2-1-e}\Copy{R2-1-e}{
	First, in Assumption \ref{as:data}, we assume that $\{\{(Y_{ji}, D_{ji}, W_{ji})\}_{j=1}^2\}_{i=1}^n$ are independent and identically distributed across $i$.
	Although we do not restrict the dependence between the variables for $j$ and those for $-j$, the dependence across different pairs is ruled out.
	We conjecture that if the dependence is sufficiently weak, we can establish similar results to those given below (cf. \citealpmain{chen2015optimal, lee2016series}).
	In addition, if there is pair-specific unobserved heterogeneity, the assumption of identical distribution might be questionable. 
	To mitigate this problem, in our empirical application, we include school fixed effects in the model.
}
Assumption \ref{as:1stage}(i) assumes the $\sqrt{n}$-consistency of the ML estimator.
In Assumption \ref{as:sieve}, we require that the functions $g^{(1,0)}$ and $b_K$ are sufficiently smooth so that they are at least $s$-times continuously differentiable for some $s \ge 2$ and that the derivatives of $g^{(1,0)}$ can be well approximated in the space spanned by the derivatives of the basis $b_K$.
Lastly, Assumption \ref{as:opnorm} requires that the linear projection of any bounded and continuous function onto $(\tilde I_i^{(1,0)} X_{1i}, \tilde T_i^{(1,0)} b_K (P_{1i}^0, P_{2i}^1))$ is ``stable'' in a sense that its sup-operator norm is stochastically bounded.
Similar conditions can be found, for example, in \citetmain{huang2003local} and \citetmain{chen2018optimal}.

The next theorem establishes the uniform convergence rate of our MTR estimator $\hat m^{(1, 0)}(x, p_1, p_2)$ and that of the ``oracle'' estimator $\tilde m^{(1, 0)}(x, p_1, p_2)$ obtained by assuming that the true parameters in the first stage are known (see Appendix \ref{subsubsec:prep} for more precise definition).
Denote $\mathcal{S}^{(1,0)} \coloneqq \text{supp}[P_1^0, P_2^1]$.

\begin{theorem}\label{thm:cdmean}
	Suppose that Assumptions \ref{as:gamemodel}, \ref{as:linear}, \ref{as:data}, \ref{as:1stage}, \ref{as:eigen}(i), and \ref{as:emom}--\ref{as:rate2} hold.
	Then, for a given $x \in \text{supp}[X]$, we have
	\begin{align*}
	\renewcommand{\arraystretch}{2}
	\begin{array}{cl}
	\text{(i)} & \displaystyle \sup_{(p_1, p_2) \in \mathcal{S}^{(1,0)}}\left| \tilde m^{(1, 0)}(x, p_1, p_2)  -  m^{(1, 0)}(x, p_1, p_2)  \right| =  O_P(\zeta_0(K) K \sqrt{\log n / n} ) +  O_P(K^{(2 - s) /2}),\\
	\text{(ii)} & \displaystyle \sup_{(p_1, p_2) \in \mathcal{S}^{(1,0)}}\left| \hat m^{(1, 0)}(x, p_1, p_2)  -  m^{(1, 0)}(x, p_1, p_2)  \right| = O_P(\zeta_0(K) K \sqrt{\log n / n} ) +  O_P(K^{(2 - s) /2}),
	\end{array}
	\renewcommand{\arraystretch}{1}
	\end{align*}
	where $\zeta_0(K) \coloneqq \sup_{(p_1, p_2) \in [0,1]^2 }\| b_K(p_1, p_2) \|$ with $\| \cdot \|$ being the Euclidean norm.  
\end{theorem}

The proof of the theorem is straightforward from Lemmas \ref{lem:parametric} and \ref{lem:unifconv2}, and thus it is omitted.
When $\zeta_0(K) \asymp \sqrt{K}$, by choosing $K \asymp (\log n /n)^{-1/(1 + s)}$, we can obtain
\[
\sup_{(p_1, p_2) \in \mathcal{S}^{(1,0)}}\left| \hat m^{(1, 0)}(x, p_1, p_2)  -  m^{(1, 0)}(x, p_1, p_2)  \right| = O_P\left(\left(\log n / n\right)^{\frac{s - 2}{2s + 2}} \right).
\]
The above result implies that our MTR estimator can converge at the optimal uniform rate of \citetmain{stone1982optimal}.

The next theorem shows that the asymptotic distribution of the feasible estimator is also equivalent to that of the infeasible oracle estimator.

\begin{theorem}\label{thm:normal1}
	Suppose that Assumptions \ref{as:gamemodel}, \ref{as:linear}, and \ref{as:data}--\ref{as:rate2} hold.
	For given $x \in \text{supp}[X]$ and $(p_1^0, p_2^1) \in \mathcal{S}^{(1,0)}$, if in addition $K / \| \ddot{b}_K(p_1^0, p_2^1) \| \to 0$, $\zeta_0(K) \sqrt{K/n} \to 0$, and $\sqrt{n} K^{(2 - s)/2} = O(1)$ hold, then we have
	\begin{align*}
		\text{(i)} & \qquad \frac{\sqrt{n} \left( \tilde m^{(1, 0)}(x, p_1^0, p_2^1)  -  m^{(1, 0)}(x, p_1^0, p_2^1) \right)}{\sigma_K^{(1,0)}(p_1^0, p_2^1)} \overset{d}{\to} N(0, 1),\\
		\text{(ii)} & \qquad \frac{\sqrt{n} \left( \hat m^{(1, 0)}(x, p_1^0, p_2^1)  -  m^{(1, 0)}(x, p_1^0, p_2^1) \right)}{\sigma_K^{(1,0)}(p_1^0, p_2^1)} \overset{d}{\to} N(0, 1),
	\end{align*}
	where the definition of $\sigma_K^{(1,0)}(p_1^0, p_2^1)$ can be found in \eqref{eq:sd10}.
\end{theorem}
The standard deviation $\sigma_K^{(1,0)}(p_1^0, p_2^1)$ can be easily estimated by a sample analog, replacing the true values and functions with their estimates.
\bigskip

Let us briefly discuss the asymptotic properties of the MTR estimator for $D=(0,0)$ given in \eqref{eq:m00}.
Under conditions similar to those in Theorem \ref{thm:normal1}, we can show the following asymptotic normality result:
\begin{align*}
	\frac{\sqrt{n} \left( \hat m^{(0, 0)}(x, p_1^0, p_2^1)  -  m^{(0, 0)}(x, p_1^0, p_2^1) \right)}{\sigma_{K}^{(0,0)}(p_1^0, p_2^1)} \overset{d}{\to} N(0, 1)
\end{align*}
for $(p_1^0, p_2^1) \in \mathcal{S}^{(1,0)}$, where the definition of $\sigma_{K}^{(0,0)}(p_1^0, p_2^1)$ can be found in \eqref{eq:sd00}.

Finally, the limiting distribution of $\widehat{\text{MTE}}_{\text{direct}}^{(0)}(x, p_1^0, p_2^1)$ can be characterized as follows:
\begin{align*}
	\frac{\sqrt{n} \left( \widehat{\text{MTE}}_{\text{direct}}^{(0)}(x, p_1^0, p_2^1)  -  \text{MTE}_{\text{direct}}^{(0)}(x, p_1^0, p_2^1) \right)}{\sqrt{\left[\sigma_K^{(1,0)}(p_1^0, p_2^1)\right]^2 + \left[\sigma_{K}^{(0,0)}(p_1^0, p_2^1)\right]^2}} \overset{d}{\to} N(0, 1).
\end{align*}
The limiting distributions of the other MTE estimators can be derived similarly.

\subsection{An over-identified estimator}\label{subsec:overidentification:est}

As mentioned above, we can improve the estimation efficiency based on the over-identification result.
For $(p_1, p_2) \in \underline{\mathcal{S}}$, where $\underline{\mathcal{S}} \coloneqq \bigcap_{(d_1, d_2) \in \{ 0, 1 \}^2} \text{supp}[P_1^{d_1}, P_2^{d_2}]$, Theorem \ref{thm:multiple} suggests that we can estimate $m^{(0, 0)}(x, p_1, p_2)$ in four ways: $\hat m_l^{(0, 0)}(x, p_1, p_2) \coloneqq x_1^\top \hat \beta_1^{(0, 0)} + \hat \bE_{ln}[ U_1^{(0, 0)} | V_1 = p_1, V_2 = p_2 ]$ for $l = 1, \dots, 4$, where $\hat \bE_{ln} [ U_1^{(0, 0)} | V_1 = p_1, V_2 = p_2 ] \coloneqq \hat \kappa_l \ddot b_K(p_1, p_2)^\top \hat \alpha_l^{(0, 0)}$ with $\hat \kappa_1 = \hat \kappa_2 = \hat \kappa_3 \coloneqq 1 / \hat h(p_1, p_2)$ and $\hat \kappa_4 \coloneqq - 1 / \hat h(p_1, p_2)$.\footnote{
	\phantomsection\label{page:R2-5}\Copy{R2-5}{
		The over-identification result essentially builds on the continuity of strategic effects, which can be examined by the first-stage estimation result.
		Then, if the continuity is verified, the estimated $\hat m_l^{(0, 0)}(x, p_1, p_2)$, $l = 1, \ldots, 4$, must all take a ``close'' value at each $(p_1, p_2)$ in the common support.
		This gives us a testable implication for our model specification.
		Investigating the property of this testing procedure is quite intriguing, but we leave this for future work.
	}
}
Let $w = (w_1, w_2, w_3, w_4)^\top \in [0,1]^4$ be a vector of fixed weights such that $\sum_{l = 1}^4 w_l = 1$.
Then, $m^{(0, 0)}(x, p_1, p_2)$ can be estimated by the weighted average of the four estimators: 
\begin{align}\label{eq:over_id_est}
\hat m_{\text{over-id}}^{(0, 0)}(x, p_1, p_2; w) \coloneqq \sum_{l=1}^{4} w_l \hat m_l^{(0, 0)}(x, p_1, p_2).
\end{align}
In the same manner as in the proof of Theorem \ref{thm:normal1}, we can show
\begin{align*}
	\frac{\sqrt{n} \left( \hat m_{\text{over-id}}^{(0, 0)}(x, p_1, p_2; w) - m^{(0, 0)}(x, p_1, p_2) \right)}{\sigma_{\text{over-id}, K}^{(0, 0)}(p_1, p_2; w) } \convd N(0, 1).
\end{align*}
By easy calculations, we can see that the variance $[\sigma_{\text{over-id}, K}^{(0, 0)}(p_1, p_2; w)]^2$ has a quadratic form with respect to $w$, and, thus, the  ``optimal'' weights that attain the minimum variance in the class of the weighted average estimators $\hat m_{\text{over-id}}^{(0, 0)}(x, p_1, p_2; w)$ can be easily found by quadratic programming.\footnote{
	Specifically, $[ \sigma_{\text{over-id}, K}^{(0, 0)}(p_1, p_2; w) ]^2 = w^\top \Upsilon w$, and $\Upsilon = (\Upsilon_{lm})$ is the $4 \times 4$ symmetric matrix with the $(l, m)$-th element given by $\Upsilon_{lm} \coloneqq \kappa_l \kappa_m \ddot b_K(p_1, p_2)^\top \mathbb{S}_{lK} [\Psi_K^{(0,0)}]^{-1} \Sigma_K^{(0, 0)} [\Psi_K^{(0,0)}]^{-1} \mathbb{S}_{mK}^\top \ddot b_K(p_1, p_2)$, where $\kappa_l$ is the true value of $\hat \kappa_l$, and $\mathbb{S}_{lK}$, $\Psi_K^{(0,0)}$, and $\Sigma_K^{(0, 0)}$ are as in \eqref{eq:sd00}.
}
Although $[\sigma_{\text{over-id}, K}^{(0, 0)}(p_1, p_2; w)]^2$ is unknown in practice, we can estimate the optimal weights by minimizing its sample analog.

\phantomsection\label{page:R1-2-2}\Copy{R1-2-2}{
	It should be noted that the over-identified estimator relies on the assumption of non-empty $\underline{\mathcal{S}}$.
	If $\underline{\mathcal{S}}$ is empty, there is at least one estimator among the four that does not contribute to the estimation of $m^{(0, 0)}(x, p_1, p_2)$.
	In this situation, the over-identified estimator may entail a bias arising from an inaccurate extrapolation.
	Thus, it is always recommended to check the joint distribution of $(\hat P_1^{d_1}, \hat P_2^{d_2})$ before computing $\hat m_{\text{over-id}}^{(0, 0)}(x, p_1, p_2; w)$ to see if $(p_1, p_2) \in \underline{\mathcal{S}}$ actually holds.
}

\phantomsection\label{page:R2-7-b-1}\Copy{R2-7-b-1}{
	Finally, as mentioned by a referee, we can consider other estimation strategies than the above-mentioned over-identified estimator, which also use the over-identification result but in a more implicit way.
	As far as we have checked via numerical simulations, the performances of these estimators, including ours, are similar.
	To save space, we discuss the constructions of these alternative estimators in Appendix \ref{subsec:alt_m00}.
}


\section{Numerical Illustrations}\label{sec:numeric}

\subsection{Monte Carlo simulation results}\label{subsec:MC}

In order to evaluate the finite sample performance of our MTE estimator, we conduct a set of Monte Carlo simulations based on the over-identified estimator \eqref{eq:over_id_est}.
To save space,  we here briefly  report the summary of the results, and detailed information about the experiments are provided in Appendix \ref{sec:MC}.

The simulation results provide us with several practical guidelines as follows.
First, when comparing bivariate power series and tensor-product B-splines for the basis functions, we find that the power-series-based estimator achieves smaller RMSE (root mean squared error) than the B-splines-based estimator, especially when the sample size is not large. 
For this result, note that the power series is a globally supported series, and thus the power-series estimator tends to have smaller variance at the cost of lower flexibility than the B-splines estimator.
Theoretically speaking, B-splines estimators have a better approximation property than power-series estimators, and in fact the latter may not attain the optimal convergence rate.
Hence, the above result may be due to the chosen sample sizes in this simulation setup.
Second, in terms of RMSE, the MTE parameters can be more precisely estimated by ridge regression with a small regularization parameter.
Remarkably, using the ridge regression often reduces the bias for the B-splines estimator.
Since the tensor-product B-splines involves a large number of basis terms, which causes a severe multicollinearity problem for the non-regularized estimators, there is a higher risk of generating extremely outlying estimates compared with the regularized estimators.
Thus, it would be desirable to employ a regularized regression in practical situations with moderate sample size.
Note that adding a sufficiently small regularization factor does not alter our asymptotic theory.

\subsection{Empirical application: Adolescent delinquency and academic performance}\label{subsec:empirical}

As an empirical illustration, we estimate the impacts of risky behaviors such as smoking and drinking alcohol on adolescents' academic performance.
The empirical analysis is performed on the National Longitudinal Study of Adolescent Health (Add Health) data that provide nationally representative information on 7--12th graders in the U.S. 
The survey was conducted during the 1994--1995 school year. 
The survey elicits information on variables such as the social and demographic characteristics of the respondents, education levels and occupations of their parents, and friendship connections.

In this analysis, following \citemain{card2013peer}, we focus on the interactions among pairs of closest opposite-gender friends.
In the survey, each respondent was asked to list up to five friends of each gender in order from best to 5th best friend.
We first exclude missing nomination data (caused by non-response or because the nominee was not included in the Add Health data) and identify each respondent's closest opposite-gender friend in the available dataset.
Since a student's closest friend's friends might not include the student him/herself, we consider that the pair is formed only when his/her best friend nominates him/her either in the first place or second place.
This procedure results in 7,631 opposite-gender pairs of students.

\phantomsection\label{page:AE-17}\Copy{AE-17}{
	Before proceeding, let us clarify two important caveats in interpreting the results of the following analysis.
	First, as previously mentioned, the results presented are all conditional on the already formed pairs.
	In reality, it is often observed that people with similar profiles are more likely to be friends with each other.
	Such \textit{homophilic} preference generally poses an additional endogeneity problem, making it difficult to generalize the results beyond the current data.
	Second, in reality, not only the interference within best friends, but also that between different friend groups may affect the students' activities.
	This type of interaction cannot be accommodated in our framework.
}

The treatment variable ($D$) is defined as the student's participation in risky behaviors including smoking, drinking, truancy, and fighting, and the outcome variable ($Y$) is the natural log of GPA.
Table \ref{table:def} summarizes the explanatory variables used ($X$) and their definitions.
For continuous IVs ($Z$) for each gender, we construct them based on the average parental characteristics of the students' fourth and fifth ranked non-best friends.
For our method to work with this choice of IV, the following two identification assumptions must be met: (1) the characteristics of non-best friends affect the students' delinquency, and (2) the non-best friends' characteristics exert no direct effect on the students' academic performance.
The arguments supporting these assumptions are presented below.

For assumption (1), in the econometrics literature, \citemain{lee2014binary} and \citemain{lin2014peer} present empirical evidence that the personalities of friends, including non-best friends, influence participation in risky behaviors.
In the education and sociology literature, several studies reveal a close relation between student delinquency and the school environment, especially the shared perception of social norms (e.g., \citealpmain{paluck2016changing}).
Because the school environment for a student is created by all school mates, not only by one's own best friends, this point is consistent with our assumption.
\citemain{rees2011one} also point out that a kind of social pressure exists in youth friendships to achieve peer acceptance, which drives the students to engage in risky behavior in groups. 
Here, the role of non-best friend relationships particularly matters because non-best friends have less intimate and caring relationships than best friends have (p.199, \citealpmain{rees2011one}).

Next, for assumption (2), several studies indicate the existence of positive interactions in terms of academic outcomes between reciprocal best friends (e.g., \citealpmain{vaquera2008you,cherng2013along}).
For example, \citemain{vaquera2008you} argue that reciprocal friendships are likely to be more emotionally supportive as well as superior information resources compared to friendships that are not reciprocal.
Because our definition of non-best friends does not involve reciprocity, in view of the above described studies, it would be reasonable to assume that the non-best friend characteristics do not directly affect the student's academic performance.\footnote{
	\phantomsection\label{page:AE-19-1}\Copy{AE-19-1}{
		To verify whether this exclusion restriction assumption is actually supported in our data, we have performed the test proposed by \citemain{kedagni2020generalized}.
		The results are reported in Table \ref{table:KMtest} in Appendix \ref{sec:supp}, which indicate that our chosen IVs are likely to be valid.
		However, note that this testing procedure is not a formal diagnosis for our specific setting, and that developing such a formal test is beyond the scope of our study.
	}
}

After excluding observations with missing values for $(D,X,Z)$, the size of the sample used in the estimation of the treatment decision model is 6,053.
When estimating the MTRs, observations with missing values for $Y$ are further excluded.
Table \ref{table:d_dist} presents the empirical distribution of $(D_1, D_2)$, where male students are labeled as ``player 1'' and females as ``player 2''.
As expected, the case of $D = (0,0)$ shows the largest share of our sample.
There is an interesting asymmetry between the cases $D = (1,0)$ and $D = (0,1)$; that is, the number of males who solely participate in delinquent behaviors is significantly larger than the number of such females.

\begingroup
\renewcommand{\arraystretch}{1.2}
\begin{table}[!h]
	\caption{Distribution of $(D_1, D_2)$}
	\begin{center}
		\small{\begin{tabular}{|cc|ccc|}
				\hline
				&  & \multicolumn{2}{c}{\underline{Male}}  & \\
				&  & $D_1 = 0$ & $D_1 = 1$ & Total\\
				\hline
				\multirow{2}{*}{\underline{Female}}& $D_2 = 0$ & 2,707       & 1,351       & 4,058\\
				& $D_2 = 1$ & 736         & 1,259       & 1,995\\
				& Total     & 3,443       & 2,610       & 6,053\\
				\hline
		\end{tabular}}
		\label{table:d_dist}
	\end{center}
\end{table}
\endgroup

To save space, the detailed estimation results of the treatment decision model are omitted here and are provided in Table \ref{table:1ststage}.
Some noteworthy findings are as follows.
First, approximately three-quarters of the pairs choose $D = (0,0)$ in the multiple equilibria situation.
Second, the education levels of the parents (the father's level of education, in particular) have a significant impact on reducing the risky activities of the students.
Third, having friends whose mothers have a professional job tends to discourage risky behavior.
Fourth, white and Asian students are more likely to be influenced by their friends' behaviors than other students.
Finally, students who belong to a sports club are less susceptible to their friends' behaviors.
The last two findings clearly indicate the presence of heterogeneous peer effects.
Although such heterogeneity is often overlooked in the literature on structural game estimation of adolescent activities (e.g., \citealpmain{soetevent2007discrete, card2013peer, lee2014binary}), it is consistent with the results in other contexts (e.g., \citealpmain{eisenberg2014peer, hsieh2018smoking}).

Figure \ref{fig:PiDelhist} in Appendix \ref{sec:supp} presents the histograms of the estimated $(\pi_1^0(W_1), \pi_2^0(W_2))$ and $(\Delta_1(W_1), \Delta_2(W_2))$, which indicate that these variables are distributed over some bounded range and, thus, that estimating the whole functional form of the MTEs would be challenging.
To see this point more concretely, we provide the histograms of the estimated $(P_j^0, P_j^1)$ for $j=1$ (male) and $2$ (female) in Figure \ref{fig:Phist}.
These figures suggest that the overlapping support condition can hardly hold outside the interval $[0.2, 0.7]$, implying that MTE estimates might not be reliable when $V$ is outside this interval. 

Now, we present our main results of estimating the MTE parameters.
Since our sample size is not very large, following the suggestion from the Monte Carlo results, we employ the over-identified estimator \eqref{eq:over_id_est} with a third-order bivariate power series for the basis function and use ridge regression for the parameter estimation with penalty  equal to $n^{-1}$. 
Figure \ref{fig:MTEs} summarizes the estimated direct, total, and indirect MTEs for both genders.
For comparison, we also report the MTE estimates obtained from a model where the first-stage treatment decision model is a standard bivariate probit model with no strategic interactions.
The MTE estimator for this model can be constructed in exactly the same way as our two-step estimator based on the identification result in Appendix \ref{subsec:nointeraction}.
For the value of $X$, the MTEs are evaluated at the median over all observations of men and women altogether.

Figure \ref{fig:directMTE} shows that the direct impact of their own delinquent activities is significantly negative for both male and female students.
The direct MTE for female students tends to have a weak negative slope with respect to the male partner's $V$. 
That is, when the male partner has a stronger hesitation in participating in risky activities, the direct MTE on female students' GPA becomes larger.
For male students, their direct MTE is relatively flat with respect to the female partner's $V$ (partially due to regularization).
Figure \ref{fig:totalMTE} illustrates that the total MTE is also significantly negative and is larger than the direct MTE for both male and female students.
The gap between the direct and total MTE corresponds to indirect MTE, which suggests that the students' own GPA is indirectly influenced by their peers' delinquent behavior, as shown in Figure \ref{fig:indirectMTE}.\footnote{\label{foot:indirect}
	Although the asymptotic distribution of the indirect MTE estimator is not formally discussed in this study, it can be easily derived by combining the results of Theorem \ref{thm:normal1} and that in Appendix \ref{subsec:totalMTE}.
}
As expected, the magnitude of the indirect MTE is clearly smaller than that of the direct MTE.
Interestingly, in contrast to the direct MTE estimates, we can observe that the indirect MTE estimates have a weak increasing tendency with respect to the best friend's $V$.
That is, if the best friend is a person who hesitates in engaging in risky behavior, the indirect effect would be slightly smaller.

The MTE estimates obtained from the non-strategic treatment model are significantly different from those obtained in our game-theoretic model.
In particular, the direct and total MTE for female students are clearly overestimated, and the indirect MTEs for male students have opposite sign as compared to the strategic model.
These results suggest that the misspecification in the first-stage model may produce certain biases that cannot be easily controlled unless the strategic interactions are actually incorporated into the model.

\begin{figure}[!ht]
	\centering
	\begin{subfigure}{16cm}
		\includegraphics[width = 16cm]{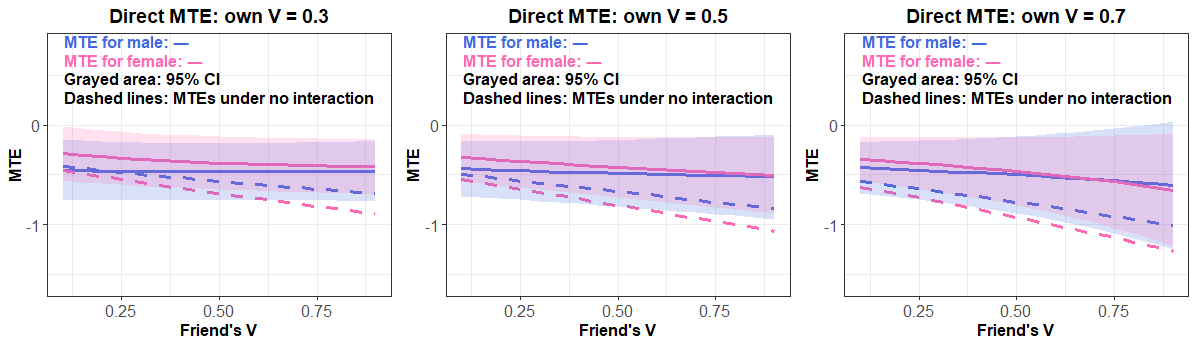}
		\caption{Estimated direct MTE.}
		\label{fig:directMTE}
	\end{subfigure}
	\begin{subfigure}{16cm}
		\includegraphics[width = 16cm]{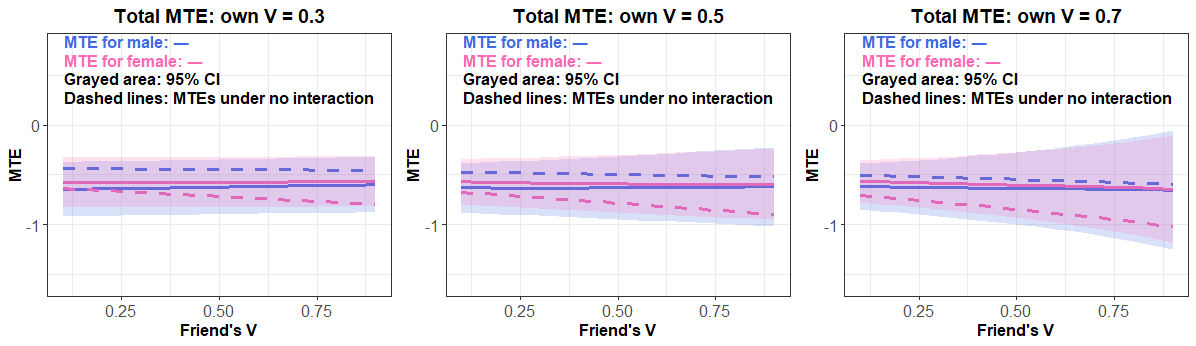}
		\caption{Estimated total MTE.}
		\label{fig:totalMTE}
	\end{subfigure}
	\begin{subfigure}{16cm}
		\includegraphics[width = 16cm]{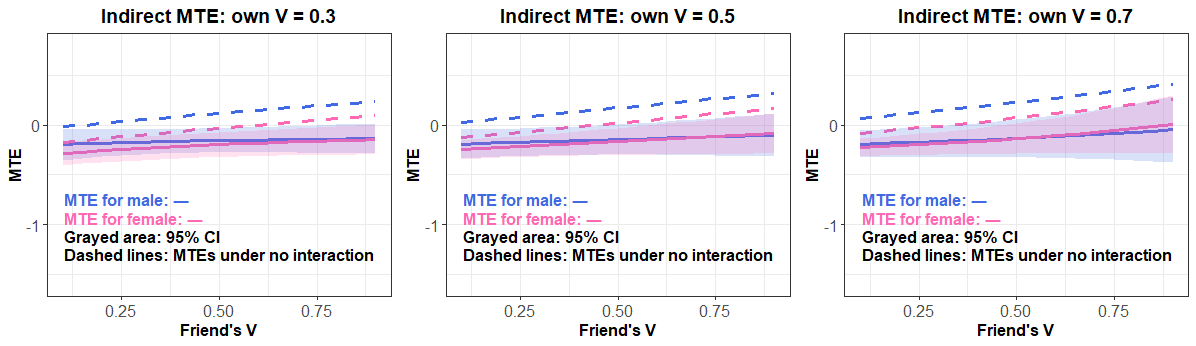}
		\caption{Estimated indirect MTE.}
		\label{fig:indirectMTE}
	\end{subfigure}
	\caption{Estimated MTEs}
	\label{fig:MTEs}
\end{figure}

\section{Conclusion} \label{sec:conclusion}

This study developed treatment effect models that admit both the treatment spillover and strategic interaction in the treatment decisions within a pair of agents.
We first demonstrated that the interaction in the treatment decisions can be modeled as a binary game of complete information with multiple equilibria.
Assuming an equilibrium selection rule, we showed that the MTE can be identified using an extended version of the LIV method.
Based on our identification results, we proposed two-step semiparametric series estimation for MTE parameters.
We showed that the proposed MTE estimator is uniformly consistent and asymptotically normal.

The results of this study suggest several extensions that would be promising to investigate, some of which have been mentioned above.
In addition, it would be worth extending our results to the case of strategic interaction among more than two players.
In the literature on game econometrics, only a few studies have addressed the point identification of game models of complete information with more than two players.
The main reason for this is that the characterization of the equilibrium is extremely complicated compared to the case of two-player games; in general, the number of multiple equilibria regions increases with the number of players (see, e.g., \citealpmain{soetevent2007discrete}).
Accordingly, the direct applicability of our approach to such cases is unclear.
On the other hand, while fixing the number of players at two, extending our model to an ordered treatment setup, as in \citetmain{vytlacil2006ordered}, should be more tractable.
As shown in \citetmain{card2013peer}, the Nash equilibrium in the two-player ordered-response game of complete information is uniquely characterized by an equilibrium selection assumption similar to ours.
Hence, although computation would be much more complicated than the present case, the identification and estimation approach proposed in this study would be applicable to such models.
Finally, it also may be beneficial to develop treatment evaluation techniques for treatment decision games of incomplete information.
Under incomplete information, we may use the rational expectation model developed in \citetmain{lee2014binary}, for example, as a treatment decision model, even when the number of players is large.
These topics are left for future research.

\section*{Acknowledgments}

The authors are grateful to the coeditor (Elie Tamer), the associate editor, and two anonymous referees for their insightful comments that significantly improved the paper.
We also thank to Yoichi Arai, Sukjin Han, Hiroaki Kaido, Shin Kanaya, Hiroyuki Kasahara, Toru Kitagawa, Tatsushi Oka, Ryo Okui, Yuya Sasaki, Yuta Toyama, Takuya Ura, Haiqing Xu, and the participants of conferences and seminars at various places for their valuable comments and suggestions.
This work was supported by JSPS KAKENHI Grant Numbers 15K17039, 17K13715, 19H01473, and 20K01597. 

This research uses data from Add Health, a program project directed by Kathleen Mullan Harris and designed by J. Richard Udry, Peter S. Bearman, and Kathleen Mullan Harris at the University of North Carolina at Chapel Hill, and funded by grant P01-HD31921 from the Eunice Kennedy Shriver National Institute of Child Health and Human Development, with cooperative funding from 23 other federal agencies and foundations. Special acknowledgment is due Ronald R. Rindfuss and Barbara Entwisle for assistance in the original design. Information on how to obtain the Add Health data files is available on the Add Health website (http://www.cpc.unc.edu/addhealth). No direct support was received from grant P01-HD31921 for this analysis.

\clearpage

\begin{center}
	{\Large Supplementary Appendix for ``Treatment Effect Models with Strategic Interaction in Treatment Decisions''}
	
	\textbf{(Not for publication)}
	\bigskip
	
	{\large Tadao Hoshino$^*$ and Takahide Yanagi$^\dagger$}
	\bigskip
	
	$^*$ School of Political Science and Economics, Waseda University.
	
	$^\dagger$ Graduate School of Economics, Kyoto University.
\end{center}

\appendix

\renewcommand{\thetable}{\Alph{section}.\arabic{table}}
\renewcommand{\thefigure}{\Alph{section}.\arabic{figure}}

\setcounter{table}{0}
\setcounter{figure}{0}

\section{Appendix: Proofs of Theorems} \label{sec:proof}

\subsection{Proofs of the theorems in Section \ref{sec:identification}}\label{sec:proof:identification}

\subsubsection{Proof of Theorem \ref{thm:gameiden}}

\paragraph{Proof of (i).}
We provide the proof for $p_1^0$ only because that for $p_2^0$ is analogous.
By Assumptions  \ref{as:complement}(ii) and \ref{as:infinity}(ii), the strict monotonicity of $F_{\varepsilon_2 | X = x}(\cdot)$ implies that $\lim_{w_{2,1} \to -\infty} H_{\rho_x}(p_1^0, p_2^1) = 0$ and, thus, $\lim_{w_{2,1} \to -\infty} \mathcal{L}^{(1, 0)}(w) = p_1^0$.
Since $\lim_{w_{2,1} \to -\infty} \mathcal{L}^{(1, 0)}(w)$ is identified from data under Assumption \ref{as:infinity}(i), this equality implies the identification of $p_1^0$.

\paragraph{Proof of (ii).}
We prove the identification of $(p_1^1, \rho_x)$ only.
We consider any pair of values $w = (w_1, w_2)$ and $\ddot w = (w_1, \ddot w_2)$ such that $p_2^0 \neq \ddot p_2^0$, where $\ddot p_2^0 = F_{\varepsilon_2 | X = x}(\pi_2^0(\ddot w_2))$.
Such $w$ and $\ddot w$ exist by Assumption \ref{as:infinity}(i) and the strict monotonicity of $F_{\varepsilon_2 | X = x}(\cdot)$.
We now have the following system:
\begin{align*}
	\left\{
	\begin{array}{l}
		\mathcal{L}^{(0, 1)}(w) = p_2^0 - H_{\rho_x}(p_1^1, p_2^0)\\
		\mathcal{L}^{(0, 1)}(\ddot w) = \ddot p_2^0 - H_{\rho_x}(p_1^1, \ddot p_2^0)\\
	\end{array}
	\right..
\end{align*}
Here, the unknown quantities are $\vartheta_{w_1} \coloneqq (p_1^1, \rho_x)$, as $p_2^0$ and $\ddot p_2^0$ are already identified by result (i).
If this system has a unique solution, we achieve the identification of $\vartheta_{w_1}$.
To proceed, we define
\begin{align*}
	G(\vartheta_{w_1}) \coloneqq 
	\begin{pmatrix}
		p_2^0 - H_{\rho_x}(p_1^1, p_2^0) \\
		\ddot p_2^0 - H_{\rho_x}(p_1^1, \ddot p_2^0)
	\end{pmatrix}
	.
\end{align*}
The Jacobian matrix of $G(\vartheta_{w_1})$ is given by 
\begin{align*}
	J_G(\vartheta_{w_1}) \coloneqq
	\begin{pmatrix}
		-H_{\rho_x}^{(1)}(p_1^1, p_2^0) & -H_{\rho_x}^{(\rho)}(p_1^1, p_2^0) \\
		-H_{\rho_x}^{(1)}(p_1^1, \ddot p_2^0) & -H_{\rho_x}^{(\rho)}(p_1^1, \ddot p_2^0)
	\end{pmatrix}
	,
\end{align*}
where $H_{\rho_x}^{(1)}(\cdot, \cdot)$ and $H_{\rho_x}^{(\rho)}(\cdot, \cdot)$ are the partial derivatives of $H_{\rho_x}(\cdot, \cdot)$ with respect to the first argument and $\rho_x$, respectively.
Its determinant is given by
\begin{align*}
	|J_G(\vartheta_{w_1})| 
	=  H_{\rho_x}^{(1)}(p_1^1, p_2^0) H_{\rho_x}^{(1)}(p_1^1, \ddot p_2^0) \left( \frac{H_{\rho_x}^{(\rho)}(p_1^1, \ddot p_2^0)}{H_{\rho_x}^{(1)}(p_1^1, \ddot p_2^0)} - \frac{H_{\rho_x}^{(\rho)}(p_1^1, p_2^0)}{H_{\rho_x}^{(1)}(p_1^1, p_2^0)} \right),
\end{align*}
which is positive for any $p_2^0 > \ddot p_2^0$ and is negative for any $p_2^0 < \ddot p_2^0$ under Assumptions \ref{as:copula}(ii)--(iii) (see Lemma 4.1 of \citealpappendix{han2017identification}).
This implies that $J_G(\vartheta_{w_1})$ is of full rank provided $p_2^0 \neq \ddot p_2^0$.
Hence, because $\text{supp}[P_1^1 | X = x] \times (\underline{c}, \bar c)$ is a simply connected space, Lemma 4.2 in \citeappendix{han2017identification} implies the identification of $\vartheta_{w_1}$.

\paragraph{Proof of (iii).}
Noting that $\mathcal{L}_{\text{mul}}(w) > 0$ under Assumption \ref{as:complement}, $\lambda_x$ is identified by $\lambda_x = ( H_{\rho_x}(p_1^1, p_2^1) - \mathcal{L}^{(1, 1)}(w) ) / \mathcal{L}_{\text{mul}}(w)$.

\qed

\subsubsection{Proof of Theorem \ref{thm:unique}}

We provide the proof for $m^{(1, 0)}(x, p_1, p_2)$ only, as the proof for $m^{(0, 1)}(x, p_1, p_2)$ is analogous.
From the law of iterated expectations, 
\begin{align*}
	& \psi^{(1, 0)}(x, p_1, p_2) \\
	&= \bE [ I^{(1, 0)} Y_1| X = x, P_1^0 = p_1, P_2^1 = p_2] \\
	&= \bE [Y_1^{(1, 0)} | X = x, P_1^0 = p_1, P_2^1 = p_2, D = (1, 0)] \cdot \Pr[D = (1, 0) | X = x, P_1^0 = p_1, P_2^1 = p_2] \\
	&= \bE [Y_1^{(1, 0)} | X = x, P_1^0 = p_1, P_2^1 = p_2, V_1 \le p_1, V_2 > p_2] \cdot \Pr[V_1 \le p_1, V_2 > p_2 | X = x, P_1^0 = p_1, P_2^1 = p_2] \\
	&= \bE [Y_1^{(1, 0)} | X = x, V_1 \le p_1, V_2 > p_2] \cdot \Pr[V_1 \le p_1, V_2 > p_2 | X = x],
\end{align*}
where we used Assumption \ref{as:IV}(i) in the last equality.
Here, it holds that
\begin{align*}
	&\bE [Y_1^{(1, 0)} | X = x, V_1 \le p_1, V_2 > p_2] = \frac{1}{\Pr[V_1 \le p_1, V_2 > p_2 | X = x]} \int_{p_2}^1 \int_0^{p_1} m^{(1, 0)}(x, v_1, v_2) h(v_1, v_2 | x) \mathrm{d}v_1 \mathrm{d}v_2.
\end{align*}
As a result, the cross-partial differentiation of $\psi^{(1, 0)}(x, p_1, p_2)$ with respect to $p_1$ and $p_2$ leads to
\begin{align*}
	& \psi^{(1,0)}(x, p_1, p_2) =  \int_{p_2}^1 \int_0^{p_1} m^{(1, 0)}(x, v_1, v_2) h(v_1, v_2 | x) \mathrm{d}v_1 \mathrm{d}v_2 \\
	& \Longrightarrow \partial_{p_1 p_2}[\psi^{(1,0)}(x, p_1, p_2)] = - m^{(1, 0)}(x, p_1, p_2) h(p_1, p_2 | x),
\end{align*}
by the Leibniz integral rule, provided that $m^{(1, 0)}(x, \cdot, \cdot) h(\cdot, \cdot | x)$ is continuous.
Rearranging the above equation, we obtain the desired result.
\qed

\subsubsection{Proof of Theorem \ref{thm:multiple}}

We prove only the case of $m^{(0, 0)}(x, p_1^0, p_2^0)$, as the proofs for the other cases are similar.
From the law of iterated expectations,
\begin{align*}
	\psi^{(0, 0)} (x, \mathbf{p} )
	&= \bE [ I^{(0, 0)} Y_1 | X = x, \mathbf{P} = \mathbf{p} ]\\
	&= \bE [Y_1^{(0, 0)} | X = x, \mathbf{P} = \mathbf{p}, D = (0, 0)] \cdot \Pr[D = (0, 0) | X = x, \mathbf{P} = \mathbf{p} ].
\end{align*}
For notational simplicity, we write $\mathcal{V}_{\text{mul}} = \mathcal{V}_{\text{mul}}(\mathbf{p}) \coloneqq \{ (v_1, v_2) \in [0, 1]^2 : p_1^0 < v_1 \le p_1^1, p_2^0 < v_2 \le p_2^1 \}$, and $\mathcal{V}_{\text{uni}}^{(0, 0)} = \mathcal{V}_{\text{uni}}^{(0, 0)}(\mathbf{p}) \coloneqq \{ (v_1, v_2) \in [0, 1]^2 : p_1^0 < v_1, p_2^0 < v_2 \} \backslash \mathcal{V}_{\text{mul}}$.
As $\mathcal{V}_{\text{uni}}^{(0, 0)}$ and $\mathcal{V}_{\text{mul}}$ are disjoint, by Assumptions \ref{as:multiple} and \ref{as:IV}(i), it holds that 
\begin{align*}
	& \bE [Y_1^{(0, 0)} | X = x, \mathbf{P} = \mathbf{p}, D = (0, 0)]\\
	&= \bE [Y_1^{(0, 0)} | X = x, V \in \mathcal{V}_{\text{uni}}^{(0, 0)} \lor (V \in \mathcal{V}_{\text{mul}} \land \epsilon \le \lambda_x)]\\
	&= \bE [Y_1^{(0, 0)} | X = x, V \in \mathcal{V}_{\text{uni}}^{(0, 0)}] \cdot \frac{\Pr[V \in \mathcal{V}_{\text{uni}}^{(0, 0)} | X = x]}{\Pr[V \in \mathcal{V}_{\text{uni}}^{(0, 0)} \lor (V \in \mathcal{V}_{\text{mul}} \land \epsilon \le \lambda_x) | X = x]} \\
	&\quad + \bE [Y_1^{(0, 0)} | X = x, (V \in \mathcal{V}_{\text{mul}} \land \epsilon \le \lambda_x)] \cdot \frac{\Pr[(V \in \mathcal{V}_{\text{mul}} \land \epsilon \le \lambda_x) | X = x]}{\Pr[V \in \mathcal{V}_{\text{uni}}^{(0, 0)} \lor (V \in \mathcal{V}_{\text{mul}} \land \epsilon \le \lambda_x) | X = x]} \\
	&= \bE [Y_1^{(0, 0)} | X = x, V \in \mathcal{V}_{\text{uni}}^{(0, 0)} ] \cdot \frac{\Pr[V \in \mathcal{V}_{\text{uni}}^{(0, 0)} | X = x]}{\Pr[V \in \mathcal{V}_{\text{uni}}^{(0, 0)} | X = x] + \lambda_x \cdot \Pr[V \in \mathcal{V}_{\text{mul}} | X = x]} \\
	&\quad + \bE [Y_1^{(0, 0)} | X = x, V \in \mathcal{V}_{\text{mul}} ] \cdot \frac{\lambda_x \cdot \Pr[V \in \mathcal{V}_{\text{mul}} | X = x]}{\Pr[V \in \mathcal{V}_{\text{uni}}^{(0, 0)} | X = x] + \lambda_x \cdot \Pr[V \in \mathcal{V}_{\text{mul}} | X = x]}.
\end{align*}
Similarly, we also have 
\begin{align*}
	\Pr[D = (0, 0) | X = x, \mathbf{P} = \mathbf{p} ]
	= \Pr[V \in \mathcal{V}_{\text{uni}}^{(0, 0)} | X = x] + \lambda_x \cdot \Pr[V \in \mathcal{V}_{\text{mul}} | X = x].
\end{align*}
As a result, we obtain
\begin{align*}
	\psi^{(0, 0)} (x, \mathbf{p})
	&= \bE [Y_1^{(0, 0)} | X = x, V \in \mathcal{V}_{\text{uni}}^{(0, 0)}] \cdot \Pr[V \in \mathcal{V}_{\text{uni}}^{(0, 0)} | X = x] \\
	&\quad +  \lambda_x \cdot \bE [Y_1^{(0, 0)} | X = x, V \in \mathcal{V}_{\text{mul}}] \cdot \Pr[V \in \mathcal{V}_{\text{mul}} | X = x].
\end{align*}
Further, it holds that
\begin{align*}
	& \bE [Y_1^{(0, 0)} | X = x, V \in \mathcal{V}_{\text{uni}}^{(0, 0)} ] \cdot \Pr[V \in \mathcal{V}_{\text{uni}}^{(0, 0)} | X = x] \\
	&= \int_{p_2^0}^1 \int_{p_1^0}^1  m^{(0, 0)}(x, v_1, v_2) h(v_1, v_2 | x) \mathrm{d}v_1 \mathrm{d}v_2 -\int_{p_2^0}^{p_2^1} \int_{p_1^0}^{p_1^1}  m^{(0, 0)}(x, v_1, v_2) h(v_1, v_2 | x) \mathrm{d}v_1 \mathrm{d}v_2,
\end{align*}
and
\begin{align*}
	\bE [Y_1^{(0, 0)} | X = x, V \in \mathcal{V}_{\text{mul}} ] \cdot \Pr[V \in \mathcal{V}_{\text{mul}} | X = x]
	= \int_{p_2^0}^{p_2^1} \int_{p_1^0}^{p_1^1}  m^{(0, 0)}(x, v_1, v_2) h(v_1, v_2 | x) \mathrm{d}v_1 \mathrm{d}v_2.
\end{align*}
Hence, we have
\begin{align*}
	\psi^{(0, 0)} (x, \mathbf{p} )
	&= \int_{p_2^0}^1 \int_{p_1^0}^1 m^{(0, 0)}(x, v_1, v_2) h(v_1, v_2 | x) \mathrm{d}v_1 \mathrm{d}v_2  - (1 - \lambda_x) \int_{p_2^0}^{p_2^1} \int_{p_1^0}^{p_1^1}  m^{(0, 0)}(x, v_1, v_2) h(v_1, v_2 | x) \mathrm{d}v_1 \mathrm{d}v_2.
\end{align*}
Partially differentiating both sides with respect to $p_1^0$ and $p_2^0$ and rearranging the equation, the Leibniz integral rule and continuity of $m^{(0, 0)}(x, \cdot, \cdot) h(\cdot, \cdot | x)$ lead to the desired result.
\qed

\subsection{Proofs of the theorems in Section \ref{sec:estimation}}\label{sec:proof:estimation}

\subsubsection{Preparation}\label{subsubsec:prep}

In the following, for a positive integer $a$, $\mathbf{I}_a$ denotes the $a \times a$ identity matrix.
For positive integers $a_1$ and $a_2$, $\mathbf{0}_{a_1 \times a_2}$ and $\mathbf{0}_{a_1}$ denote the $a_1 \times a_2$ zero matrix and the $a_1 \times 1$ zero vector, respectively.
When $A$ is a square matrix, we use $\chi_\mathrm{max}(A)$ and $\chi_\mathrm{min}(A)$ to denote its largest and smallest eigenvalues, respectively.
We denote a symmetric generalized inverse of a matrix $A$ by $A^{-}$.
For a general matrix $A$, $\|A\| = \sqrt{\mathrm{tr}\{A^\top A\}}$ denotes its Frobenius norm, where $\mathrm{tr}\{ \cdot \}$ is the trace, and $\|A\|_2 = \sqrt{\chi_{\mathrm{max}}(A^\top A)}$ denotes its spectral norm.

\begin{assumption}\label{as:data}
	$\{\{(Y_{ji}, D_{ji}, W_{ji})\}_{j=1}^2\}_{i=1}^n$ are independent and identically distributed across $i$.
\end{assumption} 

\begin{assumption}
	\hfil
	\label{as:1stage}
	\begin{enumerate}[(i)]
		\item $\|\hat \theta - \theta^*\| = O_P(n^{-1/2})$ and $\max_{1 \le i \le n} |\hat P_{ji}^d - P_{ji}^d| = O_P(n^{-1/2})$ for $j = 1, 2$ and $d = 0, 1$.
		\item $H_\rho(\cdot, \cdot)$ is Lipschitz continuous with respect to $ \rho$ in the neighborhood of $\rho^*$.
		\item For all $\rho$ in the neighborhood of $\rho^*$, $h_\rho(p_1, p_2)$ is continuous in $(p_1, p_2)\in [0,1]^2$, bounded away from zero uniformly in $(p_1, p_2) \in \mathcal{S}^{(1,0)}$, and Lipschitz continuous in $\rho$.
	\end{enumerate}
\end{assumption}

Assumption \ref{as:data} is a standard and relatively weak condition for microeconomic applications.
Assumption \ref{as:1stage}(i) is a high-level condition; however, it is standard for a parametric ML estimation under mild regularity conditions, including the non-singularity of the Fisher information matrix at $\theta^*$.\footnote{
	The $\sqrt{n}$-consistency in Assumption \ref{as:1stage}(i) is a natural consequence of the identifiability of $\theta^*$.
	By Theorem 1 of \citetmain{rothenberg1971identification}, under weak regularity conditions, the information matrix becomes non-singular at $\theta^*$ if and only if $\theta^*$ is locally identified.
	Thus, the identification result for $\theta^*$ in Appendix \ref{subsec:ident_wo_inf} yields the non-singularity of the information matrix at $\theta^*$, which in turn leads to the $\sqrt{n}$-consistency.
}
Note that although Theorem \ref{thm:gameiden} relies on the identification-at-infinity argument to prove the identification of our game model, the ML estimator is not directly based on this argument (see Appendix \ref{subsec:ident_wo_inf}).
In fact, as demonstrated in the numerical analysis in Section \ref{sec:MC}, the game parameters can be estimated at the parametric rate (see Table \ref{table:ML}).
Assumption \ref{as:1stage}(iii) ensures that $\hat h$ is uniformly bounded away from zero and infinity with probability approaching one (w.p.a.1) in conjunction with (i).

Define
\begin{align*}
	\begin{array}{ll}
	\mathbf{Y}_1^{(1, 0)}  = (\tilde I_1^{(1, 0)}Y_{11}, \dots, \tilde I_n^{(1, 0)}Y_{1n})^\top, & \hat{\mathbf{Y}}_1^{(1, 0)} = (\hat I_1^{(1, 0)}Y_{11}, \dots, \hat I_n^{(1, 0)}Y_{1n})^\top\\
	R_{K,i}^{(1,0)}  = (\tilde I_i^{(1,0)} X_{1i}^\top, \tilde T_i^{(1,0)} b_K (P_{1i}^0, P_{2i}^1)^\top)^\top, & \hat R_{K,i}^{(1,0)} = (\hat I_i^{(1,0)} X_{1i}^\top, \hat T_i^{(1,0)} b_K (\hat P_{1i}^0, \hat P_{2i}^1)^\top)^\top \\
	\mathbf{R}_K^{(1, 0)} = (R_{K,1}^{(1,0)}, \dots, R_{K,n}^{(1,0)})^\top, & \hat{\mathbf{R}}_K^{(1,0)} = (\hat R_{K,1}^{(1,0)}, \dots, \hat R_{K,n}^{(1,0)})^\top.
	\end{array}
\end{align*}

Then, our estimator of $\delta^{(1, 0)} = (\beta_1^{(1, 0)\top}, \alpha^{(1,0)\top})^\top$ is written as \eqref{eq:fcoef10}:
\begin{align}
	\hat \delta^{(1, 0)} = \left( \hat \beta_1^{(1, 0)\top},  \hat \alpha^{(1, 0)\top} \right)^\top
	&\coloneqq \left[ \hat{\mathbf{R}}_K^{(1, 0)\top} \hat{\mathbf{R}}_K^{(1, 0)} \right]^{-} \hat{\mathbf{R}}_K^{(1, 0)\top} \hat{\mathbf{Y}}_1^{(1, 0)}, \label{eq:fcoef10}\\
	\tilde \delta^{(1, 0)}   = \left( \tilde \beta_1^{(1, 0)\top},  \tilde \alpha^{(1, 0)\top} \right)^\top
	& \coloneqq \left[ \mathbf{R}_K^{(1, 0)\top} \mathbf{R}_K^{(1, 0)} \right]^{-} \mathbf{R}_K^{(1, 0)\top} \mathbf{Y}_1^{(1, 0)}, \label{eq:infcoef10}
\end{align}
while the one in \eqref{eq:infcoef10}  is an ``infeasible'' estimator with the true parameters in the first stage being treated as known.
The infeasible estimators of $g^{(1, 0)}(p_1^0, p_2^1)$, $\bE [U_1^{(1, 0)} | V_1 = p_1^0, V_2 = p_2^1] $, and $m^{(1, 0)}(x, p_1^0, p_2^1)$ can be defined similarly, which we denote by $\tilde g^{(1, 0)}(p_1^0, p_2^1)$, $\tilde \bE_n[U_1^{(1, 0)} | V_1 = p_1^0, V_2 = p_2^1] $, and $\tilde m^{(1, 0)}(x, p_1^0, p_2^1)$, respectively.
Furthermore, define $\Psi^{(1, 0)}_K \coloneqq \bE \left[ R_K^{(1,0)} R_K^{(1,0)\top} \right]$, $\Psi^{(1, 0)}_{nK} \coloneqq \mathbf{R}_{K}^{(1, 0) \top} \mathbf{R}_K^{(1, 0)}/n$, $\hat \Psi^{(1, 0)}_{nK} \coloneqq \hat{\mathbf{R}}_{K}^{(1, 0) \top} \hat{\mathbf{R}}_K^{(1, 0)}/n$, and $\Sigma^{(1, 0)}_K \coloneqq \bE \left[\left( \tilde e^{(1,0)} \right)^2 R_K^{(1,0)} R_K^{(1,0)\top} \right]$.

\begin{assumption}
	\hfil
	\label{as:eigen}
	\begin{enumerate}[(i)]
		\item There exist positive constants $\underline{c}_\Psi$ and $\bar{c}_\Psi$ such that $0 < \underline{c}_\Psi \leq \chi_\mathrm{min}\left( \Psi^{(1, 0)}_K \right) \leq \chi_\mathrm{max}\left( \Psi^{(1, 0)}_K \right) \leq \bar{c}_\Psi < \infty$ uniformly in $K$.
		\item There exist positive constants $\underline{c}_\Sigma$ and $\bar{c}_\Sigma$ such that $0 < \underline{c}_\Sigma \leq \chi_\mathrm{min}\left( \Sigma^{(1, 0)}_K \right) \leq \chi_\mathrm{max}\left( \Sigma^{(1, 0)}_K \right) \leq \bar{c}_\Sigma < \infty$ uniformly in $K$.
	\end{enumerate}
\end{assumption}

\begin{assumption}\label{as:emom}
	$\bE [ ( e^{(1,0)} )^4 | W, D ]$ is bounded. 
\end{assumption}

Assumption \ref{as:eigen}(i) ensures the existence of the inverse matrices $[ \Psi^{(1, 0)}_{nK} ]^{-1}$ and $[ \hat \Psi^{(1, 0)}_{nK} ]^{-1}$ w.p.a.1.
Assumption \ref{as:emom} is introduced to conveniently derive the limiting distribution of our estimator.
Note that under this assumption, $\bE [ ( \tilde e^{(1,0)} )^4 | W, D ]$ is also bounded by the definition of $\tau_\varpi$. 

To state the next assumption, we introduce the following notation.
For a sufficiently smooth function $g(p_1, p_2)$ and a vector of non-negative integers $\mathbf{a} = (a_1, a_2)$, let $\partial ^{\mathbf{a}} g(p_1, p_2) \coloneqq \partial^{|\mathbf{a}|}g(p_1, p_2) / (\partial ^{a_1}p_1 \partial ^{a_2}p_2)$, where $|\mathbf{a}| = a_1 + a_2$. If $|\mathbf{a}| = 0$, then $\partial ^{\mathbf{a}} g(p_1, p_2) = g(p_1, p_2)$.

\begin{assumption}\label{as:sieve}
	For some integer $s \ge 2$, the functions $g^{(1,0)}$ and $b_K$ are at least $s$-times continuously differentiable, and there exists a sequence of vectors $\alpha^{(1,0)} \in \mathbb{R}^K$ such that $\sup_{(p_1, p_2) \in [0, 1]^2} | \partial ^{\mathbf{a}} g^{(1,0)}(p_1, p_2) -  \partial ^{\mathbf{a}} b_K(p_1, p_2)^\top\alpha^{(1,0)} | = O( K^{(|\mathbf{a}| - s)/2} )$.
\end{assumption}
Assumption \ref{as:sieve} restricts the smoothness of $g^{(1,0)}$ and choice of the basis functions $b_K$. For example, Lemma 2 in \citetappendix{holland2017penalized} shows that when $g^{(1,0)}$ is $s$-times continuously differentiable on $[0,1]^2$, Assumption \ref{as:sieve} is satisfied by tensor-product B-splines of order $r$ (degree $r-1$) for $r -2 \geq s$.
For slightly more refined results, when $g^{(1,0)}$ is in a H\"{o}lder space of smoothness $s$, B-splines, wavelets, and Cattaneo and Farrell's local polynomial partitioning series can satisfy Assumption \ref{as:sieve} (for details, see \citealpappendix{chen2018optimal}, Corollary 3.1, and \citealpappendix{cattaneo2013optimal}, Lemma A.2).

\begin{assumption} \label{as:rate}
	As $n \to \infty$,
	(i) $\zeta_0(K)\sqrt{(\log K)/n} \to 0$, and
	(ii) $\zeta_1(K)/\sqrt{n} \to 0$,
	where $\zeta_d(K) \coloneqq \max_{|\mathbf{a}| \leq d} \sup_{(p_1, p_2) \in [0,1]^2 }\| \partial ^{\mathbf{a}} b_K(p_1, p_2) \|$.
\end{assumption}

Assumption \ref{as:rate}(i) is used to prove the convergence of the matrix $\Psi^{(1, 0)}_{nK}$ to $\Psi^{(1, 0)}_K$, and Assumption \ref{as:rate}(ii) is additionally introduced to ensure the convergence of $\hat \Psi^{(1, 0)}_{nK}$ to $\Psi^{(1, 0)}_K$. 
The bound of $\zeta_0(K)$ is well known for several basis functions, which is typically $\zeta_0(K) = O(\sqrt{K})$ (e.g., \citealpappendix{chen2007large}; \citealpappendix{belloni2015some}).
By contrast, there are fewer readily available results for the bound of $\zeta_1(K)$.
For example, Cattaneo and Farrell's local polynomial partitioning series satisfies $\zeta_1(K) = O(K)$ (see \citealpappendix{cattaneo2013optimal}, Lemma A.1). 
In Appendix \ref{subsec:zeta}, we show that the tensor-product B-spline also satisfies $\zeta_1(K) = O(K)$.

Next, consider a generic random vector $Q \in \text{supp}[Q] \subset \mathbb{R}^{\mathrm{dim}(Q)}$ with a finite dimension $\mathrm{dim}(Q)$.
Denote the set of uniformly bounded continuous functions on $\mathcal{D}$ as $\mathcal{C}(\mathcal{D})$.
We define the linear operator $\mathcal{P}_{nK}$ that maps a given function $g \in \mathcal{C}(\text{supp}[Q])$ to the sieve space defined by $b_K$ as follows:
\begin{align*}
\begin{split}
	\mathcal{P}_{nK}g
	= b_K(\cdot, \cdot)^\top \mathbb{S}_K \left[ \Psi^{(1, 0)}_{nK} \right]^{-1} \frac{1}{n} \sum_{i = 1}^n R_{K,i}^{(1,0)} g(Q_i),
\end{split}
\end{align*}
where $\mathbb{S}_K \coloneqq (\mathbf{0}_{K \times \mathrm{dim}(X)}, \mathbf{I}_K)$.
The operator norm of $\mathcal{P}_{nK}$ restricted on $\mathcal{S}^{(1,0)}$ is defined as
\begin{align*}
	\| \mathcal{P}_{nK} \|_{\infty} 
	\coloneqq \sup\left\{ \sup_{(p_1, p_2) \in \mathcal{S}^{(1,0)}} \left| \left( \mathcal{P}_{nK} g \right) (p_1, p_2) \right| : g \in \mathcal{C}(\text{supp}[Q]), \sup_{q \in \text{supp}[Q]} |g(q)| = 1 \right\}.
\end{align*}

\begin{assumption} \label{as:opnorm}
	(i) $\| \mathcal{P}_{nK} \|_{\infty} = O_P(1)$. 
	(ii) $\zeta_0(K) \zeta_1(K) / \sqrt{n} = O(1)$.
\end{assumption}

\begin{assumption}\label{as:deriv}
	$\sup_{(p_1, p_2) \in \mathcal{S}^{(1,0)}}\left| \partial^{\mathbf{a}} b_K(p_1, p_2)^\top \alpha \right| = O(K^{|\mathbf{a}|/2}) \sup_{(p_1, p_2) \in \mathcal{S}^{(1,0)}} \left|  b_K(p_1, p_2)^\top \alpha \right|$ for any $\alpha \in \mathbb{R}^K$. 
\end{assumption}

Assumption \ref{as:opnorm}(i) is a stability condition for the projection operator $ \mathcal{P}_{nK}$.
\citetappendix{huang2003local} shows that a similar condition holds true for spline bases under some mild regularity conditions.
In addition, wavelets can satisfy such a condition, as shown in Theorem 5.2 in \citetappendix{chen2015optimal}.
For a direct verification of Assumption \ref{as:opnorm}(i) under a particular choice of basis functions, see Appendix \ref{subsec:opnorm}.
Assumption \ref{as:opnorm}(ii) is used to show that the operator norm of the ``feasible'' version of $\mathcal{P}_{nK}$ is also bounded in probability (see Lemma \ref{lem:opnorm}).
In the proof of Corollary 3.1 of \citetappendix{chen2018optimal}, it is shown that Assumption \ref{as:deriv} holds for splines and wavelets.

\begin{assumption}
	\hfil
	\label{as:rate2}
	\begin{enumerate}[(i)]
		\item There exist finite constants $c_b > 0$ and $\omega \geq 0$ such that $\| b_K(p_1, p_2) -  b_K(p'_1, p'_2)\| \leq c_b K^\omega \| (p_1,p_2)  - (p'_1, p'_2) \| $ for all $(p_1, p_2), (p_1', p_2') \in [0,1]^2$.
		\item $\zeta_0^2(K)\sqrt{(\log n) /n} \to 0$.
	\end{enumerate}
\end{assumption}

\paragraph{The estimation of $m^{(0, 0)}(x, p_1^0, p_2^1)$.}

Similarly as above, define $\hat{\mathbf{Y}}_1^{(0,0)} = (\hat I_1^{(0,0)} Y_{11}, \dots, \hat I_n^{(0,0)} Y_{1n})^\top$ and $\hat{\mathbf{R}}_K^{(0, 0)} = (\hat R_{K,1}^{(0,0)}, \dots, \hat R_{K,n}^{(0,0)})^\top$, where
\begin{align*}
	\hat R_{K,i}^{(0,0)} 
	\coloneqq (\hat I_i^{(0,0)} X_{1i}^\top, & \; \hat T_i^{(0,0)} \hat \lambda_{X_i} b_K(\hat P_{1i}^0, \hat P_{2i}^0)^\top, \; \hat T_i^{(0,0)} (1 - \hat \lambda_{X_i}) b_K(\hat P_{1i}^1,\hat P_{2i}^0)^\top, \\
	& \; \hat T_i^{(0,0)} (1 - \hat \lambda_{X_i}) b_K(\hat P_{1i}^0, \hat P_{2i}^1)^\top, \; \hat T_i^{(0,0)} (1 - \hat \lambda_{X_i}) b_K(\hat P_{1i}^1, \hat P_{2i}^1)^\top)^\top.
\end{align*}
Then, the estimator of $\delta^{(0, 0)} = ( \beta_1^{(0, 0)\top}, \alpha_1^{(0,0)\top}, \ldots , \alpha_4^{(0,0)\top} )^\top$ can be written as
\begin{align}\label{eq:coef00}
	\hat \delta^{(0, 0)} = \left( \hat \beta_1^{(0, 0)\top},  \hat \alpha_1^{(0, 0)\top}, \ldots, \hat \alpha_4^{(0, 0)\top} \right)^\top
	&\coloneqq \left[ \hat{\mathbf{R}}_K^{(0, 0)\top} \hat{\mathbf{R}}_K^{(0, 0)} \right]^{-} \hat{\mathbf{R}}_K^{(0, 0)\top} \hat{\mathbf{Y}}_1^{(0, 0)}.
\end{align}
Then, the estimator of $m^{(0, 0)}(x, p_1^0, p_2^1)$ can be obtained by \eqref{eq:m00}.
The standard deviation for this MTR estimator is given by
\begin{align}\label{eq:sd00}
	\sigma_{K}^{(0,0)}(p_1^0, p_2^1) &\coloneqq \frac{\sqrt{ \ddot{b}_K(p_1^0, p_2^1)^\top \mathbb{S}_{3K} \left[ \Psi_K^{(0,0)}\right]^{-1} \Sigma^{(0, 0)}_K  \left[ \Psi_K^{(0,0)}\right]^{-1} \mathbb{S}_{3K}^\top \ddot{b}_K(p_1^0, p_2^1)}}{h(p_1^0, p_2^1)},
\end{align}
where $\mathbb{S}_{3K} \coloneqq (\mathbf{0}_{K \times \mathrm{dim}(X)}, \mathbf{0}_{K \times 2K}, \mathbf{I}_K,  \mathbf{0}_{K \times K})$.
The definitions of the matrices $\Psi_K^{(0,0)}$ and $\Sigma_K^{(0,0)}$ should be clear from the context.

\paragraph{Additional notaions.}

For notational simplicity, when there is no confusion, we often suppress the superscript $(1,0)$ in the rest of this section.
Here, we introduce additional notations as follows.
Let $\mathbb{S}_a$ be a selection matrix of dimension $a \times (\mathrm{dim}(X) + K)$ such that $\mathbb{S}_a \delta^{(1,0)}$ is the corresponding $a \times 1$ subvector of $\delta^{(1,0)}$.
Specifically, we write $\mathbb{S}_X \delta^{(1, 0)} = \beta_1^{(1, 0)}$ and $\mathbb{S}_K \delta^{(1, 0)} = \alpha^{(1, 0)}$ with $\mathbb{S}_X = \mathbb{S}_{\mathrm{dim}(X)} = (\mathbf{I}_{\mathrm{dim}(X)}, \mathbf{0}_{\mathrm{dim}(X) \times K})$ and $\mathbb{S}_K = (\mathbf{0}_{K \times \mathrm{dim}(X)}, \mathbf{I}_K)$.

In addition, we introduce the following notations:
\begin{align*}
\begin{array}{ll}
	\tilde{\mathbf{e}}^{(1, 0)} = (\tilde e_1^{(1, 0)}, \dots, \tilde e_n^{(1, 0)})^\top 
	& \hat{\mathbf{e}}^{(1, 0)} = (\hat e_1^{(1, 0)}, \dots, \hat e_n^{(1, 0)})^\top, \;\; \text{where} \;\; \hat e^{(1, 0)} \coloneqq \tau_\varpi(\hat{\mathcal{L}}^{(1, 0)}) e^{(1, 0)}\\
	\mathbf{g}^{(1, 0)} = (g^{(1, 0)}( P_{11}^0, P_{21}^1 ), \dots, g^{(1,0)} ( P_{1n}^0, P_{2n}^1) )^\top
	& \hat{\mathbf{g}}^{(1, 0)} = ( g^{(1, 0)}( \hat P_{11}^0, \hat P_{21}^1 ), \dots, g^{(1,0)} ( \hat P_{1n}^0, \hat P_{2n}^1) )^\top\\
	\mathbf{X}_1^{(1, 0)} = (\tilde I_1^{(1, 0)} X_{11}, \dots, \tilde I_n^{(1, 0)} X_{1n})^\top
	& \hat{\mathbf{X}}_1^{(1, 0)} = (\hat I_1^{(1, 0)} X_{11}, \dots, \hat I_n^{(1, 0)} X_{1n})^\top\\
	\tilde{\mathbf{T}}^{(1,0)} = \text{diag}(\tilde T_1^{(1,0)}, \dots, \tilde T_n^{(1,0)})
	& \hat{\mathbf{T}}^{(1,0)} = \text{diag}(\hat T_1^{(1,0)}, \dots, \hat T_n^{(1,0)})\\
	\multicolumn{2}{l}{\breve{\mathbf{T}}^{(1,0)} = \text{diag}(\breve T_1^{(1,0)}, \dots, \breve T_n^{(1,0)}), \;\; \text{where} \;\; \breve T^{(1, 0)} \coloneqq \tau_\varpi(\hat{\mathcal{L}}^{(1, 0)}) T^{(1,0)}} \\
	\mathbf{b}_K^{(1, 0)} = (b_K(P_{11}^0, P_{21}^1), \dots, b_K(P_{1n}^0, P_{2n}^1))^\top
	& \hat{\mathbf{b}}_K^{(1, 0)} = (b_K(\hat P_{11}^0, \hat P_{21}^1), \dots, b_K(\hat P_{1n}^0, \hat P_{2n}^1))^\top\\
	\mathbf{u}^{(1,0)} = \mathbf{g}^{(1, 0)} - \mathbf{b}_K^{(1, 0)} \alpha^{(1,0)} 
	& \hat{\mathbf{u}}^{(1,0)} = \hat{\mathbf{g}}^{(1, 0)} - \hat{\mathbf{b}}_K^{(1, 0)} \alpha^{(1,0)}
\end{array}
\end{align*}
By definition, $\mathbf{R}_K^{(1, 0)} = (\mathbf{X}_1^{(1, 0)} , \tilde{\mathbf{T}}^{(1,0)} \mathbf{b}_K^{(1, 0)})$ and $\hat{\mathbf{R}}_K^{(1, 0)} = (\hat{\mathbf{X}}_1^{(1, 0)}, \hat{\mathbf{T}}^{(1,0)} \hat{\mathbf{b}}_K^{(1, 0)})$.
By \eqref{eq:plm2}, the infeasible estimator $\mathbb{S}_a \tilde \delta^{(1, 0)}$ defined in \eqref{eq:infcoef10} can be decomposed as follows:
\begin{align}\label{eq:infdecomp1}
\begin{split}
	\mathbb{S}_a \left( \tilde \delta^{(1, 0)} - \delta^{(1, 0)} \right) 
	& = \mathbb{S}_a \Psi_{nK}^{-1} \mathbf{R}_K^\top \tilde{\mathbf{T}} \mathbf{u} / n + \mathbb{S}_a \Psi_{nK}^{-1} \mathbf{R}_K^\top \tilde{\mathbf{e}} / n.
\end{split}
\end{align}
Next, by \eqref{eq:plm1}, we can write
\begin{align*}
	\hat I^{(1,0)} Y_1
	& = \tau_\varpi(\hat{\mathcal{L}}^{(1, 0)}) \left( I^{(1, 0)} X_{1}^\top \beta_1^{(1, 0)} + T^{(1, 0)} g^{(1, 0)}(P_1^0, P_2^1) + e^{(1, 0)} \right) \\
	& = \hat R_{K}^{(1,0)\top} \delta^{(1,0)} + \breve T^{(1,0)} \left( g^{(1, 0)}(P_1^0, P_2^1) - g^{(1, 0)}(\hat P_1^0, \hat P_2^1) \right) \\
	& \quad + \left( \breve T^{(1, 0)} - \hat T^{(1,0)} \right) g^{(1, 0)}(\hat P_1^0, \hat P_2^1) + \hat T^{(1,0)} \left( g^{(1, 0)}(\hat P_1^0, \hat P_2^1) - b_K(\hat P_1^0, \hat P_2^1)^\top \alpha^{(1,0)} \right) + \hat e^{(1, 0)}.
\end{align*}
Thus, the feasible estimator $\mathbb{S}_a \hat \delta^{(1, 0)}$ defined in \eqref{eq:fcoef10} can be decomposed as follows:
\begin{align}\label{eq:fdecomp1}
\begin{split}
	\mathbb{S}_a \left( \hat \delta^{(1, 0)} - \delta^{(1, 0)} \right)  
	&= \mathbb{S}_a \hat \Psi_{nK}^{-1} \hat{\mathbf{R}}_K^\top \breve{\mathbf{T}} \left(\mathbf{g}^{(1, 0)} - \hat{\mathbf{g}}^{(1, 0)} \right) / n + \mathbb{S}_a \hat \Psi_{nK}^{-1} \hat{\mathbf{R}}_K^\top \left( \breve{\mathbf{T}} - \hat{\mathbf{T}} \right) \hat{\mathbf{g}}^{(1, 0)} / n\\
	& \quad + \mathbb{S}_a \hat \Psi_{nK}^{-1} \hat{\mathbf{R}}_K^\top \hat{\mathbf{T}} \hat{\mathbf{u}} / n + \mathbb{S}_a \hat \Psi_{nK}^{-1} \hat{\mathbf{R}}_K^\top \hat{\mathbf{e}} / n.
\end{split}
\end{align}

\subsubsection{Proof of Theorem \ref{thm:normal1}}

Let
\begin{align}\label{eq:sd10}
	\sigma_K^{(1,0)}(p_1^0, p_2^1) \coloneqq \frac{\sqrt{ \ddot{b}_K(p_1^0, p_2^1)^\top \mathbb{S}_K \left[ \Psi_K^{(1,0)} \right]^{-1} \Sigma^{(1, 0)}_K \left[ \Psi_K^{(1,0)} \right]^{-1} \mathbb{S}_K^\top \ddot{b}_K(p_1^0, p_2^1)}}{h(p_1^0, p_2^1)}.
\end{align}
We first show that there exists a constant $0 < c_\sigma < \infty$ such that for a given $(p_1^0, p_2^1) \in \mathcal{S}^{(1,0)}$,
\begin{align}\label{eq:sigmalbound}
	\sigma_K^{(1,0)}(p_1^0, p_2^1) \geq c_\sigma \cdot \| \ddot{b}_K(p_1^0, p_2^1) \|.
\end{align}
By Assumptions \ref{as:1stage}(iii) and \ref{as:eigen}, noting that $\mathbb{S}_K \mathbb{S}_K^\top = \mathbf{I}_K$, we observe
\begin{align*}
	\left( \sigma_K^{(1,0)}(p_1^0, p_2^1) \right)^2 
	& =  \frac{1}{h^2(p_1^0, p_2^1)} \ddot{b}_K(p_1^0, p_2^1)^\top \mathbb{S}_K \Psi_K^{-1} \Sigma_K \Psi_K^{-1}\mathbb{S}_K^\top \ddot{b}_K(p_1^0, p_2^1) \\
	& \geq \underbrace{\frac{\underline{c}_\Sigma}{\bar{c}^2_\Psi \cdot h^2(p_1^0, p_2^1)}}_{> 0} \cdot \| \ddot{b}_K(p_1^0, p_2^1) \|^2,
\end{align*}
which implies \eqref{eq:sigmalbound}.
Hence, $K / \sigma_K^{(1,0)}(p_1^0, p_2^1) \to 0$ under the assumption $K / \| \ddot{b}_K(p_1^0, p_2^1) \| \to 0$.

\paragraph{Proof of (i).}

By the definition of the infeasible estimator $\tilde m^{(1, 0)}(x, p_1^0, p_2^1)$, we have
\begin{align*}
	& \tilde m^{(1, 0)}(x, p_1^0, p_2^1)  -  m^{(1, 0)}(x, p_1^0, p_2^1) \\
	& = x_1^\top \left( \tilde \beta_1^{(1, 0)} - \beta_1^{(1, 0)}\right) - \frac{1}{h(p_1^0, p_2^1)} \ddot{b}_K(p_1^0, p_2^1)^\top \left( \tilde \alpha^{(1, 0)} - \alpha^{(1, 0)} \right) +  \frac{1}{h(p_1^0, p_2^1)}\left( \partial_{p_1^0 p_2^1}  [g^{(1, 0)}(p_1^0, p_2^1)] - \ddot{b}_K(p_1^0, p_2^1)^\top \alpha^{(1, 0)} \right) \\
	& = - \frac{1}{h(p_1^0, p_2^1)} \ddot{b}_K(p_1^0, p_2^1)^\top \left( \tilde \alpha^{(1, 0)} - \alpha^{(1, 0)} \right) +  O_P(n^{-1/2}) + O(K^{(2 - s)/2}) \\
	& = A_{1n} + A_{2n} + O_P(n^{-1/2}),
\end{align*}
by Lemma \ref{lem:parametric}(i), and Assumption \ref{as:sieve}, where
\begin{align*}
	A_{1n} \coloneqq - \frac{1}{h(p_1^0, p_2^1)} \ddot{b}_K(p_1^0, p_2^1)^\top \mathbb{S}_K \Psi_{nK}^{-1} \mathbf{R}_K^\top \tilde{\mathbf{T}} \mathbf{u} / n, 
	\quad 
	A_{2n} \coloneqq - \frac{1}{h(p_1^0, p_2^1)} \ddot{b}_K(p_1^0, p_2^1)^\top \mathbb{S}_K \Psi_{nK}^{-1} \mathbf{R}_K^\top \tilde{\mathbf{e}}/n.
\end{align*}
By Assumptions \ref{as:1stage}(iii), \ref{as:sieve}, \ref{as:opnorm}, and \ref{as:deriv}, we have
\begin{align}\label{eq:A1n}
\begin{split}
	|A_{1n} | 
	\leq O(K) \cdot \|\mathcal{P}_{nK}\|_\infty \cdot O(K^{-s/2})
	= O_P(K^{(2-s)/2}).
\end{split}
\end{align}
Define
\begin{align*}
	A'_{2n} & \coloneqq - \frac{1}{h(p_1^0, p_2^1)} \ddot{b}_K(p_1^0, p_2^1)^\top \mathbb{S}_K \Psi_K^{-1} \mathbf{R}_K^\top \tilde{\mathbf{e}} / n.
\end{align*}
It is easy to see that $| A_{2n} - A'_{2n} | \le O(1) \cdot \| \ddot{b}_K(p_1^0, p_2^1) \mathbb{S}_K (\Psi_{nK}^{-1} - \Psi_{K}^{-1}) \mathbf{R}_K^\top \tilde{\mathbf{e}} / n \| = \| \ddot{b}_K(p_1^0, p_2^1) \| \cdot  o_P(n^{-1/2})$ by \eqref{eq:eigenPsi}, Lemma \ref{lem:matLLN}(iii) and Assumption \ref{as:1stage}(iii).
Thus, by \eqref{eq:sigmalbound}, we obtain
\begin{align*}
	\frac{\sqrt{n} \left( \tilde m^{(1, 0)}(x, p_1^0, p_2^1)  -  m^{(1, 0)}(x, p_1^0, p_2^1) \right)}{\sigma_K^{(1,0)}(p_1^0, p_2^1)} 
	=  \frac{\sqrt{n} ( A_{1n} + A_{2n})}{\sigma_K^{(1,0)}(p_1^0, p_2^1) } + o_P(1)
	=  \frac{\sqrt{n} A'_{2n}}{\sigma_K^{(1,0)}(p_1^0, p_2^1) } + o_P(1).
\end{align*}

We now show the asymptotic normality of $\sqrt{n} A'_{2n} / \sigma_K^{(1,0)}(p_1^0, p_2^1)$.
Let $\xi_i \coloneqq - \Pi_K(p_1^0, p_2^1) R_{K,i} \tilde e_i / \sqrt{n}$, where
\begin{align*}
	\Pi_K(p_1^0, p_2^1) \coloneqq \left[ \sigma_K^{(1,0)}(p_1^0, p_2^1) \cdot h(p_1^0, p_2^1) \right]^{-1}\ddot{b}_K(p_1^0, p_2^1)^\top  \mathbb{S}_K \Psi_K^{-1},
\end{align*}
so that $\sum_{i = 1}^n \xi_i =  \sqrt{n} A'_{2n} / \sigma_K^{(1,0)}(p_1^0, p_2^1)$.
By construction, $\bE [\xi_i] = 0$ and $\bE [\xi_i^2] = n^{-1}$ hold.
Then, Assumption \ref{as:emom} and the law of iterated expectations yield
\begin{align*}
	\bE \left[ \xi_i^4 \right]
	& = n^{-2} \bE \left[ \Pi_K(p_1^0, p_2^1) R_{K,i} R_{K,i}^\top \Pi_K(p_1^0, p_2^1)^\top \Pi_K(p_1^0, p_2^1) R_{K,i} 	R_{K,i}^\top \Pi_K(p_1^0, p_2^1)^\top \bE \left[ \tilde e_i^4 \middle| W_i, D_i \right] \right] \\
	& \leq O(n^{-2} ) \cdot \bE \left[ \Pi_K(p_1^0, p_2^1) R_{K,i} R_{K,i}^\top \Pi_K(p_1^0, p_2^1)^\top \Pi_K(p_1^0, p_2^1) R_{K,i} R_{K,i}^\top \Pi_K(p_1^0, p_2^1)^\top \right] \\
		& \leq O(n^{-2} ) \cdot \chi_\mathrm{max}\left(\Pi_K(p_1^0, p_2^1)^\top \Pi_K(p_1^0, p_2^1)\right) \bE \left[ \mathrm{tr}\left\{ R_{K,i} R_{K,i}^\top \Pi_K(p_1^0, p_2^1)^\top \Pi_K(p_1^0, p_2^1) R_{K,i} R_{K,i}^\top \right\} \right] \\
	& \leq O(n^{-2} ) \cdot \left[ \chi_\mathrm{max}\left(\Pi_K(p_1^0, p_2^1)^\top \Pi_K(p_1^0, p_2^1)\right) \right]^2 \bE \left[ \mathrm{tr}\left\{ R_{K,i} R_{K,i}^\top R_{K,i} R_{K,i}^\top \right\} \right] \\
	& = O(\zeta^2_0(K) K / n^2 ) \cdot \left[ \chi_\mathrm{max}\left(\Pi_K(p_1^0, p_2^1)^\top \Pi_K(p_1^0, p_2^1)\right) \right]^2,
\end{align*}
where the last equality follows from $\bE [\mathrm{tr}\{ R_{K,i} R_{K,i}^\top R_{K,i} R_{K,i}^\top \}] \le \zeta^2_0(K) \mathrm{tr}\{ \bE [ R_{K,i} R_{K,i}^\top ] \} = O(\zeta^2_0(K) K)$ under Assumption \ref{as:eigen}(i).
Since \eqref{eq:sigmalbound} and Assumption \ref{as:1stage}(iii) imply that
\begin{align*}
	\chi_\mathrm{max}\left(\Pi_K(p_1^0, p_2^1)^\top \Pi_K(p_1^0, p_2^1)\right) 
	& \leq \mathrm{tr}\left\{ \Pi_K(p_1^0, p_2^1)^\top \Pi_K(p_1^0, p_2^1) \right\} \\
	& \leq O(1)\cdot \| \ddot{b}_K(p_1^0, p_2^1) \|^{-2} \cdot \ddot{b}_K(p_1^0, p_2^1)^\top  \mathbb{S}_K \Psi_K^{-2} \mathbb{S}_K^\top \ddot{b}_K(p_1^0, p_2^1) = O(1),
\end{align*}
we have $\sum_{i=1}^n \bE \left[ \xi_i^4 \right] = O(\zeta^2_0(K) K / n ) = o(1)$.
Hence, result (i) follows from Lyapunov's central limit theorem.

\paragraph{Proof of  (ii).}
From the definition of the feasible estimator $\hat m^{(1, 0)}(x, p_1^0, p_2^1)$, we have
\begin{align*}
	&  \hat m^{(1, 0)}(x, p_1^0, p_2^1)  -  m^{(1, 0)}(x, p_1^0, p_2^1) \\
	& = x_1^\top \left( \hat \beta_1^{(1, 0)} - \beta_1^{(1, 0)}\right) - \frac{1}{\hat h(p_1^0, p_2^1)} \ddot{b}_K(p_1^0, p_2^1)^\top \left( \hat \alpha^{(1, 0)} - \alpha^{(1, 0)} \right) \\
	& \quad + \frac{1}{\hat h(p_1^0, p_2^1)}\left( \partial_{p_1^0 p_2^1}  [g^{(1, 0)}(p_1^0, p_2^1)] - \ddot{b}_K(p_1^0, p_2^1)^\top \alpha^{(1, 0)} \right) + \left( \frac{1}{ h(p_1^0, p_2^1)} - \frac{1}{\hat h(p_1^0, p_2^1)} \right)  \partial_{p_1^0 p_2^1}  [g^{(1, 0)}(p_1^0, p_2^1)]  \\
	& = - \frac{1}{\hat h(p_1^0, p_2^1)} \ddot{b}_K(p_1^0, p_2^1)^\top \left( \hat \alpha^{(1, 0)} - \alpha^{(1, 0)} \right) +  O_P(n^{-1/2})  + O_P(K^{(2 - s)/2}) \\
	& = C_{1n} + C_{2n} + C_{3n} + C_{4n} + O_P(n^{-1/2}),
\end{align*}
by \eqref{eq:hdif}, Lemma \ref{lem:parametric}(iii), and Assumption \ref{as:sieve}, where
\begin{align*}
	C_{1n} & \coloneqq - \frac{1}{\hat h(p_1^0, p_2^1)} \ddot{b}_K(p_1^0, p_2^1)^\top  \mathbb{S}_K \hat \Psi_{nK}^{-1} \hat{\mathbf{R}}_K^\top \breve{\mathbf{T}} \left( \mathbf{g}^{(1,0)} - \hat{\mathbf{g}}^{(1,0)}\right)/n, \\
	C_{2n} & \coloneqq  - \frac{1}{\hat h(p_1^0, p_2^1)} \ddot{b}_K(p_1^0, p_2^1)^\top  \mathbb{S}_K \hat \Psi_{nK}^{-1} \hat{\mathbf{R}}_K^\top \left(\breve{\mathbf{T}} - \hat{\mathbf{T}} \right) \hat{\mathbf{g}}^{(1,0)} / n, \\
	C_{3n} & \coloneqq - \frac{1}{\hat h(p_1^0, p_2^1)} \ddot{b}_K(p_1^0, p_2^1)^\top  \mathbb{S}_K \hat \Psi_{nK}^{-1} \hat{\mathbf{R}}_K^\top \hat{\mathbf{T}} \hat{\mathbf{u}}/n, \\
	C_{4n} & \coloneqq - \frac{1}{\hat h(p_1^0, p_2^1)} \ddot{b}_K(p_1^0, p_2^1)^\top  \mathbb{S}_K \hat \Psi_{nK}^{-1} \hat{\mathbf{R}}_K^\top \hat{\mathbf{e}}/n.
\end{align*}
The fact that $1/\hat h(p_1^0, p_2^1) = O_P(1)$ and Assumption \ref{as:deriv} imply that 
\begin{align*}
	|C_{1n}|
	&\le O_P(K) \cdot \sup_{(p_1,p_2) \in \mathcal{S}^{(1,0)}} \left| b_K(p_1, p_2)^\top  \mathbb{S}_K \hat \Psi_{nK}^{-1} \hat{\mathbf{R}}_K^\top \breve{\mathbf{T}} \left( \mathbf{g}^{(1,0)} - \hat{\mathbf{g}}^{(1,0)}\right)/n \right| = O_P(K/\sqrt{n}),
\end{align*}
where the last equality follows from \eqref{eq:supgdif}.
Analogously, we can observe that $|C_{2n}| = O_P(K/\sqrt{n})$.
In addition, the same argument as in \eqref{eq:A1n} implies that $| C_{3n} | = O_P(K^{(2-s)/2})$ by Lemma \ref{lem:opnorm}.
Further, 
\begin{align*}
	C_{4n} 
	&= - \frac{1}{h(p_1^0, p_2^1)} \ddot{b}_K(p_1^0, p_2^1)^\top \mathbb{S}_K \Psi_{nK}^{-1} \mathbf{R}_K^\top \tilde{\mathbf{e}}/n \\
	& \quad - \frac{1}{\hat h(p_1^0, p_2^1)} \left( \ddot{b}_K(p_1^0, p_2^1)^\top \mathbb{S}_K \hat{\Psi}_{nK}^{-1} \hat{\mathbf{R}}_K^\top \hat{\mathbf{e}}/n - \ddot{b}_K(p_1^0, p_2^1)^\top  \mathbb{S}_K \Psi_{nK}^{-1} \mathbf{R}_K^\top \hat{\mathbf{e}}/n \right)  \\
	& \quad - \frac{1}{\hat h(p_1^0, p_2^1)} \ddot{b}_K(p_1^0, p_2^1)^\top  \mathbb{S}_K \Psi_{nK}^{-1} \mathbf{R}_K^\top \left( \hat{\mathbf{e}} - \tilde{\mathbf{e}} \right) / n - \left( \frac{1}{\hat h(p_1^0, p_2^1)} - \frac{1}{h(p_1^0, p_2^1)} \right) \ddot{b}_K(p_1^0, p_2^1)^\top \mathbb{S}_K \Psi_{nK}^{-1} \mathbf{R}_K^\top \tilde{\mathbf{e}}/n \\
	&= \underbrace{- \frac{1}{ h(p_1^0, p_2^1)} \ddot{b}_K(p_1^0, p_2^1)^\top \mathbb{S}_K \Psi_{nK}^{-1} \mathbf{R}_K^\top \tilde{\mathbf{e}}/n}_{= A_{2n}} + \| \ddot{b}_K(p_1^0, p_2^1) \| \cdot o_P(n^{-1/2}) + \| \ddot{b}_K(p_1^0, p_2^1) \| \cdot O_P(n^{-1}),
\end{align*}
by Lemma \ref{lem:matLLN}(iii), \eqref{eq:meanval1}, $\max_{1 \le i \le n} |\hat{\mathcal{L}}_i^{(1, 0)} - \mathcal{L}_i^{(1, 0)}| = O_P(n^{-1/2})$, and \eqref{eq:hdif}.

Summarizing these results, we obtain
\begin{align*}
	\frac{\sqrt{n} \left( \hat m^{(1, 0)}(x, p_1^0, p_2^1)  -  m^{(1, 0)}(x, p_1^0, p_2^1) \right)}{\sigma_K^{(1,0)}(p_1^0, p_2^1)} 
	& =  \frac{\sqrt{n} ( C_{1n} + C_{2n} + C_{3n} + C_{4n})}{\sigma_K^{(1,0)}(p_1^0, p_2^1) } + o_P(1) \\
	& =  \frac{\sqrt{n} A_{2n}}{\sigma_K^{(1,0)}(p_1^0, p_2^1) } + o_P(1),
\end{align*}
by \eqref{eq:sigmalbound}.
Then, the remaining part of the proof follows by the same argument as in result (i). \qed

\subsection{Lemmas}\label{subsec:lemmas}

\begin{lemma}\label{lem:matLLN}
	Suppose that Assumptions \ref{as:gamemodel}, \ref{as:linear}, \ref{as:data}, \ref{as:1stage}, \ref{as:eigen}(i), and \ref{as:rate} hold.
	Then, we have 
	\begin{align*}
	\renewcommand{\arraystretch}{2}
	\begin{array}{cl}
	\text{(i)}  & \left\| \Psi_{nK}^{(1,0)}  - \Psi_K^{(1,0)} \right\|_2 = O_P (\zeta_0(K)\sqrt{( \log K)/ n}) = o_P(1), \\
	\text{(ii)} & \left\| \hat \Psi^{(1, 0)}_{nK}  - \Psi_K^{(1,0)} \right\|_2 = O_P (\zeta_0(K)\sqrt{( \log K)/ n}) + O_P(\zeta_1(K)/ \sqrt{n})  = o_P(1), \\
	\text{(iii)} & \left\| \left[ \Psi_{nK}^{(1,0)}  \right]^{-1} - \left[ \Psi_K^{(1,0)} \right]^{-1} \right\|_2 = O_P (\zeta_0(K)\sqrt{( \log K)/ n})  = o_P(1), \\
	\text{(iv)} & \left\| \left[\hat \Psi_{nK}^{(1,0)}\right]^{-1} - \left[ \Psi_K^{(1,0)} \right]^{-1}  \right\|_2 = O_P (\zeta_0(K)\sqrt{( \log K)/ n}) + O_P(\zeta_1(K)/ \sqrt{n})  = o_P(1).\\
	\end{array}
	\renewcommand{\arraystretch}{1}
	\end{align*}
\end{lemma}

\begin{proof}
	(i) Let $\Xi_i \coloneqq (R_{K,i}^{(1,0)} R_{K,i}^{(1,0)\top} - \bE [R_{K,i}^{(1,0)} R_{K,i}^{(1,0)\top}])/n$.
	By Bernstein's inequality for random matrices (\citealpappendix{tropp2012user}, Theorem 1.6), it holds that
	\begin{align*}
	\Pr \left( \left\|\sum_{i=1}^n \Xi_i  \right\|_2 \geq t_n \right) \leq \exp\left( \log(2\mathrm{dim}(X) + 2K) + \frac{-t_n^2 / 2}{\sigma_n^2 + t_n \cdot r_n/3}\right), 
	\end{align*}
	for any $t_n \ge 0$, where $r_n$ is any non-negative value such that $\max_{1 \le i \le n}\|\Xi_i\|_2 \le r_n$, and $\sigma^2_n \coloneqq \|\sum_{i=1}^n \bE [\Xi_i \Xi_i^\top] \|_2$.
	Since $\tilde T_i^{(1,0)}$ is bounded, we obtain $r_n = O(\zeta_0^2(K)/n)$ and similarly $\sigma^2_n = O(\zeta_0^2(K)/n)$ by Assumption \ref{as:eigen}(i).
	The rest of the proof follows immediately from Corollary 4.1 of \citeappendix{chen2015optimal}.
	\bigskip
	
	(ii) The triangle inequality leads to
	\begin{align*}
		\| \hat \Psi_{nK} - \Psi_K \|_2 \leq \| \hat \Psi_{nK} - \Psi_{nK} \|_2  + \| \Psi_{nK} - \Psi_K \|_2.
	\end{align*}
	The second term is $O_P (\zeta_0(K)\sqrt{( \log K)/ n})$ by (i).
	For the first term, the triangle inequality implies
	\begin{align*}
		\| \hat \Psi_{nK} - \Psi_{nK} \|_2 
		&	\leq \| \hat \Psi_{nK} - \Psi_{nK} \| \\
		& \leq \left\| \left( \hat{\mathbf{R}}_K^\top - \mathbf{R}_K^\top \right) \left( \hat{\mathbf{R}}_K - \mathbf{R}_K \right) /n \right\| + 2 \left\| \mathbf{R}_K^\top \left( \hat{\mathbf{R}}_K - \mathbf{R}_K \right) /n \right\|.
	\end{align*}
	By the mean value theorem, Assumptions \ref{as:1stage}(i)--(iii) yield
	\begin{align}\label{eq:Tdif}
	\begin{split}
		\hat T_i^{(1,0)} - \tilde T_i^{(1,0)}
		& = \mathbf{1}\{D_i = (1, 0)\} \left( \frac{\tau_\varpi(\hat{\mathcal{L}}_i^{(1,0)})}{\hat{\mathcal{L}}_i^{(1,0)}} - \frac{\tau_\varpi(\mathcal{L}_i^{(1,0)})}{\mathcal{L}_i^{(1,0)}} \right) \\
		& = O(1) \cdot (\hat{\mathcal{L}}_i^{(1,0)} - \mathcal{L}_i^{(1,0)}) \cdot \frac{\mathbf{1}\{\bar{\mathcal{L}}_i^{(1,0)} \ge \varpi \} }{\left(\bar{\mathcal{L}}_i^{(1,0)}\right)^2} \\
		& = O(1) \cdot \left( \hat P_{1i}^0 - P_{1i}^0 - H_{\hat \rho}(\hat P_{1i}^0, \hat P_{2i}^1) + H_{\rho^*}(P_{1i}^0, P_{2i}^1) \right) 
		= O_P(n^{-1/2}),
	\end{split}
	\end{align}
	where $\bar{\mathcal{L}}_i^{(1, 0)} \in [\hat{\mathcal{L}}_i^{(1, 0)}, \mathcal{L}_i^{(1, 0)}]$.
	Further, the mean value expansion and Assumption \ref{as:1stage}(i) lead to
	\begin{align*}
		b_K(\hat P_{1i}^0, \hat P_{2i}^1) - b_K(P_{1i}^0, P_{2i}^1)
		& = \left\{ \frac{\partial b_K \left( \bar P_{1i}^0, \bar P_{2i}^1 \right)}{ \partial p_1} + \frac{\partial b_K \left( \bar P_{1i}^0, \bar P_{2i}^1 \right)}{ \partial p_2}\right\} \cdot O_P(n^{-1/2}),
	\end{align*}
	where $\bar P_{1i}^0 \in [\hat P_{1i}^0, P_{1i}^0]$ and $\bar P_{2i}^1 \in [\hat P_{2i}^1, P_{2i}^1]$.
	Thus, by the triangle inequality,
	\begin{align}\label{eq:meanval1}
	\begin{split}
		\left\| \left( \hat{\mathbf{R}}_K - \mathbf{R}_K \right)/\sqrt{n} \right\| 
		& \le \left\| \left( \hat{\mathbf{T}} - \tilde{\mathbf{T}} \right) \hat{\mathbf{b}}_K /\sqrt{n} \right\| + \left\| \tilde{\mathbf{T}} \left( \hat{\mathbf{b}}_K - \mathbf{b}_K \right) /\sqrt{n} \right\| + O_P(1/\sqrt{n}) \\
		& \leq  O_P(n^{-1/2}) \cdot\left\{ \left\| \hat{\mathbf{b}}_K / \sqrt{n} \right\| + \left\| \partial_1 \bar{\mathbf{b}}_K/\sqrt{n} \right\|  + \left\| \partial_2 \bar{\mathbf{b}}_K/\sqrt{n} \right\| \right\} + O_P(1/\sqrt{n}) \\
		&= O_P(\zeta_0(K) / \sqrt{n}) + O_P(\zeta_1(K)/ \sqrt{n}) + O_P(1/\sqrt{n})
		= O_P(\zeta_1(K)/ \sqrt{n}),
	\end{split}
	\end{align}
	where $\partial_j \bar{\mathbf{b}}_K \coloneqq ( \partial b_K (\bar P_{11}^0, \bar P_{21}^1) / \partial p_j, \ldots, \partial b_K (\bar P_{1n}^0, \bar P_{2n}^1) / \partial p_j )^\top$ for $j = 1,2$.
	Hence, the first term satisfies $\| ( \hat{\mathbf{R}}_K^\top - \mathbf{R}_K^\top ) ( \hat{\mathbf{R}}_K - \mathbf{R}_K ) /n \| \le \| ( \hat{\mathbf{R}}_K - \mathbf{R}_K ) / \sqrt{n} \|^2 = O_P(\zeta^2_1(K) / n)$.
	
	For the second term, we obtain
	\begin{align*}
		\left\| \mathbf{R}_K^\top \left( \hat{\mathbf{R}}_K - \mathbf{R}_K \right) /n \right\|^2
		& = \mathrm{tr} \left\{ \left( \hat{\mathbf{R}}_K^\top - \mathbf{R}_K^\top \right) \mathbf{R}_K \mathbf{R}_K^\top \left( \hat{\mathbf{R}}_K - \mathbf{R}_K \right) /n^2 \right\} \\
		& \leq [\bar{c}_\Psi + o_P(1)] \cdot \mathrm{tr} \left\{ \left( \hat{\mathbf{R}}_K^\top - \mathbf{R}_K^\top \right)  \left( \hat{\mathbf{R}}_K - \mathbf{R}_K \right) /n \right\} \\
		& = [\bar{c}_\Psi + o_P(1)] \cdot \| ( \hat{\mathbf{R}}_K - \mathbf{R}_K ) / \sqrt{n} \|^2
		= O_P(\zeta^2_1(K) / n),
	\end{align*}
	since result (i) and Assumption \ref{as:eigen}(i) imply that
	\begin{align}\label{eq:eigenPsi}
		0 < \underline{c}_\Psi + o_P(1) \leq \chi_\mathrm{min}\left( \Psi_{nK} \right) \leq \chi_\mathrm{max}\left( \Psi_{nK} \right) \leq \bar{c}_\Psi + o_P(1) < \infty, \quad \text{w.p.a.1.}
	\end{align}
	
	Thus, result (ii) holds by
	\begin{align} \label{eq:differencePsi}
		\| \hat \Psi_{nK} - \Psi_{nK} \|_2
		= O_P(\zeta_1(K) / \sqrt{n}) + O_P(\zeta^2_1(K) / n)
		= O_P(\zeta_1(K) / \sqrt{n}).
	\end{align}
	\bigskip
	
	(iii) We first note that $\Psi_{nK}^{-1} - \Psi_K^{-1} = \Psi_K^{-1} \left( \Psi_K - \Psi_{nK} \right) \Psi_{nK}^{-1}$.
	Then, \eqref{eq:eigenPsi} and result (i) imply the desired result: $\|\Psi_{nK}^{-1} - \Psi_K^{-1}\|_2 \le \|\Psi_K^{-1}\|_2 \| \Psi_K - \Psi_{nK} \|_2 \| \Psi_{nK}^{-1}\|_2 \overset{p}{\asymp} \| \Psi_K - \Psi_{nK} \|_2$.
	\bigskip
	
	(iv) The proof is the same as in result (iii) by noting that result (ii) and Assumption \ref{as:eigen}(i) imply
	\begin{align}\label{eq:eigenPsihat}
		0 < \underline{c}_\Psi + o_P(1) \leq \chi_\mathrm{min}\left( \hat \Psi_{nK}  \right) \leq \chi_\mathrm{max}\left( \hat \Psi_{nK} \right) \leq \bar{c}_\Psi + o_P(1) < \infty, \quad \text{w.p.a.1.}
	\end{align}
\end{proof}

\begin{lemma}\label{lem:parametric}
	Suppose that Assumptions \ref{as:gamemodel}, \ref{as:linear}, \ref{as:data}, \ref{as:1stage}, \ref{as:eigen}(i), and \ref{as:emom}--\ref{as:rate} hold. 
	If in addition $\sqrt{n} K^{-s/2} = O(1)$ holds, then we have 	
	\begin{align*}
	\renewcommand{\arraystretch}{1.5}
	\begin{array}{clcl}
	\text{(i)} & \left\| \tilde \beta_1^{(1, 0)} - \beta_{1}^{(1, 0)} \right\| =  O_P(n^{-1/2}), 
	& \text{(ii)} & \left\|\tilde \alpha^{(1, 0)} - \alpha^{(1, 0)}  \right\| =  O_P(\sqrt{K/n}), \\
	\text{(iii)} & \left\| \hat \beta_1^{(1, 0)} - \beta_{1}^{(1, 0)} \right\| =  O_P(n^{-1/2}), 
	& \text{(iv)} & \left\|\hat \alpha^{(1, 0)} - \alpha^{(1, 0)}  \right\| =  O_P(\sqrt{K/n}).
	\end{array}
	\renewcommand{\arraystretch}{1}
	\end{align*}
\end{lemma}


\begin{proof}
 
    (i)--(ii) We first show that the first term on the right-hand side of \eqref{eq:infdecomp1} is of order $O_P(n^{-1/2})$ for any choice of $\mathbb{S}_a$.
    Since each element of the diagonal matrix $\tilde{\mathbf{T}}$ is bounded, we have 
    \begin{align*}
    	\| \mathbb{S}_a \Psi_{nK}^{-1} \mathbf{R}_K^\top \tilde{\mathbf{T}} \mathbf{u} /n \|^2 
    	& = \mathrm{tr}\left\{ \mathbf{u}^\top \tilde{\mathbf{T}} \mathbf{R}_K \Psi_{nK}^{-1} \mathbb{S}_a^\top \mathbb{S}_a \Psi_{nK}^{-1} \mathbf{R}_K^\top \tilde{\mathbf{T}} \mathbf{u} \right\} / n^2  \\
    	& \leq \mathrm{tr}\left\{ \mathbf{u}^\top \tilde{\mathbf{T}} \mathbf{R}_K \Psi_{nK}^{-2} \mathbf{R}_K^\top \tilde{\mathbf{T}} \mathbf{u} \right\} / n^2  \\
    	& \leq [ \underline{c}_\Psi + o_P(1)]^{-2} \cdot \mathrm{tr}\left\{ \mathbf{u}^\top \tilde{\mathbf{T}} \left( \mathbf{R}_K \mathbf{R}_K^\top / n \right) \tilde{\mathbf{T}} \mathbf{u} \right\} /n 
    	= O_P(1) \cdot \| \mathbf{u} \|^2 /n,
    \end{align*}
    where the second and last equalities follow from \eqref{eq:eigenPsi}.
    From Assumption \ref{as:sieve},
    \begin{align*}
    	\| \mathbf{u} \|^2 
    	\leq n \cdot \left[ \sup_{(p_1, p_2) \in [0,1]^2 } \left| g^{(1,0)}(p_1, p_2) - b_K(p_1, p_2)^\top \alpha^{(1, 0)} \right| \right]^2 
    	= O(n K^{-s}), 
    \end{align*}
    implying that
    \begin{align}\label{eq:biasorder}
    	\| \mathbb{S}_a \Psi_{nK}^{-1} \mathbf{R}_K^\top \tilde{\mathbf{T}} \mathbf{u} /n \| =  \underbrace{O_P(K^{-s/2})}_{O_P(n^{-1/2})}.
    \end{align}
    For the second term in \eqref{eq:infdecomp1}, Assumptions \ref{as:data} and \ref{as:emom} and \eqref{eq:eigenPsi} imply that
    \begin{align*}
    	\bE \left[ \left\| \mathbb{S}_a \Psi_{nK}^{-1} \mathbf{R}_K^\top \tilde{\mathbf{e}} / n \right\|^2 \middle| \{W_i, D_i \}_{i = 1}^n \right]
    	& = \mathrm{tr}\left\{ \mathbb{S}_a \Psi_{nK}^{-1} \mathbf{R}_K^\top \bE \left[ \tilde{\mathbf{e}} \tilde{\mathbf{e}}^\top \middle| \{W_i, D_i\}_{i = 1}^n \right] \mathbf{R}_K \Psi_{nK}^{-1} \mathbb{S}_a^\top \right\} / n^2  \\
    	& = O_P(1) \cdot \mathrm{tr}\left\{ \mathbb{S}_a \Psi_{nK}^{-1} \mathbb{S}_a^\top \right\} /n 
    	= O_P\left(a/n \right).
    \end{align*}
    Hence, we have
    \begin{align}\label{eq:variorder}
    	\left\| \mathbb{S}_a \Psi_{nK}^{-1} \mathbf{R}_K^\top \tilde{\mathbf{e}} / n \right\| =  O_P\left(\sqrt{ a/n } \right),
    \end{align}
    by Markov's inequality. 
    Results (i) and (ii) follow by noting that they are the cases when $a = \mathrm{dim}(X)$ and $a = K$, respectively.
    \bigskip
    
    (iii)--(iv) Using \eqref{eq:eigenPsihat} and the same argument as in \eqref{eq:biasorder}, it holds that $\| \mathbb{S}_a \hat \Psi_{nK}^{-1} \hat{\mathbf{R}}_K^\top \hat{\mathbf{T}} \hat{\mathbf{u}} / n \| =  O_P(K^{-s/2}) = O_P(n^{-1/2})$ for any $\mathbb{S}_a$.
    Further, similarly to \eqref{eq:variorder}, we have
    \begin{align*}
    \bE \left[ \left\| \mathbb{S}_a \hat \Psi_{nK}^{-1}  \hat{\mathbf{R}}_K^\top \hat{\mathbf{e}}/n \right\|^2 \middle| \{W_i, D_i\}_{i = 1}^n \right]
    & = \mathrm{tr}\left\{ \mathbb{S}_a \hat \Psi_{nK}^{-1}  \hat{\mathbf{R}}_K^\top \bE \left[ \hat{\mathbf{e}} \hat{\mathbf{e}}^\top \middle| \{W_i, D_i\}_{i = 1}^n \right] \hat{\mathbf{R}}_K  \hat \Psi_{nK}^{-1}  \mathbb{S}_a^\top \right\} / n^2  \\
    & = O_P(1) \cdot \mathrm{tr}\left\{ \mathbb{S}_a \hat \Psi_{nK}^{-1} \mathbb{S}_a^\top \right\} /n = O_P\left(a/n \right),
    \end{align*}
    by Assumptions \ref{as:data} and \ref{as:emom} and \eqref{eq:eigenPsihat}, which leads to $\| \mathbb{S}_a \hat \Psi_{nK}^{-1}  \hat{\mathbf{R}}_K^\top \hat{\mathbf{e}} / n \| = O_P(\sqrt{a/n})$ by Markov's inequality.
    Therefore, by \eqref{eq:fdecomp1},
    \begin{align}\label{eq:hatdltecomp}
    \begin{split}
    	\mathbb{S}_a\left(\hat \delta^{(1, 0)} - \delta^{(1, 0)}\right)
    	& =  \mathbb{S}_a \hat \Psi_{nK}^{-1} \hat{\mathbf{R}}_K^\top \breve{\mathbf{T}} \left(\mathbf{g}^{(1,0)} - \hat{\mathbf{g}}^{(1,0)} \right) / n + \mathbb{S}_a \hat \Psi_{nK}^{-1} \hat{\mathbf{R}}_K^\top \left( \breve{\mathbf{T}} - \hat{\mathbf{T}} \right) \hat{\mathbf{g}}^{(1,0)} / n + O_P(\sqrt{a/n}).
    \end{split}
    \end{align}
    By the mean value expansion under Assumptions \ref{as:1stage}(i) and \ref{as:sieve}, we have
    \begin{align*}
    	g^{(1, 0)}( \hat P_{1i}^0, \hat P_{2i}^1 ) - 	g^{(1, 0)}(P_{1i}^0, P_{2i}^1)
    	& = \left\{ \frac{\partial g^{(1, 0)} \left( \bar P_{1i}^0, \bar P_{2i}^1 \right)}{ \partial p_1} + \frac{\partial g^{(1, 0)} \left( \bar P_{1i}^0, \bar P_{2i}^1 \right)}{ \partial p_2}\right\} \cdot O_P(n^{-1/2}),
    \end{align*}
    where $\bar P_{1i}^0 \in [\hat P_{1i}^0, P_{1i}^0]$ and $\bar P_{2i}^1 \in [\hat P_{2i}^1, P_{2i}^1]$.
    Thus, noting that $\breve T^{(1, 0)} = \tau_\varpi(\hat{\mathcal{L}}^{(1, 0)}) I^{(1,0)} / \mathcal{L}^{(1, 0)}$ is bounded w.p.a.1 since $\hat{\mathcal{L}}^{(1, 0)} - \mathcal{L}^{(1, 0)} = O_P(n^{-1/2})$, the triangle inequality and \eqref{eq:eigenPsihat} lead to
    \begin{align}\label{eq:gmeanval}
    	\begin{split}
    	\left\|  \mathbb{S}_a \hat{\Psi}_{nK}^{-1} \hat{\mathbf{R}}_K^\top \breve{\mathbf{T}} \left(\mathbf{g}^{(1, 0)} - \hat{\mathbf{g}}^{(1, 0)} \right) / n  \right\| 
    	& \leq O_P(n^{-1/2}) \cdot \left\{ \left\|  \mathbb{S}_a \hat{\Psi}_{nK}^{-1} \hat{\mathbf{R}}_K^\top \breve{\mathbf{T}} \partial_1\bar{\mathbf{g}}^{(1, 0)}/n \right\|  + \left\|  \mathbb{S}_a \hat{\Psi}_{nK}^{-1}\hat{\mathbf{R}}_K^\top \breve{\mathbf{T}} \partial_2\bar{\mathbf{g}}^{(1, 0)}/n \right\| \right\}\\
    	& \leq O_P(n^{-1/2}) \cdot \sqrt{\mathrm{tr}\left\{ \mathbb{S}_a \hat \Psi_{nK}^{-1} \mathbb{S}_a^\top \right\}}  = O_P(\sqrt{a/n}),
    	\end{split}
    \end{align}
    where $\partial_j \bar{\mathbf{g}}^{(1, 0)} = ( \partial g^{(1,0)} ( \bar P_{11}^0, \bar P_{21}^1 ) / \partial p_j, \ldots, \partial g^{(1,0)} ( \bar P_{1n}^0, \bar P_{2n}^1 ) / \partial p_j)^\top$ for $j = 1,2$, implying that the first term in \eqref{eq:hatdltecomp} is of order $O_P(\sqrt{a/n})$.
    For the second term in \eqref{eq:hatdltecomp}, it can be similarly verified that $\| \mathbb{S}_a \hat \Psi_{nK}^{-1} \hat{\mathbf{R}}_K^\top (\breve{\mathbf{T}} - \hat{\mathbf{T}}) \hat{\mathbf{g}}^{(1,0)} / n \| =  O_P(\sqrt{a/n})$.
    This completes the proof. 
\end{proof}

To state the next lemma, we define the following linear operator:
\begin{align*}
	\hat{\mathcal{P}}_{nK}^{(1,0)} g 
	\coloneqq b_K(\cdot, \cdot)^\top \mathbb{S}_K \left[ \hat \Psi^{(1, 0)}_{nK} \right]^{-1} \frac{1}{n} \sum_{i = 1}^n \hat R_{K,i}^{(1,0)} g(Q_i),
\end{align*}
The operator norm of $\hat{\mathcal{P}}_{nK}^{(1,0)}$ (restricted on $\mathcal{S}^{(1,0)}$) is denoted as $\| \hat{\mathcal{P}}_{nK}^{(1,0)} \|_{\infty}$.

\begin{lemma}\label{lem:opnorm}
	Suppose that Assumptions \ref{as:gamemodel}, \ref{as:linear}, \ref{as:data}, \ref{as:1stage}, \ref{as:eigen}(i), \ref{as:rate}, and \ref{as:opnorm} hold.
	Then, we have
	\begin{align*}
		\| \hat{\mathcal{P}}_{nK}^{(1,0)} \|_{\infty} 
		= \| \mathcal{P}_{nK}^{(1,0)} \|_{\infty} + O_P(1)
		= O_P(1).
	\end{align*}
\end{lemma}

\begin{proof}
	The triangle inequality implies that
	\begin{align*}
		\left| \| \hat{\mathcal{P}}_{nK} \|_{\infty} - \| \mathcal{P}_{nK} \|_{\infty} \right| 
		\le \sup_{(p_1, p_2) \in \mathcal{S}^{(1,0)}, \; g \in \mathcal{C}(\text{supp}[Q])} \left| \left( \hat{\mathcal{P}}_{nK} g \right) (p_1, p_2) - \left( \mathcal{P}_{nK} g \right) (p_1, p_2) \right|.
	\end{align*}
	For any $(p_1, p_2) \in \mathcal{S}^{(1,0)}$ and $g \in \mathcal{C}(\text{supp}[Q])$, we have
	\begin{align*}
		\left| \left( \hat{\mathcal{P}}_{nK} g \right) (p_1, p_2) - \left( \mathcal{P}_{nK} g \right) (p_1, p_2) \right|
		& \le \left| b_K(p_1, p_2)^\top \mathbb{S}_K \left( \hat \Psi_{nK}^{-1} - \Psi_{nK}^{-1} \right) \hat{\mathbf{R}}_K^\top \mathbf{g} / n \right| \\
		& \quad + \left| b_K(p_1, p_2)^\top \mathbb{S}_K \Psi_{nK}^{-1} \left( \hat{\mathbf{R}}_K - \mathbf{R}_K \right)^\top \mathbf{g} / n \right|,
	\end{align*}
	where $\mathbf{g} = (g(Q_1), \dots, g(Q_n))^\top$.
	For the first term, we have
	\begin{align*}
		\left| b_K(p_1, p_2)^\top \mathbb{S}_K \left( \hat \Psi_{nK}^{-1} - \Psi_{nK}^{-1} \right) \hat{\mathbf{R}}_K^\top \mathbf{g} / n \right|
		& \le \| b_K(p_1, p_2) \| \cdot \left\| \mathbb{S}_K \left( \hat \Psi_{nK}^{-1} - \Psi_{nK}^{-1} \right) \hat{\mathbf{R}}_K^\top \mathbf{g} / n \right\| \\
		& \le \| b_K(p_1, p_2) \| \cdot O_P(\zeta_1(K) / \sqrt{n})
		= O_P(\zeta_0(K) \zeta_1(K) / \sqrt{n}),
	\end{align*}
	where the second inequality can be shown by \eqref{eq:eigenPsihat} and
	\begin{align} \label{eq:differencePsiInv}
		\| \hat \Psi_{nK}^{-1} - \Psi_{nK}^{-1} \|_2 = O_P(\zeta_1(K) / \sqrt{n}),
	\end{align}
	which is implied by \eqref{eq:differencePsi} and the same argument as in the proof of Lemma \ref{lem:matLLN}.
	Similarly, we can show that the second term is $O_P(\zeta_0(K) \zeta_1(K) / \sqrt{n})$ based on \eqref{eq:meanval1}.
	Thus, under Assumption \ref{as:opnorm}(ii),
	\begin{align*}
		\left| \left( \hat{\mathcal{P}}_{nK} g \right) (p_1, p_2) - \left( \mathcal{P}_{nK} g \right) (p_1, p_2) \right| 
		= O_P(\zeta_0(K) \zeta_1(K) / \sqrt{n})
		= O_P(1),
	\end{align*}
	uniformly in $(p_1, p_2) \in \mathcal{S}^{(1,0)}$ and $g \in \mathcal{C}(\text{supp}[Q])$.
\end{proof}

\begin{lemma}\label{lem:unifconv1}
	Suppose that Assumptions \ref{as:gamemodel}, \ref{as:linear}, \ref{as:data}, \ref{as:1stage}, \ref{as:eigen}(i), and \ref{as:emom}--\ref{as:rate2} hold.
	Then, we have
	\begin{align*}
	\renewcommand{\arraystretch}{2}
	\begin{array}{cl}
	\text{(i)} & \displaystyle \sup_{(p_1, p_2) \in \mathcal{S}^{(1,0)}}\left| \tilde g^{(1, 0)}(p_1, p_2) -  g^{(1, 0)}(p_1, p_2)  \right| =  O_P(\zeta_0(K) \sqrt{(\log n) / n} ) + O_P(K^{-s/2}),\\
	\text{(ii)} & \displaystyle \sup_{(p_1, p_2) \in \mathcal{S}^{(1,0)}}\left| \partial_{p_1 p_2 }\left( \tilde g^{(1, 0)}(p_1, p_2) -  g^{(1, 0)}(p_1, p_2) \right)  \right| = O_P(\zeta_0(K) K \sqrt{(\log n) / n} ) + O_P(K^{(2-s)/2}),\\
	\text{(iii)} & \displaystyle \sup_{(p_1, p_2) \in \mathcal{S}^{(1,0)}}\left| \hat g^{(1, 0)}(p_1, p_2) -  g^{(1, 0)}(p_1, p_2)  \right| =  O_P(\zeta_0(K) \sqrt{(\log n) / n} ) + O_P(K^{-s/2}),\\
	\text{(iv)} & \displaystyle \sup_{(p_1, p_2) \in \mathcal{S}^{(1,0)}}\left|  \partial_{p_1 p_2 }\left( \hat g^{(1, 0)}(p_1, p_2) -  g^{(1, 0)}(p_1, p_2) \right)  \right| = O_P(\zeta_0(K) K \sqrt{(\log n) / n} ) + O_P(K^{(2-s)/2}).
	\end{array}
	\renewcommand{\arraystretch}{1}
	\end{align*}
\end{lemma}

\begin{proof}
	(i) By the triangle inequality, Assumptions \ref{as:sieve} and \ref{as:opnorm}, and \eqref{eq:infdecomp1}, we have
	\begin{align*}
	&  \sup_{(p_1, p_2) \in \mathcal{S}^{(1,0)}}\left| \tilde g^{(1, 0)}(p_1, p_2) -  g^{(1, 0)}(p_1, p_2)  \right| \\
	& \leq  \sup_{(p_1, p_2) \in \mathcal{S}^{(1,0)}}\left| b_K(p_1, p_2)^\top ( \tilde \alpha^{(1, 0)} -  \alpha^{(1, 0)}) \right| +  \sup_{(p_1, p_2) \in \mathcal{S}^{(1,0)}}\left| g^{(1, 0)}(p_1, p_2) -  b_K(p_1, p_2)^\top \alpha^{(1, 0)}  \right| \\
	& =  \sup_{(p_1, p_2) \in \mathcal{S}^{(1,0)}}\left| b_K(p_1, p_2)^\top ( \tilde \alpha^{(1, 0)} -  \alpha^{(1, 0)}) \right| + O(K^{-s/2})\\
	& \leq \sup_{(p_1, p_2) \in \mathcal{S}^{(1,0)}} \left| b_K(p_1, p_2)^\top \mathbb{S}_K \Psi_{nK}^{-1} \mathbf{R}_K^\top \tilde{\mathbf{e}} / n \right|  + \| \mathcal{P}_{nK} \|_{\infty} \cdot O(K^{-s/2})  + O(K^{-s/2}) \\
	& = \sup_{(p_1, p_2) \in \mathcal{S}^{(1,0)}} \left| b_K(p_1, p_2)^\top \mathbb{S}_K \Psi_{nK}^{-1} \mathbf{R}_K^\top \tilde{\mathbf{e}} / n \right| + O_P(K^{-s/2}).
	\end{align*}
	
	Following the proof of Lemma 3.1(ii) of \citetappendix{chen2018optimal}, we show that the first term on the right-hand side is of order $O_P(\zeta_0(K) \sqrt{(\log n) / n})$.
	Firstly, we partition the interval $[0, 1]$ into countably many sub-intervals of equal length and let the set of the partitioning points (including 0 and 1) be $\mathcal{T}_n$.
	We can construct the partition such that for any $(p_1,p_2) \in [0,1]^2$ there exists a point $( t_{p_1}, t_{p_2} ) \in \mathcal{T}_n^2$ satisfying 
	\[
	\| (p_1,p_2)  - (  t_{p_1}, t_{p_2}  ) \| \leq c_p \zeta_0(K) K^{-(\omega + 1/2)},
	\]
	for some positive constant $c_p > 0$, where $\omega$ is as in Assumption \ref{as:rate2}(i).
	Then, by Assumption \ref{as:rate2}(i), we have
	\begin{align*}
	& \sup_{(p_1, p_2) \in \mathcal{S}^{(1,0)}}\left| b_K(p_1, p_2)^\top \mathbb{S}_K \Psi_{nK}^{-1} \mathbf{R}_K^\top \tilde{\mathbf{e}} / n \right| \\
	&\leq \max_{(t_1, t_2) \in \mathcal{T}_n^2}\left| b_K(t_1, t_2)^\top \mathbb{S}_K \Psi_{nK}^{-1} \mathbf{R}_K^\top \tilde{\mathbf{e}} / n \right|  +  \sup_{(p_1, p_2) \in [0, 1]^2}\left| \left\{ b_K(p_1, p_2) -  b_K(t_{p_1}, t_{p_2}) \right\}^\top \mathbb{S}_K \Psi_{nK}^{-1} \mathbf{R}_K^\top \tilde{\mathbf{e}} / n \right|  \\
	& \leq \max_{(t_1, t_2) \in \mathcal{T}_n^2}\left| b_K(t_1, t_2)^\top \mathbb{S}_K \Psi_{nK}^{-1} \mathbf{R}_K^\top \tilde{\mathbf{e}} / n \right|  +  O(\zeta_0(K) K^{-1/2}) \cdot \left\|\mathbb{S}_K \Psi_{nK}^{-1} \mathbf{R}_K^\top \tilde{\mathbf{e}} / n \right\| \\
	& =  \max_{(t_1, t_2) \in \mathcal{T}_n^2}\left| b_K(t_1, t_2)^\top \mathbb{S}_K \Psi_{nK}^{-1} \mathbf{R}_K^\top \tilde{\mathbf{e}} / n \right|  +  O_P( \zeta_0(K) /\sqrt{n} ),
	\end{align*}
	where the last equality follows from \eqref{eq:variorder}.
	
	To examine the first term on the right-hand side, decompose $\tilde e_i = \tilde e_{1i} + \tilde e_{2i}$, where
	\begin{align*}
	\tilde e_{1i} & \coloneqq \tilde e_i \mathbf{1}\{| \tilde e_i| \leq M_n \} - \bE [ \tilde e_i \mathbf{1}\{| \tilde e_i| \leq M_n \} | W_i, D_i], \\
	\tilde e_{2i} & \coloneqq \tilde e_i \mathbf{1}\{|\tilde e_i| > M_n \} - \bE [\tilde e_i \mathbf{1}\{|\tilde e_i| > M_n \} | W_i, D_i],
	\end{align*}
	and $M_n$ is a sequence of positive numbers diverging to $\infty$.
	Let $\tilde{\mathbf{e}}_{1} = (\tilde e_{11}  , \ldots , \tilde e_{1n} )^\top$ and $\tilde{\mathbf{e}}_{2} = (\tilde e_{21}, \ldots, \tilde e_{2n})^\top$.
	Then, we observe
	\begin{align*}
	& b_K(t_1, t_2)^\top \mathbb{S}_K \Psi_{nK}^{-1} \mathbf{R}_K^\top \tilde{\mathbf{e}} / n \\
	& = b_K(t_1, t_2)^\top \mathbb{S}_K \Psi_K^{-1} \mathbf{R}_K^\top \tilde{\mathbf{e}} /n + b_K(t_1, t_2)^\top \mathbb{S}_K \left( \Psi_{nK}^{-1}  - \Psi_K^{-1} \right) \mathbf{R}_K^\top \tilde{\mathbf{e}} /n \\
	& = b_K(t_1, t_2)^\top \mathbb{S}_K \Psi_K^{-1} \mathbf{R}_K^\top \tilde{\mathbf{e}}_{1} /n + b_K(t_1, t_2)^\top \mathbb{S}_K \Psi_K^{-1}  \mathbf{R}_K^\top \tilde{\mathbf{e}}_{2} /n + b_K(t_1, t_2)^\top \mathbb{S}_K \left( \Psi_{nK}^{-1} - \Psi_K^{-1} \right) \mathbf{R}_K^\top \tilde{\mathbf{e}} /n \\
	& = b_K(t_1, t_2)^\top \mathbb{S}_K \Psi_K^{-1} \mathbf{R}_K^\top \tilde{\mathbf{e}}_{1} / n + b_K(t_1, t_2)^\top \mathbb{S}_K \Psi_K^{-1}  \mathbf{R}_K^\top \tilde{\mathbf{e}}_{2} /n + \underbrace{O_P(\zeta^2_0(K) \sqrt{(\log K)}/n )}_{o_P(\zeta_0(K) \sqrt{(\log n)/ n })},
	\end{align*}
	where the last equality follows from Lemma \ref{lem:matLLN}(iii) and Markov's inequality.
	
	Let $q_i(t_1,t_2) \coloneqq b_K(t_1, t_2)^\top \mathbb{S}_K \Psi_K^{-1} R_{K,i}$, so that
	\begin{align*}
		\frac{1}{n}\sum_{i=1}^n q_i(t_1,t_2) \tilde e_{1i} = b_K(t_1, t_2)^\top \mathbb{S}_K \Psi_K^{-1} \mathbf{R}_K^\top \tilde{\mathbf{e}}_{1} / n.
	\end{align*}
	Note that $\bE [q_i(t_1, t_2) \tilde e_{1i}] = 0$.
	Furthermore, it is straightforward by the Cauchy--Schwarz inequality that there exist positive constants $c_1, c_2 >0$ such that $|q_i(t_1,t_2)| \le \| b_K(t_1, t_2) \| \cdot \| \Psi_K^{-1} R_{K,i} \| \le c_1 \zeta^2_0(K)$ and that $\bE [q_i^2(t_1,t_2)] \leq  c_2 \zeta^2_0(K)$.
	Therefore, for all $(t_1, t_2) \in \mathcal{T}_n^2$, we have $|q_i(t_1,t_2) \tilde e_{1i}|  \leq  c'_1 \zeta^2_0(K) M_n$ 	and $ \bE [q_i^2(t_1,t_2) \tilde e^2_{1i}] = \bE \left[q_i^2(t_1,t_2)\bE [\tilde e^2_{1i} | W_i, D_i]\right] \leq c'_2 \zeta^2_0(K)$ for some $c'_1, c'_2 > 0$ by Assumption \ref{as:emom}.
	From Bernstein's inequality for any non-negative $\varrho_n \ge 0$, we have
	\begin{align*}
		\Pr\left[ \max_{(t_1, t_2) \in \mathcal{T}_n^2}\left| b_K(t_1, t_2)^\top \mathbb{S}_K \Psi_K^{-1} \mathbf{R}_K^\top \tilde{\mathbf{e}}_{1} / n \right|  > \varrho_n \right]
		& \leq |\mathcal{T}_n^2| \max_{(t_1, t_2) \in \mathcal{T}_n^2} \Pr\left[ \left| \frac{1}{n}\sum_{i=1}^n q_i(t_1,t_2) \tilde e_{1i} \right|  > \varrho_n \right] \\
		& \leq 2 \exp\left\{\log |\mathcal{T}_n^2|-\frac{1}{2}\frac{\varrho_n^2}{c'_2 \zeta^2_0(K)/n +  c'_1 \zeta^2_0(K) M_n \varrho_n /(3n)} \right\}\\
		& \leq 2 \exp\left\{\log |\mathcal{T}_n^2| -\frac{\varrho_n^2}{c_3 (\zeta^2_0(K)/n) [1+  M_n \varrho_n ]} \right\},
	\end{align*}
	for some positive constant $c_3 > 0$, where $|\mathcal{T}_n^2|$ denotes the cardinality of the set $\mathcal{T}_n^2$.
	Then, setting $\varrho_n = C \zeta_0(K)\sqrt{(\log n)/ n}$ for a large constant $C > 0$, provided that $|\mathcal{T}_n^2|$ and $M_n$ grow sufficiently slowly so that $M_n \varrho_n = o(1)$, we have
	\begin{align*}
		\log |\mathcal{T}_n^2| -\frac{\varrho_n^2}{c_3 (\zeta^2_0(K)/n) [1+  M_n \varrho_n ]}
		= \log |\mathcal{T}_n^2| -\frac{C^2 \zeta^2_0(K) (\log n)/ n}{c_3 (\zeta^2_0(K)/n) [1+  o(1) ]}
		\asymp \log\left(\frac{ |\mathcal{T}_n^2| }{n^{C^2} } \right) \to - \infty,
	\end{align*}
	as $C \to \infty$, implying that $\max_{(t_1, t_2) \in \mathcal{T}_n^2}\left| b_K(t_1, t_2)^\top \mathbb{S}_K \Psi_K^{-1}  \mathbf{R}_K^\top \tilde{\mathbf{e}}_{1} /n \right| = O_P(\zeta_0(K)\sqrt{(\log n)/ n})$.
	
	Next, by Markov's inequality and Assumption \ref{as:emom}, it holds that
	\begin{align*}
	\Pr\left[ \max_{(t_1, t_2) \in \mathcal{T}_n^2}\left| b_K(t_1, t_2)^\top \mathbb{S}_K \Psi_K^{-1} \mathbf{R}_K^\top \tilde{\mathbf{e}}_{2} / n \right| > \varrho_n \right] 
	& \leq \Pr\left[ \zeta_0(K) \left\| \Psi_K^{-1} \mathbf{R}_K^\top \tilde{\mathbf{e}}_{2} /n \right\| > \varrho_n \right] \\
	& \leq \Pr\left[  \frac{c_4 \zeta^2_0(K)}{n}\sum_{i = 1}^n | \tilde e_{2i} | > \varrho_n \right] \\
	& \leq \frac{2 c_4 \zeta^2_0(K) }{\varrho_n} \bE [|\tilde e_i| \mathbf{1}\{|\tilde e_i| > M_n \}] \\
	& \leq \frac{2 c_4 \zeta^2_0(K) }{\varrho_n M_n^3} \bE [ \tilde e^4_i \mathbf{1}\{|\tilde e_i| > M_n \}] = O\left( \frac{ \zeta^2_0(K) }{\varrho_n M_n^3} \right).
	\end{align*}
	Again, setting $\varrho_n = C \zeta_0(K)\sqrt{(\log n)/ n}$ for a large constant $C > 0$, if $\zeta_0(K) / \sqrt{(\log n)/ n} = O(M_n^3)$,
	\begin{align*}
		\frac{ \zeta^2_0(K) }{\varrho_n M_n^3} = \frac{1}{C} \frac{ \zeta_0(K) }{\sqrt{(\log n)/ n} M_n^3} \to 0,
	\end{align*}
	as $C \to \infty$, which implies $\max_{(t_1, t_2) \in \mathcal{T}_n^2}\left| b_K(t_1, t_2)^\top \mathbb{S}_K \Psi_K^{-1}  \mathbf{R}_K^\top \tilde{\mathbf{e}}_{2} / n \right| = O_P(\zeta_0(K)\sqrt{(\log n)/ n})$.
	It should be noted that $ \zeta_0(K) / \sqrt{(\log n)/ n} = O(M_n^3)$ is not inconsistent with the requirement $M_n \zeta_0(K)\sqrt{(\log n)/ n} = o(1)$ under Assumption \ref{as:rate2}(ii). 
	By combining these results, the proof is completed.
	\bigskip
	
	(ii) From the triangle inequality and Assumption \ref{as:sieve}, we have
	\begin{align*}
	&  \sup_{(p_1, p_2) \in  \mathcal{S}^{(1,0)}}\left| \partial_{p_1 p_2}\left( \tilde g^{(1, 0)}(p_1, p_2) -  g^{(1, 0)}(p_1, p_2)\right)  \right| \\
	& \leq  \sup_{(p_1, p_2) \in  \mathcal{S}^{(1,0)}}\left| \partial_{p_1 p_2}\left( b_K(p_1, p_2)^\top ( \tilde \alpha^{(1, 0)} -  \alpha^{(1, 0)}) \right) \right|  +  \sup_{(p_1, p_2) \in  \mathcal{S}^{(1,0)}}\left| \partial_{p_1 p_2}\left( g^{(1, 0)}(p_1, p_2) -  b_K(p_1, p_2)^\top \alpha^{(1, 0)} \right) \right| \\
	& \leq  \sup_{(p_1, p_2) \in  \mathcal{S}^{(1,0)}}\left| \partial_{p_1 p_2}\left( b_K(p_1, p_2)^\top ( \tilde \alpha^{(1, 0)} -  \alpha^{(1, 0)}) \right) \right| +  O(K^{(2 - s)/2}).
	\end{align*}
	Further, by Assumption \ref{as:deriv}, we have
	\begin{align*}
	\sup_{(p_1, p_2) \in  \mathcal{S}^{(1,0)}}\left| \partial_{p_1 p_2}\left( b_K(p_1, p_2)^\top ( \tilde \alpha^{(1, 0)} -  \alpha^{(1, 0)}) \right) \right| 
	& = O(K) \sup_{(p_1, p_2) \in  \mathcal{S}^{(1,0)}}\left|  b_K(p_1, p_2)^\top ( \tilde \alpha^{(1, 0)} -  \alpha^{(1, 0)}) \right| \\
	& = O_P(\zeta_0(K)K\sqrt{(\log n)/ n}) +  O_P(K^{(2-s)/2}),
	\end{align*}
	where the second equality follows from result (i).
	This completes the proof.
	\bigskip
	
	(iii) From \eqref{eq:fdecomp1}, the triangle inequality, Assumptions \ref{as:sieve} and \ref{as:opnorm}, and Lemma \ref{lem:opnorm} imply
	\begin{align*}
	& \sup_{(p_1, p_2) \in  \mathcal{S}^{(1,0)}}\left| \hat g^{(1, 0)}(p_1, p_2) -  g^{(1, 0)}(p_1, p_2)  \right| \\ 
	& \leq  \sup_{(p_1, p_2) \in  \mathcal{S}^{(1,0)}}\left| b_K(p_1, p_2)^\top ( \hat \alpha^{(1, 0)} - \alpha^{(1, 0)}) \right| +  \sup_{(p_1, p_2) \in  \mathcal{S}^{(1,0)}}\left| g^{(1, 0)}(p_1, p_2) -  b_K(p_1, p_2)^\top \alpha^{(1, 0)}  \right| \\
	& \leq  \sup_{(p_1, p_2) \in \mathcal{S}^{(1,0)}}\left| b_K(p_1, p_2)^\top \mathbb{S}_K \hat{\Psi}_{nK}^{-1} \hat{\mathbf{R}}_K^\top \breve{\mathbf{T}} \left(\mathbf{g}^{(1,0)} - \hat{\mathbf{g}}^{(1,0)}\right) / n \right|  \\
	& \quad + \sup_{(p_1, p_2) \in  \mathcal{S}^{(1,0)}}\left| b_K(p_1, p_2)^\top \mathbb{S}_K \hat \Psi_{nK}^{-1} \hat{\mathbf{R}}_K^\top\left(\breve{\mathbf{T}} - \hat{\mathbf{T}} \right) \hat{\mathbf{g}}^{(1,0)} / n \right| \\
	& \quad + \sup_{(p_1, p_2) \in  \mathcal{S}^{(1,0)}}\left| b_K(p_1, p_2)^\top \mathbb{S}_K \hat{\Psi}_{nK}^{-1} \hat{\mathbf{R}}_K^\top \hat{\mathbf{e}} / n \right| +  \underbrace{ \| \hat{\mathcal{P}}_{nK} \|_{\infty} \cdot O(K^{-s/2}) }_{O_P(K^{-s/2})} +  O(K^{-s/2}).
	\end{align*}
	
	By the mean value theorem, it is easy to see that
	\begin{align}\label{eq:supgdif}
	\begin{split}
		& \sup_{(p_1,p_2) \in \mathcal{S}^{(1,0)}}\left| b_K(p_1, p_2)^\top \mathbb{S}_K \hat \Psi_{nK}^{-1} \hat{\mathbf{R}}_K^\top \breve{\mathbf{T}} \left(\mathbf{g}^{(1,0)} - \hat{\mathbf{g}}^{(1,0)}\right) / n \right| \\
		& = \sup_{(p_1,p_2) \in \mathcal{S}^{(1,0)}}\left| b_K(p_1, p_2)^\top \mathbb{S}_K \hat \Psi_{nK}^{-1} \hat{\mathbf{R}}_K^\top \breve{\mathbf{T}} \left( \partial_1\bar{\mathbf{g}}^{(1, 0)} + \partial_2\bar{\mathbf{g}}^{(1, 0)} \right) / n \right| \cdot O_P(n^{-1/2}) \\
		& \le \| \hat{\mathcal{P}}_{nK} \|_{\infty} \cdot O_P(n^{-1/2})
		= O_P(n^{-1/2}),
	\end{split}
	\end{align}
 	and that
	\begin{align*}
		\sup_{(p_1,p_2) \in \mathcal{S}^{(1,0)}} \left| b_K(p_1, p_2)^\top \mathbb{S}_K \hat \Psi_{nK}^{-1} \hat{\mathbf{R}}_K^\top \left(\breve{\mathbf{T}} - \hat{\mathbf{T}} \right) \hat{\mathbf{g}}^{(1,0)} / n \right|= O_P(n^{-1/2}),
	\end{align*}
	by Lemma \ref{lem:opnorm}.
	Further, we can easily find that
	\begin{align*}
		\left| b_K(p_1, p_2)^\top \mathbb{S}_K \hat{\Psi}_{nK}^{-1} \hat{\mathbf{R}}_K^\top \hat{\mathbf{e}} / n\right| 
		& \leq \left| b_K(p_1, p_2)^\top \mathbb{S}_K \Psi_{nK}^{-1} \mathbf{R}_K^\top \tilde{\mathbf{e}} / n\right| + o_P(\zeta_0(K) \sqrt{(\log n)/n}),
	\end{align*}
	uniformly in $(p_1, p_2) \in \mathcal{S}^{(1,0)}$.
	Thus, we obtain
	\begin{align*}
	\sup_{(p_1, p_2) \in \mathcal{S}^{(1,0)}}\left| \hat g^{(1, 0)}(p_1, p_2) - g^{(1, 0)}(p_1, p_2)  \right|
	& \leq  \sup_{(p_1, p_2) \in \mathcal{S}^{(1,0)}} \left| b_K(p_1, p_2)^\top \mathbb{S}_K \Psi_{nK}^{-1} \mathbf{R}_K^\top \tilde{\mathbf{e}} / n \right| \\
	& \quad + o_P(\zeta_0(K)\sqrt{(\log n) /n }) +  O_P(K^{-s/2}).
	\end{align*} 
	Finally, the result follows from the fact that the first term on the right-hand side is of order $O_P(\zeta_0(K)\sqrt{\log n /n })$, as shown in the proof of result (i).
	\bigskip
	
	(iv) The proof of result (iv) is analogous to that of result (ii).
\end{proof}

\begin{lemma}\label{lem:unifconv2}
	Suppose that Assumptions \ref{as:gamemodel}, \ref{as:linear}, \ref{as:data}, \ref{as:1stage}, \ref{as:eigen}(i), and \ref{as:emom}--\ref{as:rate2} hold.
	Then, we have
	\begin{align*}
	\renewcommand{\arraystretch}{2}
	\begin{array}{cl}
	\text{(i)} & \displaystyle \sup_{(p_1, p_2) \in \mathcal{S}^{(1,0)}}\left| \tilde \bE_n[U_1^{(1, 0)} | V_1 = p_1, V_2 = p_2] -  \bE [U_1^{(1, 0)} | V_1 = p_1, V_2 = p_2]   \right| \\
	& \hspace{240pt} =  O_P(\zeta_0(K) K \sqrt{\log n / n} ) +  O_P(K^{(2 - s) /2}),\\
	\text{(ii)} & \displaystyle \sup_{(p_1, p_2) \in \mathcal{S}^{(1,0)}}\left| \hat \bE_n[U_1^{(1, 0)} | V_1 = p_1, V_2 = p_2] -  \bE [U_1^{(1, 0)} | V_1 = p_1, V_2 = p_2]   \right| \\
	& \hspace{240pt} =  O_P(\zeta_0(K) K \sqrt{\log n / n} ) +  O_P(K^{(2 - s) /2}).\\
	\end{array}
	\renewcommand{\arraystretch}{1}
	\end{align*}
\end{lemma}

\begin{proof}
	(i) The proof of result (i) is immediate from Lemma \ref{lem:unifconv1}(ii).
	\bigskip
	
	(ii) Assumptions \ref{as:1stage}(i) and (iii) imply that $\hat h(p_1, p_2)$ is uniformly consistent for $h(p_1, p_2)$ and that $\hat h(p_1, p_2)$ is uniformly bounded away from zero w.p.a.1.
	Thus, uniformly in $(p_1, p_2) \in \mathcal{S}^{(1,0)}$, we have
	\begin{align}\label{eq:hdif}
	\begin{split}
	\left|  \frac{1}{h(p_1, p_2)} - \frac{1}{\hat h(p_1, p_2)}   \right|  
	& =  \left|  \frac{\hat h(p_1, p_2) - h(p_1, p_2)}{h(p_1, p_2) \cdot \hat h(p_1, p_2)}   \right|  \\
	& \leq  O_P(1) \cdot \left|  \hat h(p_1, p_2) - h(p_1, p_2) \right| = O_P(1)  \cdot |\hat \rho - \rho^*| = O_P(n^{-1/2}).
	\end{split}
	\end{align}
	
	From the triangle and Cauchy--Schwarz inequalities, \eqref{eq:hdif}, and Lemma \ref{lem:unifconv1}(iv), it holds that
	\begin{align*}
	&\left| \hat \bE_n[U_1^{(1, 0)} | V_1 = p_1, V_2 = p_2] -  \bE [U_1^{(1, 0)} | V_1 = p_1, V_2 = p_2]\right| \\
	& \leq 	\left| \frac{1}{\hat h(p_1, p_2)} \partial_{p_1 p_2}\left( g^{(1,0)}(p_1, p_2)  - \hat g^{(1,0)}(p_1, p_2) \right) \right| + \left| \left( \frac{1}{h(p_1, p_2)} - \frac{1}{\hat h(p_1, p_2)} \right) \partial_{p_1 p_2} [g^{(1,0)}(p_1, p_2)]  \right| \\
	& = O_P(\zeta_0(K)K \sqrt{\log n/n}) + O_P(K^{(2 - s)/2}) + O_P(n^{-1/2}),
	\end{align*}
	uniformly in $(p_1, p_2) \in \mathcal{S}^{(1,0)}$.
\end{proof}

\setcounter{table}{0}
\setcounter{figure}{0}

\section{Appendix: Supplementary Technical Material}\label{appendix:supptech}

\subsection{Identification of the parametric game model in Section \ref{sec:estimation}}\label{subsec:ident_wo_inf}

In this appendix, we establish the identification of the parametric game model in Section \ref{sec:estimation}.
For exposition, we again display the likelihood function of our game model:
\begin{align*}
	&\begin{array}{ll}
		P_j^0(\gamma) = F_{\varepsilon_j}(W_j^\top \gamma_0), & \quad P_j^1(\gamma) = F_{\varepsilon_j}(W_j^\top \gamma_0 + \Delta(W_j^\top \gamma_1)), \\ [5 pt]
		\mathcal{L}^{(1, 0)}(\theta) = P_1^0(\gamma) - H_\rho(P_1^0(\gamma), P_2^1(\gamma)), & \quad
		\mathcal{L}^{(1, 1)}(\theta) = H_\rho(P_1^1(\gamma),P_2^1(\gamma)) - \Lambda(\tilde X^\top \lambda) \cdot \mathcal{L}_{\text{mul}}(\theta), \\ [5pt]
		\mathcal{L}^{(0, 1)}(\theta) = P_2^0(\gamma) - H_\rho(P_1^1(\gamma), P_2^0(\gamma)), & \quad
		\mathcal{L}^{(0, 0)}(\theta) = 1 - \sum_{(d_1,d_2) \neq (0,0)}\mathcal{L}^{(d_1, d_2)}(\theta),
	\end{array}
\end{align*}
where $\mathcal{L}_{\text{mul}}(\theta) = H_\rho(P_1^1(\gamma), P_2^1(\gamma)) - H_\rho(P_1^1(\gamma), P_2^0(\gamma)) - H_\rho(P_1^0(\gamma), P_2^1(\gamma)) + H_\rho(P_1^0(\gamma), P_2^0(\gamma))$.
Denote the true value of $\theta = (\gamma^\top, \lambda^\top, \rho)^\top$ as $\theta^* = (\gamma^{*\top}, \lambda^{*\top}, \rho^*)^\top$ and the parameter space for $\theta$ as $\Theta$.
The individual log-likelihood function is given by $\ell(\theta) \coloneqq \sum_{d_1, d_2} I^{(d_1, d_2)} \log \mathcal{L}^{(d_1, d_2)}(\theta)$.

Let $\tilde \theta \coloneqq (\gamma^\top, \lambda^\top)^\top$, and denote the parameter space for $\tilde \theta$ as $\tilde \Theta$.
For each $\rho \in (\underline{c}, \bar c)$, define
\begin{align}\label{eq:cml}
	\tilde \theta^*(\rho) 
	\coloneqq \argmax_{\tilde \theta \in \tilde \Theta} \; \bE[ \ell(\tilde \theta, \rho)].
\end{align}
Here, by a straightforward calculation with Jensen's inequality, we can show that $\theta^*$ is a maximizer of $\bE[ \ell(\theta)]$; however, its uniqueness is not ensured at this stage.
With this fact, if we can show that (i) $\tilde \theta^*(\rho)$ is unique for all $\rho$, this implies the equivalence $\theta^* = (\tilde \theta^*(\rho^*)^\top, \rho^*)^\top$.
Further, if (ii) $\rho^*$ is identifiable, $\theta^*$ is identified as $\theta^* = (\tilde \theta^*(\rho^*)^\top, \rho^*)^\top$, as desired.
We first establish result (i) under the following assumption.
\begin{assumption} \label{as:parametricid1}
	\hfil
	\begin{enumerate}[(i)]
		\item $\bE|\ell(\theta)| < \infty$ for all $\theta \in \Theta$.
		\item For each $\rho \in (\underline{c}, \bar c)$, the copula $H_{\rho}(\cdot, \cdot)$ is strictly increasing in its arguments.
		\item $\mathcal{L}_{\text{mul}}(\theta) > 0$ for all $\theta \in \Theta$.
		\item $\tilde X$ includes a constant and does not lie in a proper linear subspace of $\mathbb{R}^{\mathrm{dim}(\tilde X)-1}$ a.s.
		\item For all $\rho \in (\underline{c}, \bar c)$ and $\gamma \neq \gamma^*(\rho)$, at least one of the following probabilities is positive:
		\begin{align*}
			& \Pr \left[ W_1^\top \left( \gamma_0^*(\rho) - \gamma_0 \right) \ge 0, \; W_2^\top \left( \gamma_0^*(\rho) - \gamma_0 \right) + \left[ \Delta(W_2^\top \gamma_1^*(\rho)) - \Delta(W_2^\top \gamma_1) \right] < 0 \right], \\
			& \Pr \left[ W_1^\top \left( \gamma_0^*(\rho) - \gamma_0 \right) \le 0, \; W_2^\top \left( \gamma_0^*(\rho) - \gamma_0 \right) + \left[ \Delta(W_2^\top \gamma_1^*(\rho)) - \Delta(W_2^\top \gamma_1) \right] > 0 \right], \\
			& \Pr \left[ W_2^\top \left( \gamma_0^*(\rho) - \gamma_0 \right) \ge 0, \; W_1^\top \left( \gamma_0^*(\rho) - \gamma_0 \right) + \left[ \Delta(W_1^\top \gamma_1^*(\rho)) - \Delta(W_1^\top \gamma_1) \right] < 0 \right], \\
			& \Pr \left[ W_2^\top \left( \gamma_0^*(\rho) - \gamma_0 \right) \le 0, \; W_1^\top \left( \gamma_0^*(\rho) - \gamma_0 \right) + \left[ \Delta(W_1^\top \gamma_1^*(\rho)) - \Delta(W_1^\top \gamma_1) \right] > 0 \right].
		\end{align*}
	\end{enumerate}
\end{assumption}

Assumptions \ref{as:parametricid1}(i)--(ii) are mild regularity conditions.
Assumption \ref{as:parametricid1}(iii) requires that multiple equilibria occur under any parameter values.
Assumption \ref{as:parametricid1}(iv) is the standard rank condition.
In Assumption \ref{as:parametricid1}(v), we restrict the parameter space $\Theta$, the functional form of $\Delta$, and the distribution of $W$.
Importantly, this assumption may be satisfied without large support regressors when the parameter space of $\gamma_0$ is bounded and $\Delta$ is a bounded function.
For instance, if at least one of $\Pr[ W_1^\top a_0 \ge 0, \; W_2^\top a_0 + a_1 < 0 ]$ and $\Pr[ W_1^\top a_0 \le 0, \; W_2^\top a_0 + a_1 > 0 ]$ is positive for any non-zero $a_0$ and $a_1$ bounded by some large constants, Assumption \ref{as:parametricid1}(v) is satisfied, which is possible even when $(W_1, W_2)$ do not have full support.

\begin{theorem}\label{thm:parametricid1}
	Under Assumptions \ref{as:gamemodel} and \ref{as:parametricid1}, $\tilde \theta^*(\rho)$ is unique and identified for all $\rho \in (\underline{c}, \bar c)$.
\end{theorem}

The proof is relegated to the end of this subsection.
Now, the remaining task is to prove result (ii), the identifiability of $\rho^*$. 
To this end, define
\begin{align*}
	L^{(1,0)}(\rho, W) &\coloneqq P_1^0(\gamma^*(\rho)) - H_{\rho}[P_1^0(\gamma^*(\rho)), P_2^1(\gamma^*(\rho))], \\
	L^{(0,1)}(\rho, W) &\coloneqq P_2^0(\gamma^*(\rho)) - H_{\rho}[P_1^1(\gamma^*(\rho)), P_2^0(\gamma^*(\rho))], \\
	L^{(1,1)}(\rho, W) &\coloneqq H_{\rho}[P_1^1(\gamma^*(\rho)), P_2^1(\gamma^*(\rho))] - \Lambda(\tilde X^\top \lambda^*(\rho)) \cdot \mathcal{L}_{\text{mul}}(\tilde \theta^*(\rho), \rho), \\
	L^{(0,0)}(\rho, W) &\coloneqq 1 - \sum_{(d_1,d_2) \neq (0,0)} L^{(d_1,d_2)}(\rho, W).
\end{align*}
By construction, we have
\begin{align}\label{eq:uniquerho}
	\Pr[  D = (d_1, d_2) | W = w ] = L^{(d_1, d_2)}(\rho^*, w)
	\quad
	\text{for $(d_1,d_2) \in \{ 0, 1 \}^2$}.
\end{align}
Noting that the left-hand side is an observable quantity, $\rho^*$ is identified if some $w \in \text{supp}[W]$ and $(d_1,d_2) \in \{ 0, 1 \}^2$ exist such that this equation is uniquely solved at $\rho^*$.

To proceed, define
\begin{align*}
	\bar \ell^{(11)}(\tilde \theta, \rho) 
	\coloneqq \frac{\partial^2}{\partial \tilde \theta \partial \tilde \theta^\top} \bE[ \ell(\tilde \theta, \rho) ].
\end{align*}
The second-order condition for \eqref{eq:cml} implies that $\bar \ell^{(11)}(\tilde \theta^*(\rho), \rho)$ is negative semidefinite for all $\rho$.

\begin{assumption}\label{as:parametricid2}
	\hfil
	\begin{enumerate}[(i)]
		\item $\Theta$ is an open subset in Euclidean space.
		\item $F_{\varepsilon_1}$, $F_{\varepsilon_2}$, $\Delta$, and $\Lambda$ are twice continuously differentiable.
		\item $H_{\rho}(\cdot, \cdot)$ is twice continuously differentiable in its arguments and $\rho$.
		\item $ \bar \ell^{(11)}(\tilde \theta^*(\rho), \rho)$ is negative definite for all $\rho \in (\underline{c}, \bar c)$.
		\item There exist $w \in \text{supp}[W]$ and $(d_1,d_2) \in \{ 0, 1 \}^2$ such that $\partial L^{(d_1,d_2)}(\rho, w)/\partial \rho \neq 0$ for all $\rho \in (\underline{c}, \bar c)$.
	\end{enumerate}
\end{assumption}

Assumptions \ref{as:parametricid2}(i)--(iv) ensure that $L^{(d_1,d_2)}(\rho, W)$ is continuously differentiable in $\rho$.
Assumptions \ref{as:parametricid2}(iv) is a natural consequence from Theorem \ref{thm:parametricid1}.
In conjunction with the continuous differentiability, Assumption \ref{as:parametricid2}(v) implies that $L^{(d_1,d_2)}(\rho, w)$ is strictly monotonic in $\rho$ for some $w \in \text{supp}[W]$ and $(d_1,d_2) \in \{ 0, 1 \}^2$.
Then, $\rho^*$ can be identified as a unique solution of \eqref{eq:uniquerho}.

To clarify Assumption \ref{as:parametricid2}(v), consider the partial derivative of $L^{(1,0)}(\rho, W)$ with respect to $\rho$ for instance.
By a straightforward calculation, we can observe that
\begin{align} \label{eq:partialf1}
\begin{split}
	& \frac{\partial L^{(1,0)}(\rho, W)}{\partial \rho} \\
	& = W_1^\top \dot \gamma_0^*(\rho) f_{\varepsilon_1} \left( W_1^\top \gamma_0^*(\rho) \right) \left[ 1 - H_{\rho}^{(1)} \left[ P_1^0(\gamma^*(\rho)), P_2^1(\gamma^*(\rho)) \right] \right] \\ 
	& \quad - W_2^\top \left[ \dot \gamma_0^*(\rho) + \dot \gamma_1^*(\rho) \dot \Delta(W_2^\top \gamma_1^*(\rho)) \right] f_{\varepsilon_2} \left( W_2^\top \gamma_0^*(\rho) + \Delta(W_2^\top \gamma_1^*(\rho)) \right) H_{\rho}^{(2)} \left[ P_1^0(\gamma^*(\rho)), P_2^1(\gamma^*(\rho)) \right] \\
	& \quad - H_{\rho}^{(\rho)} \left[ P_1^0(\gamma^*(\rho)), P_2^1(\gamma^*(\rho)) \right],
\end{split}
\end{align}
where $H_{\rho}^{(1)}$, $H_{\rho}^{(2)}$, and $H_{\rho}^{(\rho)}$ are respectively the partial derivatives of $H_{\rho}(\cdot, \cdot)$ with respect to the first and second arguments and $\rho$, $f_{\varepsilon_j}$ and $\dot \Delta$ are respectively the derivatives of $F_{\varepsilon_j}$ and $\Delta$, $\dot \gamma_0^*(\rho) \coloneqq \partial \gamma_0^*(\rho) / \partial \rho$, and $\dot \gamma_1^*(\rho) \coloneqq \partial \gamma_1^*(\rho) / \partial \rho$.
From this expression, we can see that $\partial L^{(1,0)}(\rho, w)/\partial \rho \neq 0$ for all $\rho$ for some $w \in \text{supp}[W]$ if $(W_1, W_2)$ have sufficiently rich variation.
In particular, if the functions on the right-hand side of \eqref{eq:partialf1} are bounded and at least one element of $\dot \gamma_0^* f_{\varepsilon_1}  [1 - H_{\rho}^{(1)}]$ or $[\dot \gamma_0^* + \dot \gamma_1^* \dot \Delta] f_{\varepsilon_2} H_{\rho}^{(2)}$ has a non-zero constant sign, Assumption \ref{as:parametricid2}(v) holds for $w \in \text{supp}[W]$ satisfying $w_1^\top a_1 - w_2^\top a_2 - a_3 < 0$ for all $(a_1, a_2, a_3)$ in the ranges of $\dot \gamma_0^* f_{\varepsilon_1}  [ 1 - H_{\rho}^{(1)}]$, $[ \dot \gamma_0^* + \dot \gamma_1^* \dot \Delta ] f_{\varepsilon_2} H_{\rho}^{(2)}$, and $H_{\rho}^{(\rho)}$, respectively.
Again, the existence of large support regressors is not necessary to meet such conditions.

\begin{theorem}\label{thm:parametricid2}
	Under Assumptions \ref{as:gamemodel}, \ref{as:parametricid1}, and \ref{as:parametricid2}, $\rho^*$ is identified.
	As a result, $\theta^*$ is identified by $\theta^* = (\tilde \theta^*(\rho^*)^\top, \rho^*)^\top$ as a unique maximizer of $\bE[\ell(\theta)]$.
\end{theorem}

\subsubsection{Proof of Theorem \ref{thm:parametricid1}}

Fix arbitrary $\rho \in (\underline{c}, \bar c)$ and denote $\theta^*(\rho) = (\tilde \theta^*(\rho)^\top, \rho)^\top = (\gamma^*(\rho)^\top, \lambda^*(\rho)^\top, \rho)^\top$.
First, suppose that $\gamma^*(\rho)$ is identified.
It is straightforward from the functional form of $\mathcal{L}^{(1, 1)}(\theta)$ and Assumption \ref{as:parametricid1}(iii) to see that $\Lambda(\tilde x^\top \lambda^*(\rho))$ is also identified for all $\tilde x \in \text{supp}[\tilde X]$.
Further, since $\Lambda^{-1}$ exists by Assumption \ref{as:gamemodel}(iii), $\lambda^*(\rho)$ is identified under Assumption \ref{as:parametricid1}(iv).
Consequently, with Assumption \ref{as:parametricid1}(i) and Jensen's inequality, we can show that $\tilde \theta^*(\rho)$ is a unique maximizer of $\bE[\ell(\tilde \theta, \rho)]$ for each $\rho$.

We next prove that $\gamma^*(\rho)$ is indeed identifiable by showing
\begin{align}\label{eq:distinct}
	\Pr\left[ \left( \mathcal{L}^{(1, 0)}(\theta^*(\rho)) \neq \mathcal{L}^{(1, 0)}(\theta) \right) \lor \left( \mathcal{L}^{(0, 1)}(\theta^*(\rho)) \neq \mathcal{L}^{(0, 1)}(\theta) \right) \right] > 0
\end{align}
for any $\theta = (\gamma^\top, \lambda^\top, \rho)^\top$ such that $\gamma \neq \gamma^*(\rho)$ (notice that $\rho$ is common to $\theta$ and $\theta^*(\rho)$).
Since $\mathcal{L}^{(1, 0)}(\theta)$ and $\mathcal{L}^{(0, 1)}(\theta)$ do not depend on $\lambda$, the above inequality means that the distribution of the observed variables generated by $\gamma^*(\rho)$ differs from that generated by $\gamma$, implying that $\gamma^*(\rho)$ is observationally distinct from $\gamma$.

Because $\mathcal{L}^{(1, 0)}(\theta)$ is strictly increasing and decreasing in $P_1^0(\gamma)$ and $P_2^1(\gamma)$, respectively, we have 
\begin{align*}
	\left( P_1^0(\gamma^*(\rho)) \ge P_1^0(\gamma) \right) \land \left( P_2^1(\gamma^*(\rho)) < P_2^1(\gamma) \right) 
	& \quad \Longrightarrow \quad \mathcal{L}^{(1, 0)}(\theta^*(\rho)) > \mathcal{L}^{(1, 0)}(\theta),\\
	\left( P_1^0(\gamma^*(\rho)) \le P_1^0(\gamma) \right) \land \left( P_2^1(\gamma^*(\rho)) > P_2^1(\gamma) \right) 
	& \quad \Longrightarrow \quad \mathcal{L}^{(1, 0)}(\theta^*(\rho)) < \mathcal{L}^{(1, 0)}(\theta).
\end{align*}
Similarly, changing the roles of players 1 and 2 leads to
\begin{align*}
	\left( P_2^0(\gamma^*(\rho)) \ge P_2^0(\gamma) \right) \land \left( P_1^1(\gamma^*(\rho)) < P_1^1(\gamma) \right) 
	& \quad \Longrightarrow \quad \mathcal{L}^{(0, 1)}(\theta^*(\rho)) > \mathcal{L}^{(0, 1)}(\theta),\\
	\left( P_2^0(\gamma^*(\rho)) \le P_2^0(\gamma) \right) \land \left( P_1^1(\gamma^*(\rho)) > P_1^1(\gamma) \right) 
	& \quad \Longrightarrow \quad \mathcal{L}^{(0, 1)}(\theta^*(\rho)) < \mathcal{L}^{(0, 1)}(\theta).
\end{align*}
Because $F_{\varepsilon_j}$ is strictly increasing by Assumption \ref{as:gamemodel}(ii), the probabilities of the left-hand sides of these are equivalent to the probabilities given in Assumption \ref{as:parametricid1}(v).
Thus, \eqref{eq:distinct} holds.
\qed

\subsubsection{Proof of Theorem \ref{thm:parametricid2}}

We provide the proof when Assumption \ref{as:parametricid2}(v) is satisfied with $(d_1,d_2)=(1,0)$; the other cases are analogous.
First, we prove the continuous differentiability of $L^{(1,0)}(\rho, W)$ in $\rho$.
Because the terms other than $\dot{\gamma}_0^*(\rho)$ and $\dot{\gamma}_1^*(\rho)$ in \eqref{eq:partialf1} are continuous in $\rho$ under Assumptions \ref{as:parametricid2}(ii) and (iii), it suffices to prove the continuous differentiability of $\gamma^*(\rho)$.
By the first-order condition for \eqref{eq:cml}, we have
\begin{align*}
	\frac{\partial \bE[ \ell(\tilde \theta^*(\rho), \rho) ]}{\partial \tilde \theta}  = \mathbf{0}_{\mathrm{dim}(\tilde \theta)}
\end{align*}
for all $\rho \in (\underline{c}, \bar c)$.
By the implicit function theorem with Assumptions \ref{as:parametricid2}(i)--(iv), it holds that $\tilde \theta^*(\rho)$ is continuously differentiable with respect to $\rho$ in the following form:
\begin{align*}
	\frac{\partial \tilde \theta^*(\rho)}{\partial \rho}
	= - \left[ \bar \ell^{(11)}(\tilde \theta^*(\rho), \rho) \right]^{-1} \bar \ell^{(12)}(\tilde \theta^*(\rho), \rho),
\end{align*}
where $\bar \ell^{(12)}(\tilde \theta, \rho) \coloneqq \partial^2 \bE[ \ell(\tilde \theta, \rho) ] / (\partial \tilde \theta \partial \rho)$.
This implies that $\gamma^*(\rho)$ is continuously differentiable in $\rho$, and thus $\partial L^{(1,0)}(\rho, W)/\partial \rho$ is well defined.

Now, consider some $w \in \text{supp}[W]$ satisfying Assumption \ref{as:parametricid2}(v) for $(d_1,d_2) = (1,0)$.
For such $w$, for all $\rho$ either $\partial L^{(1,0)}(\rho, w) / \partial \rho > 0$ or $\partial L^{(1,0)}(\rho, w) / \partial \rho < 0$ is true, implying the strict monotonicity of $L^{(1,0)}(\rho, w)$ in $\rho$.
As a result, $\rho^*$ is identified as a unique solution of \eqref{eq:uniquerho} for $(d_1,d_2) = (1,0)$.
Thus, in conjunction with Theorem \ref{thm:parametricid1}, $\theta^*$ is identified by $\theta^* = (\tilde \theta^*(\rho^*)^\top, \rho^*)^\top$ as a unique maximizer of $\bE[ \ell(\theta) ]$.
\qed

\subsection{Identification of the marginal treatment response function for non-strategic models} \label{subsec:nointeraction}

When there is no strategic interaction, we can identify the MTR functions in a straightforward manner.
Suppose that the payoff function for player $j$ is given by $u_j(d_j, d_{-j}) = d_j [ \pi_j(W_j) - \varepsilon_j ]$ for $(d_j, d_{-j}) \in \{ 0, 1 \}^2$.
Based on the payoff maximization principle, the treatment equation is given by $D_j = \mathbf{1} \{ \pi_j(W_j) \ge \varepsilon_j \}$.
This can be rewritten as $D_j = \mathbf{1}\{ P_j \ge V_j \}$, where $P_j = F_{\varepsilon_j | X_j}( \pi_j(W_j) )$ and $V_j = F_{\varepsilon_j | X_j}( \varepsilon_j )$. 
Because the model is no longer a ``game'', conditioning on the opponent's $X_{-j}$ is unnecessary. 
By construction, $V_j$ is distributed as $\text{Uniform}[0, 1]$ given $X_j$, implying that $P_j$ is identified by $\bE [ D_j | W_j ] = P_j$ under Assumption \ref{as:IV}.

The identification strategy in this study can be directly applied to identifying the MTR $m^{(d_1, d_2)}(x, p_1, p_2) = \bE [Y_1^{(d_1, d_2)} | X = x, V_1 = p_1, V_2 = p_2]$ for all $(d_1, d_2) \in \{0,1\}^2$.
In particular, because the current model is a complete econometric model, the problem of multiple equilibria does not occur. 
Below, we present the identification result only for the case of $(d_1, d_2) = (0, 0)$; the other cases are analogous.
Letting $\psi^{(d_1, d_2)} (x, p_1, p_2) = \bE [ I^{(d_1, d_2)} Y_1 | X = x, P_1 = p_1, P_2 = p_2]$, we observe that
\begin{align*}
	\psi^{(0, 0)} (x, p_1, p_2)
	& = \bE [ Y_1^{(0, 0)} | D = (0, 0), X = x, P_1 = p_1, P_2 = p_2 ] \Pr[D = (0, 0) | X = x, P_1 = p_1, P_2 = p_2]\\
	& = \bE [ Y_1^{(0, 0)} | V_1 > p_1, V_2 > p_2, X = x ] \Pr[ V_1 > p_1, V_2 > p_2 | X = x] \\
	& = \int_{p_2}^1 \int_{p_1}^1 m^{(0, 0)}(x, v_1, v_2) h(v_1, v_2 | x) \mathrm{d}v_1 \mathrm{d}v_2,
\end{align*}
where the second equality follows from Assumption \ref{as:IV}(i).
Thus, we have
\begin{align*}
	m^{(0, 0)}(x, p_1, p_2) 
	= \frac{\partial_{p_1 p_2} [\psi^{(0, 0)} (x, p_1, p_2)]}{h(p_1, p_2 | x)} .
\end{align*}
This equality implies that $m^{(0, 0)}(x, p_1, p_2)$ is identified for any $(p_1, p_2) \in \text{supp}[P_1, P_2 | X = x ]$, provided that $h(p_1, p_2 | x)$ is identified or known.

\subsection{The estimator of the total marginal treatment effect and its limiting distribution}\label{subsec:totalMTE}

We can estimate $\text{MTE}_{\text{total}}(x, p^0_1, p^0_2) = m^{(1,1)}(x, p_1^0, p_2^0) - m^{(0,0)}(x, p_1^0, p_2^0)$ in the same manner as described in Section \ref{sec:estimation}.
The estimator of $m^{(0,0)}(x, p_1^0, p_2^0)$ can be obtained by
\begin{align*}
	\hat m^{(0, 0)}(x, p_1^0, p_2^0) \coloneqq x_1^\top \hat \beta_1^{(0, 0)} + \frac{\ddot{b}_K(p_1^0, p_2^0)^\top \hat \alpha_1^{(0, 0)}}{\hat h(p_1^0, p_2^0)} .
\end{align*}
The definitions of $\hat \beta_1^{(0, 0)}$ and $\hat \alpha_1^{(0, 0)}$ can be found in \eqref{eq:coef00}.
For the estimation of $m^{(1,1)}(x, p_1^0, p_2^0)$, we first estimate the following partially linear additive regression model:
\begin{align*}
\begin{split}
	\tilde I^{(1, 1)} Y_1 = \tilde I^{(1, 1)} X_{1}^\top \beta_1^{(1, 1)} 
	& + \tilde T^{(1, 1)} \lambda_X^* g_1^{(1, 1)}(P_1^0, P_2^0) + \tilde T^{(1, 1)} \lambda_X^* g_2^{(1, 1)}(P_1^1, P_2^0) \\
	& + \tilde T^{(1, 1)} \lambda_X^* g_3^{(1, 1)}(P_1^0, P_2^1) + \tilde T^{(1, 1)} (1 - \lambda_X^*) g_4^{(1, 1)}(P_1^1, P_2^1) + \tilde e^{(1, 1)},
\end{split}
\end{align*}
where $\tilde I^{(1, 1)} \coloneqq \tau_\varpi(\mathcal{L}^{(1, 1)}) I^{(1, 1)}$, $\tilde T^{(1, 1)} \coloneqq \tilde I^{(1, 1)} / \mathcal{L}^{(1,1)}$, and $\bE [ \tilde e^{(1, 1)} | I^{(1, 1)}, X, \mathbf{P}] = 0$.\footnote{
	The definitions of $g_l^{(1,1)}$'s are as follows: $g_1^{(1,1)}(p_1^0, p_2^0) = -\int_0^{p_1^0} \int_0^{p_2^0} \bE [U_{1}^{(1, 1)}| V_1= v_1, V_2= v_2] h(v_1, v_2) \mathrm{d}v_1 \mathrm{d}v_2$, $g_2^{(1,1)}(p_1^1, p_2^0) = \int_0^{p_1^1} \int_0^{p_2^0} \bE [U_{1}^{(1, 1)}| V_1= v_1, V_2= v_2] h(v_1, v_2) \mathrm{d}v_1 \mathrm{d}v_2$, and $g_3^{(1,1)}(p_1^0, p_2^1)$ and $g_4^{(1,1)}(p_1^1, p_2^1)$ are obtained by replacing $(p_1^1, p_2^0)$ in the right-hand side of $g_2^{(1,1)}(p_1^1, p_2^0)$ with $(p_1^0, p_2^1)$ and $(p_1^1, p_2^1)$, respectively.
}
Using the series approximation $g_l^{(1,1)}(\cdot, \cdot) \approx b_K(\cdot, \cdot)^\top \alpha_l^{(1,1)}$ for each $l = 1, \ldots, 4$, the same estimation procedure as in the case of $D = (0,0)$ gives the LS estimator $(\hat \beta_1^{(1,1)}, \hat \alpha_l^{(1, 1)})$ of $(\beta_1^{(1, 1)}, \alpha_l^{(1,1)})$.
Then, we can estimate $m^{(1, 1)}(x, p_1^0, p_2^0)$ as
\begin{align*}
	\hat m^{(1, 1)}(x, p_1^0, p_2^0) \coloneqq x_1^\top \hat \beta_1^{(1, 1)}-\frac{\ddot{b}_K(p_1^0, p_2^0)^\top \hat \alpha_1^{(1, 1)}}{\hat h(p_1^0, p_2^0)}.
\end{align*}
Finally, $\text{MTE}_{\text{total}}(x, p^0_1, p^0_2)$ can be estimated by
\begin{align*}
	 \widehat{\text{MTE}}_{\text{total}}(x, p^0_1, p^0_2) \coloneqq \hat m^{(1, 1)}(x, p_1^0, p_2^0) - \hat m^{(0, 0)}(x, p_1^0, p_2^0).
\end{align*}

Under similar conditions to those in Theorem \ref{thm:normal1}, we can have
\begin{align*}
	\frac{\sqrt{n} \left( \hat m^{(0, 0)}(x, p_1^0, p_2^0)  -  m^{(0, 0)}(x, p_1^0, p_2^0) \right)}{\sigma_{K}^{(0,0)}(p_1^0, p_2^0)} \overset{d}{\to} N(0, 1) \;\; \text{and} \;\; \frac{\sqrt{n} \left( \hat m^{(1, 1)}(x, p_1^0, p_2^0)  -  m^{(1, 1)}(x, p_1^0, p_2^0) \right)}{\sigma_{K}^{(1,1)}(p_1^0, p_2^0)} \overset{d}{\to} N(0, 1), 
\end{align*}
where
\begin{align*}
	\sigma_{K}^{(0,0)}(p_1^0, p_2^0) &\coloneqq \frac{\sqrt{ \ddot{b}_K(p_1^0, p_2^0)^\top \mathbb{S}_{1K} \left[ \Psi_K^{(0,0)}\right]^{-1} \Sigma^{(0, 0)}_K  \left[ \Psi_K^{(0,0)}\right]^{-1} \mathbb{S}_{1K}^\top \ddot{b}_K(p_1^0, p_2^0)}}{h(p_1^0, p_2^0) },\\
	\sigma_{K}^{(1,1)}(p_1^0, p_2^0) &\coloneqq \frac{\sqrt{ \ddot{b}_K(p_1^0, p_2^0)^\top \mathbb{S}_{1K} \left[ \Psi_K^{(1,1)}\right]^{-1} \Sigma^{(1, 1)}_K  \left[ \Psi_K^{(1,1)}\right]^{-1} \mathbb{S}_{1K}^\top \ddot{b}_K(p_1^0, p_2^0)}}{h(p_1^0, p_2^0) },
\end{align*}
and $\mathbb{S}_{1K} \coloneqq (\mathbf{0}_{K \times \mathrm{dim}(X)}, \mathbf{I}_K, \mathbf{0}_{K \times 3K} )$.
The definitions of the matrices $\Psi_K^{(1,1)}$ and $\Sigma_K^{(1,1)}$ are clear from the context.
Consequently, the limiting distribution of $\widehat{\text{MTE}}_{\text{total}}(x, p_1^0, p_2^0)$ can be given by
\begin{align*}
    \frac{\sqrt{n} \left( \widehat{\text{MTE}}_{\text{total}}(x, p_1^0, p_2^0)  -  \text{MTE}_{\text{total}}(x, p_1^0, p_2^0) \right)}{\sqrt{\left[\sigma_K^{(1,1)}(p_1^0, p_2^0)\right]^2 + \left[\sigma_{K}^{(0,0)}(p_1^0, p_2^0)\right]^2}} \overset{d}{\to} N(0, 1).
\end{align*}

\subsection{Alternative estimation strategies for $m^{(0,0)}(x, p_1, p_2)$ and $m^{(1,1)}(x, p_1, p_2)$}\label{subsec:alt_m00}

\Copy{R2-7-b-2}{
	Here, we describe two alternative estimators for $m^{(0, 0)}(x, p_1, p_2) = x_1^\top \beta_1^{(0, 0)} + \bE [U_1^{(0, 0)} | V_1 = p_1, V_2 = p_2]$.
	The same discussions as below apply to the estimation of $m^{(1, 1)}(x, p_1, p_2)$.
	
	Let 
	\begin{align*}
		g^{(0,0)}(p_1, p_2) \coloneqq \int_{p_1}^1 \int_{p_2}^1 \bE [U_{1}^{(0, 0)}| V_1= v_1, V_2= v_2] h(v_1, v_2) \mathrm{d}v_1 \mathrm{d}v_2.
	\end{align*}
	Then, recalling \eqref{eq:additive} and the discussion in Footnote \ref{foot:gfunc}, we can observe that
	\begin{align*}
		\begin{split}
			& \bE [U_{1}^{(0, 0)}| I^{(0, 0)} = 1, X, \mathbf{P}]  \\
			& = \frac{\lambda_X^* g^{(0, 0)}(P_1^0, P_2^0) + (1 - \lambda_X^*) g^{(0, 0)}(P_1^1, P_2^0) + (1 - \lambda_X^*) g^{(0, 0)}(P_1^0, P_2^1) - (1 - \lambda_X^*) g^{(0, 0)}(P_1^1, P_2^1)}{\mathcal{L}^{(0, 0)}(\mathbf{P})}.
		\end{split}
	\end{align*}
	Hence, assuming that $g^{(0,0)}(p_1, p_2)$ can be uniformly approximated by $g^{(0,0)}(\cdot, \cdot) \approx b_K(\cdot, \cdot)^\top \alpha^{(0,0)}$, we obtain the following regression model:
	\begin{align*}
		\begin{split}
			\tilde I^{(0, 0)} Y_1
			\approx \tilde I^{(0, 0)} X_1^\top \beta_1^{(0, 0)} 
			& + \tilde T^{(0, 0)} B_K(\mathbf{P})^\top \alpha^{(0,0)} + \tilde e^{(0, 0)},
		\end{split}
	\end{align*}
	where 
	\begin{align*}
		B_K(\mathbf{P}) \coloneqq \lambda_X^* b_K(P_1^0, P_2^0) + (1 - \lambda_X^*) b_K(P_1^1, P_2^0) + (1 - \lambda_X^*) b_K(P_1^0, P_2^1) - (1 - \lambda_X^*) b_K(P_1^1, P_2^1).
	\end{align*}
	Then, we can obtain LS estimators for $\beta_1^{(0, 0)}$ and $\alpha^{(0,0)}$, say $\hat \beta_1^{(0, 0)}$ and $\hat \alpha^{(0,0)}$, in a similar manner to our main estimator.
	Consequently, $\bE [U_{1}^{(0, 0)}| V_1= p_1, V_2= p_2]$ can be estimated by
	\begin{align*}
		\begin{split}
			\hat \bE_n[U_1^{(0, 0)} | V_1 = p_1, V_2 = p_2] 
			= \frac{\ddot{b}_K(p_1, p_2)^\top \hat \alpha^{(0, 0)}}{\hat h(p_1, p_2)}.
		\end{split}
	\end{align*}
	
	Several comments on this estimation procedure are as follows.
	First, it is clear that this estimation method does not involve the post-estimation adjustment required in our over-identified estimator.
	Therefore, this method is easier to implement than our proposed method.
	Second, to the best of our knowledge, this estimator gives a new class of series regression models, namely, generalized additive models with a prior knowledge that the additive functions share the same functional form.
	How to choose appropriate basis functions for this model is a difficult problem.
	In particular, when using splines basis, where to place the inner knots is a critical issue that affects the theoretical properties of the spline approximation (cf. \citealpappendix{huang2003local}).
	Third, when we numerically compared the performance of this estimator with that of our over-identified estimator, there was no significant difference, and we cannot conclude which is superior.\footnote{
		In this numerical simulation, we simply located the inner knots uniformly on $[0,1]^2$.
		Detailed simulation results are available on request.
	}
	This result would not be surprising because the alternative estimator also utilizes the identification power of all four points (A, B, C, and D in Figure \ref{fig:nash-comp2}) for the estimation of $g^{(0,0)}(p_1, p_2)$, as long as $(p_1, p_2) \in \underline{\mathcal{S}}$.
	
	Lastly, we suggest one more alternative estimator for $m^{(0, 0)}(x, p_1, p_2)$.
	That is, we can consider directly series expand $\bE [U_1^{(0, 0)} | V_1 = p_1, V_2 = p_2]$ as $\bE [U_1^{(0, 0)} | V_1 = p_1, V_2 = p_2] \approx b_K(p_1, p_2)^\top \alpha^{(0,0)}$ (with abuse of notations).
	Then, the function $g^{(0,0)}(p_1, p_2)$ would be approximated by $g^{(0,0)}(p_1, p_2) \approx \tilde b_K(p_1, p_2)^\top \alpha^{(0,0)}$, with $\tilde b_K(p_1, p_2) \coloneqq \int_{p_1}^1 \int_{p_2}^1 b_K(v_1, v_2) h(v_1, v_2) \mathrm{d}v_1 \mathrm{d}v_2$.
	The estimation of $\alpha^{(0,0)}$ is straightforward by LS regression on the ``integrated'' basis functions $\tilde b_K(p_1, p_2)$.
	A similar idea to this estimator can be found in \citeappendix{mogstad2018using} in a different context, but we believe that the theoretical properties of this type of series regression is largely unknown.
	It would be quite intriguing to investigate the properties of these alternative estimators, but this should be discussed in another study.
}

\subsection{Verification of $\zeta_1(K) = O(K)$ for tensor-product B-splines}\label{subsec:zeta}

We consider univariate B-splines of order $r$ with quasi-uniform $k$ internal knots, i.e., $b_r(p) = (b_{r,1}(p), \ldots  b_{r,k + r}(p))^\top$ for $p \in [0,1]$, where the length of each knot interval is proportional to $1/k$ with the internal knots at $\{t_j\}_{j=1}^k$.
In the notation of the main text, $b_K(p_1, p_2) = b_r(p_1) \otimes b_r(p_2)$ such that $k^2 \asymp K$.
As is well known, the derivatives of B-spline functions can be simply expressed in terms of lower-order B-spline functions.
Specifically, the first derivative of $b_r(p)$ can be written as $\partial b_r(p) / \partial p = (r - 1) \Omega_r b_{r-1}(p)$ (e.g., \citealpappendix{zhou2000derivative}), where
\begin{align*}
	\Omega_r \coloneqq  \underbrace{\left( 
	\begin{array}{cccccc}
		\frac{-1}{t_1 - t_{2 - r}} & 0                                         & 0                                           & \cdots & 0           & 0 \\
		\frac{1}{t_1 - t_{2 - r}}  & \frac{-1}{t_2 - t_{3 - r}} & 0                                           & \cdots & 0           & 0  \\
		0                                         & \frac{1}{t_2 - t_{3 - r}}  &  \frac{-1}{t_3 - t_{4 - r}}  & \cdots & 0           & 0  \\
		\vdots                                & \vdots                                  & \vdots                                  & \ddots & \vdots & \vdots \\
		0                                         & 0                                          & 0                                          & \cdots & 0           & \frac{1}{t_{k + r - 1} - t_k}
	\end{array}\right)}_{(k +r) \times (k + r -1)}.
\end{align*}
Then, since $\chi_{\mathrm{max}}(\Omega_r^\top \Omega_r) = O(k^2)$, we have $\| \partial b_r(p) / \partial p \|^2 = (r - 1)^2 \cdot b^\top_{r-1}(p) \Omega_r^\top \Omega_r b_{r-1}(p) \leq O(k^2 ) \|b_{r-1}(p)\|^2 = O(k^3)$ uniformly in $p$.
Hence,
\begin{align*}
	\left\| \frac{\partial b_K(p_1, p_2)}{\partial p_1} \right\|^2  = \left\| \frac{\partial b_r(p_1)}{\partial p_1} \otimes b_r(p_2) \right\|^2 
	& = \left\| \frac{\partial b_r(p_1)}{\partial p_1} \right\|^2 \cdot \left\| b_r(p_2) \right\|^2 \\
	& \leq O(k^4) = O(K^2),
\end{align*}
uniformly in $(p_1, p_2) \in [0,1]^2$.
This implies the desired result. 

\subsection{Verification of Assumption \ref{as:opnorm}(i)}\label{subsec:opnorm}

Below is a sketch of how we can verify Assumption \ref{as:opnorm}(i) for basis functions with a locally supported polynomial structure, such as B-splines and the partitioning polynomial series in \citeappendix{cattaneo2013optimal}.
Let $\mathbf{g} \coloneqq (g(Q_1), \ldots, g(Q_n))^\top$ for $g \in \mathcal{C}(\text{supp}[Q])$. 
Assuming that \eqref{eq:eigenPsi} holds, we can observe that
\begin{align*}
	\left| \left(\mathcal{P}_{nK}g\right)(p_1,p_2) \right|^2
	& = n^{-2} \cdot \mathrm{tr}\left\{ \Psi_{nK}^{-1} \mathbf{R}_K\mathbf{g} \mathbf{g}^\top \mathbf{R}_K \Psi_{nK}^{-1} \mathbb{S}_K^\top b_K(p_1, p_2) b_K(p_1, p_2)^\top \mathbb{S}_K  \right\}\\
	& \le O_P(n^{-2}) \cdot \mathrm{tr}\left\{\mathbb{S}_K^\top b_K(p_1, p_2) b_K(p_1, p_2)^\top \mathbb{S}_K \mathbf{R}_K\mathbf{g} \mathbf{g}^\top \mathbf{R}_K  \right\} \\
	&= O_P(1) \cdot  \left| \frac{1}{n}\sum_{i=1}^n b_K(p_1, p_2)^\top \mathbb{S}_K R_{K,i}g(Q_i)  \right|^2,
\end{align*}
where the superscript $(1,0)$ is suppressed for simplicity.
Here, by appropriately placing the partitioning knots based on the quantiles of $(P_1^0, P_2^1)$, for any $(p_1, p_2) \in \mathcal{S}^{(1,0)}$,  we can find a cell $\mathcal{R}(p_1, p_2) \subset \mathcal{S}^{(1,0)}$ such that $b_K(p_1, p_2)^\top b_K(P_1^0, P_2^1) = 0$ if $(P_1^0, P_2^1) \not\in \mathcal{R}(p_1, p_2)$ and $\Pr[ (P_1^0, P_2^1)  \in \mathcal{R}(p_1, p_2) ] = O(1/K)$.
Then, we have
\begin{align*}
	\left| \left(\mathcal{P}_{nK}g\right)(p_1,p_2) \right|
	& \le O_P(1) \cdot   \frac{1}{n}\sum_{i=1}^n \left| b_K(p_1, p_2)^\top b_K(P_{1i}^0, P_{2i}^1) g(Q_i)  \right| \\
	& \le O_P(1) \cdot  \frac{1}{n}\sum_{i=1}^n \mathbf{1}\{ (P_{1i}^0, P_{2i}^1)  \in \mathcal{R}(p_1, p_2)\} \underbrace{\left| b_K(p_1, p_2)^\top b_K( P_{1i}^0, P_{2i}^1) \right|}_{\zeta^2_0(K) = O(K)} \sup_{q \in \text{supp}[Q]} | g(q) |  = O_P(1).
\end{align*}
Since the above inequality holds for any $(p_1, p_2) \in \mathcal{S}^{(1,0)}$, we obtain the desired result.

\section{Appendix: Identification of Several Treatment Parameters} \label{sec:severalparameter}

\subsection{Individual-specific treatment effects}

It may be of interest to estimate the treatment effects when only player 1's treatment status switches, whereas that of player 2 is unspecified and subject to change endogenously. 
The parameter of interest in this situation would be $\text{MTE}_{\text{indiv}}(x, p_1, p_2) \coloneqq \bE [Y_1^{(1, D_2)} - Y_1^{(0, D_2)}| X = x, V_1 = p_1, V_2 = p_2]$ where $Y_1^{(d_1, D_2)} \coloneqq (1 - D_2) Y_1^{(d_1, 0)} + D_2 Y_1^{(d_1, 1)}$.
We call this MTE parameter the \textit{individual} MTE.\footnote{
	This is somewhat similar to the framework in \citetappendix{frolich2017direct}, where the identification of causal models that allow the presence of an endogenous ``mediator'' variable was investigated. 
	In our model, $D_2$ may be regarded as the mediator of $D_1$.
}
Let $m^{(d_1, D_2)}(x, p_1, p_2) \coloneqq \bE [Y_1^{(d_1, D_2)} | X = x, V_1 = p_1, V_2 = p_2]$.
After some calculations, we can show that
\begin{alignat*}{2}
	m^{(0, D_2)}(x, p_1^{d_1}, p_2^{d_2}) 
	&= m^{(0,0)}(x, p_1^{d_1}, p_2^{d_2}) && \quad \text{for $(d_1, d_2) \neq (1,0)$},\\
	m^{(0, D_2)}(x, p_1^1, p_2^0) 
	&= \frac{1}{\lambda_x}\left( m^{(0,1)}(x, p_1^1, p_2^0) - (1 - \lambda_x) \cdot m^{(0, 0)}(x, p_1^1, p_2^0) \right) && \quad \text{for $\lambda_x > 0$},\\
	m^{(1, D_2)}(x, p_1^{d_1}, p_2^{d_2}) 
	&= m^{(1,1)}(x, p_1^{d_1}, p_2^{d_2}) && \quad \text{for $(d_1, d_2) \neq (0,1)$},\\
	m^{(1, D_2)}(x, p_1^0, p_2^1)
	&= \frac{1}{(1 - \lambda_x)} \left( m^{(1,0)}(x, p_1^0, p_2^1) - \lambda_x \cdot m^{(1, 1)}(x, p_1^0, p_2^1) \right) && \quad \text{for $\lambda_x < 1$},
\end{alignat*}
under Assumptions \ref{as:complement}, \ref{as:multiple}, and \ref{as:IV}.
This implies that the individual MTE can be identified as follows:
\begin{align*}
	\text{MTE}_{\text{indiv}}(x, p_1^0, p_2^0)  
	& = \text{MTE}_{\text{total}}(x, p_1^0, p_2^0), \\
	\text{MTE}_{\text{indiv}}(x, p_1^0, p_2^1)  
	& = \frac{1}{1 - \lambda_x} \text{MTE}_{\text{direct}}^{(0)}(x, p_1^0, p_2^1) - \frac{\lambda_x}{1 - \lambda_x} \text{MTE}_{\text{total}}(x, p_1^0, p_2^1),
\end{align*}
for example.
Thus, for the estimation of the individual MTE, it is sufficient to calculate the direct MTE and total MTE, so that no additional estimation is required.

\subsection{LATE}

It is also possible to identify the LATE: average causal effect for agents whose treatment status is strictly altered by the IVs.
To define the LATE parameters in our context, we consider $z, z' \in \text{supp}[Z]$.
Suppose that the values of $\mathbf{P}$ are $\mathbf{p} = (p_1^0, p_1^1, p_2^0, p_2^1)$ when $Z = z$, and $\mathbf{p}' = (p_1^{0'}, p_1^{1'}, p_2^{0'}, p_2^{1'})$ when $Z = z'$.
For illustrative purposes, we assume that $\mathbf{p} > \mathbf{p}'$ (where the inequality is element-wise), as depicted in Figure \ref{fig:LATE1}.
Although we can consider several different LATE parameters, as examples, we here focus on the \textit{direct} LATE: $\bE [Y_1^{(1, 0)} - Y_1^{(0, 0)} | X = x, p_1^{0'} < V_1 \le p_1^0, p_2^1 < V_2 \le 1]$; and the \textit{total} LATE: $\bE [Y_1^{(1, 1)} - Y_1^{(0, 0)} | X = x, p_1^{1'} < V_1 \le p_1^1, p_2^{1'} < V_2 \le p_2^1]$.
The former and latter indicate the average causal effects for the players in regions [A] and [B], respectively.
The pairs of players in region [A] change their treatment status from $D = (1,0)$ to $(0,0)$ as the value of $Z$ shifts from $z$ to $z'$. 
Similarly, the pairs of players in region [B] select $D = (0, 0)$ when $Z = z'$, but $D = (1, 1)$ or $(0, 0)$ when $Z = z$.
As in \citetappendix{heckman2005structural}, we can write the LATE parameters as the weighted averages of the MTE parameters:
\begin{align*}
	\text{Direct LATE} 
	& = \int_{p_2^1}^{1} \int_{p_1^{0'}}^{p_1^0} \text{MTE}_{\text{direct}}^{(0)} (x, v_1, v_2) \frac{h(v_1, v_2 | x)}{\Pr[p_1^{0'} < V_1 \le p_1^0, p_2^1 < V_2 \le 1 | X = x]} \mathrm{d}v_1 \mathrm{d}v_2, \\
	\text{Total LATE}
	& = \int_{p_2^{1'}}^{p_2^1} \int_{p_1^{1'}}^{p_1^1} \text{MTE}_{\text{total}}(x, v_1, v_2) \frac{h(v_1, v_2 | x)}{\Pr[p_1^{1'} < V_1 \le p_1^1, p_2^{1'} < V_2 \le p_2^1 | X = x]}  \mathrm{d}v_1 \mathrm{d}v_2.
\end{align*}
Because the MTE parameters and weight functions in the integrals are identified, the LATE parameters are also identified.

\begin{figure}[!h]
	\centering
	\fbox{\includegraphics[width = 8cm, bb = 0 0 720 540]{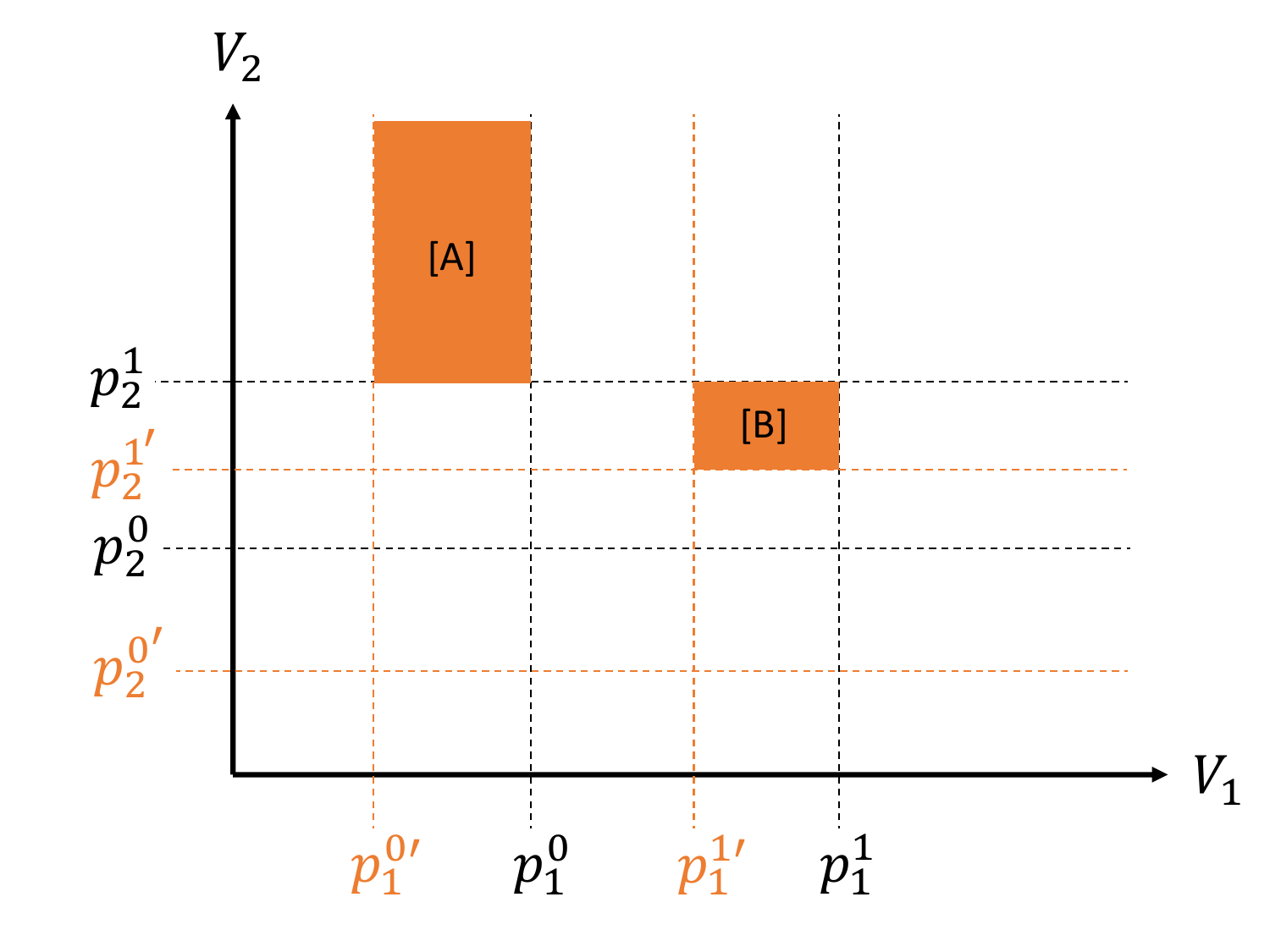}}
	\caption{LATE.}
	\label{fig:LATE1}
\end{figure}

\subsection{PRTE}\label{subsec:PRTE}

The PRTE is the difference in the average outcomes when switching from the baseline policy to a counterfactual policy that induces a change in the distribution of IV (\citealpappendix{heckman2005structural}).
To proceed, given $X = x$, we write $I^{(d_1, d_2)} = \mathcal{I}^{(d_1, d_2)}(\mathbf{P}, V, \epsilon; \lambda_x)$ with a known function $\mathcal{I}^{(d_1, d_2)}$ to clarify the dependence of the treatment decisions on $\mathbf{P}$, $V$, $\epsilon$, and $\lambda_x$.
For example, $I^{(0, 0)} = \mathcal{I}^{(0, 0)}(\mathbf{P}, V, \epsilon; \lambda_x) = \mathbf{1}\{ V \in \mathcal{V}_{\text{uni}}^{(0, 0)}(\mathbf{P}) \} + \mathbf{1}\{ V \in \mathcal{V}_{\text{mul}}(\mathbf{P}), \epsilon \le \lambda_x \}$ as in Assumption \ref{as:multiple}.
We consider a counterfactual policy that does alter $\mathbf{P}$ and/or $\lambda_x$ but does not affect $Y_1^{(d_1, d_2)}$, $X$, $V$, and $\epsilon$.
Let $\mathbf{P}^\star$ be a counterfactual variable of $\mathbf{P}$ with a known distribution (given $X = x$) and $\lambda_x^\star$ be a known counterfactual value of $\lambda_x$ (given $X = x$).
We denote the treatment decisions under $(\mathbf{P}^\star, \lambda_x^\star)$ as $D^\star = (D_1^\star, D_2^\star)$.
Given $X = x$, the outcome for player 1 under the counterfactual policy is $Y_1^\star = \sum_{d_1, d_2} \mathbf{1}\{ D^\star = (d_1, d_2) \} Y_1^{(d_1, d_2)}$, where $\mathbf{1}\{ D^\star = (d_1, d_2) \} = \mathcal{I}^{(d_1, d_2)}(\mathbf{P}^\star, V, \epsilon; \lambda_x^\star)$.
The PRTE is defined as $\bE [Y_1^\star | X = x] - \bE [Y_1 | X = x]$.
The law of iterated expectations leads to $\bE [Y_1^\star | X = x] = \sum_{d_1, d_2} \bE [ \bE [ \mathbf{1}\{ D^\star = (d_1, d_2) \} Y_1^{(d_1, d_2)} | X = x, \mathbf{P}^\star ] | X = x ]$.
\phantomsection\label{page:R2-6}\Copy{R2-6}{
	Under the assumptions made here, we can observe that
	\begin{align*}
		\bE \left[ \mathbf{1}\{ D^\star = (d_1, d_2) \} Y_1^{(d_1, d_2)} \middle| X = x, \mathbf{P}^\star = \mathbf{p}^\star \right] 
		&= \bE \left[ \mathcal{I}^{(d_1, d_2)}(\mathbf{p}^\star, V, \epsilon; \lambda_x^\star) Y_1^{(d_1, d_2)}  \middle| X = x \right] \\
		&= \bE \left[ \mathcal{I}^{(d_1, d_2)}(\mathbf{p}^\star, V, \epsilon; \lambda_x^\star) \bE \left[ Y_1^{(d_1, d_2)} \middle| X = x, V, \epsilon \right] \middle| X = x \right] \\
		&= \int_{[0,1]^3} \mathcal{I}^{(d_1, d_2)}(\mathbf{p}^\star, v, e; \lambda_x^\star) m^{(d_1, d_2)}(x, v_1, v_2) h(v_1, v_2|x) \mathrm{d}v_1 \mathrm{d}v_2 \mathrm{d}e.
	\end{align*}
	Hence, writing $\mathbb{P}^{\star (d_1,d_2)}(x, v_1, v_2) \coloneqq  \int_{[0,1]}\bE \left[ \mathcal{I}^{(d_1, d_2)}(\mathbf{P}^\star, v, e; \lambda_x^\star) \middle| X = x \right]\mathrm{d}e$ for simplicity and noting that $\sum_{d_1, d_2} \mathbb{P}^{\star (d_1,d_2)}(x, v_1, v_2) = 1$, it holds that
	\begin{align*}
		\bE [Y_1^\star | X = x] 
		& = \sum_{d_1, d_2} \int_{[0,1]^2} \mathbb{P}^{\star (d_1,d_2)}(x, v_1, v_2) m^{(d_1, d_2)}(x, v_1, v_2) h(v_1, v_2|x) \mathrm{d}v_1 \mathrm{d}v_2 \\
		& = \int_{[0,1]^2} \mathbb{P}^{\star (1,1)}(x, v_1, v_2) \text{MTE}_\text{total}(x, v_1, v_2) h(v_1, v_2|x) \mathrm{d}v_1 \mathrm{d}v_2 \\
		& \quad + \int_{[0,1]^2} \mathbb{P}^{\star (1,0)}(x, v_1, v_2) \text{MTE}_\text{direct}^{(0)}(x, v_1, v_2) h(v_1, v_2|x) \mathrm{d}v_1 \mathrm{d}v_2 \\
		& \quad + \int_{[0,1]^2} \mathbb{P}^{\star (0,1)}(x, v_1, v_2) \text{MTE}_\text{indirect}^{(0)}(x, v_1, v_2) h(v_1, v_2|x) \mathrm{d}v_1 \mathrm{d}v_2 \\
		& \quad + \int_{[0,1]^2} m^{(0, 0)}(x, v_1, v_2) h(v_1, v_2|x) \mathrm{d}v_1 \mathrm{d}v_2.
	\end{align*}
	In the same manner, defining $\mathbb{P}^{(d_1,d_2)}(x, v_1, v_2)$ similarly, we can observe the following for the baseline policy:
	\begin{align*}
		\bE [Y_1 | X = x] 
		& = \int_{[0,1]^2} \mathbb{P}^{(1,1)}(x, v_1, v_2) \text{MTE}_\text{total}(x, v_1, v_2) h(v_1, v_2|x) \mathrm{d}v_1 \mathrm{d}v_2 \\
		& \quad + \int_{[0,1]^2} \mathbb{P}^{(1,0)}(x, v_1, v_2) \text{MTE}_\text{direct}^{(0)}(x, v_1, v_2) h(v_1, v_2|x) \mathrm{d}v_1 \mathrm{d}v_2 \\
		& \quad + \int_{[0,1]^2} \mathbb{P}^{(0,1)}(x, v_1, v_2) \text{MTE}_\text{indirect}^{(0)}(x, v_1, v_2) h(v_1, v_2|x) \mathrm{d}v_1 \mathrm{d}v_2 \\
		& \quad + \int_{[0,1]^2} m^{(0, 0)}(x, v_1, v_2) h(v_1, v_2|x) \mathrm{d}v_1 \mathrm{d}v_2.
	\end{align*}
	Therefore, we can find that the PRTE can be expressed as the summation of weighted total, direct, and indirect MTEs: 
	\begin{align*}
		\text{PRTE}
		& = \int_{[0,1]^2} \left\{\mathbb{P}^{\star (1,1)}(x, v_1, v_2) - \mathbb{P}^{(1,1)}(x, v_1, v_2)\right\}\text{MTE}_\text{total}(x, v_1, v_2) h(v_1, v_2|x) \mathrm{d}v_1 \mathrm{d}v_2 \\
		& \quad + \int_{[0,1]^2} \left\{\mathbb{P}^{\star (1,0)}(x, v_1, v_2) - \mathbb{P}^{(1,0)}(x, v_1, v_2)\right\} \text{MTE}_\text{direct}^{(0)}(x, v_1, v_2) h(v_1, v_2|x) \mathrm{d}v_1 \mathrm{d}v_2 \\
		& \quad + \int_{[0,1]^2} \left\{\mathbb{P}^{\star (0,1)}(x, v_1, v_2) - \mathbb{P}^{(0,1)}(x, v_1, v_2)\right\}  \text{MTE}_\text{indirect}^{(0)}(x, v_1, v_2) h(v_1, v_2|x) \mathrm{d}v_1 \mathrm{d}v_2.
	\end{align*}
}
\setcounter{table}{0}
\setcounter{figure}{0}

\section{Appendix: Detailed Information on the Monte Carlo Simulations}\label{sec:MC}

This section presents the detailed information on the Monte Carlo simulation analysis summarized in Subsection \ref{subsec:MC}.
The treatment variable is generated by $D_j = \mathbf{1}\{ \gamma_{01} + \gamma_{02} Z_{0j} + D_{-j} \cdot \exp ( \gamma_{11} + \gamma_{12} Z_{1j} ) \ge \varepsilon_j \}$ where $Z_{0j} \sim N(0, 1)$, $Z_{1j} \sim N(0, 1)$, and $\varepsilon_j \sim N(0, 1)$.
The true values of the parameters are $\gamma_0^* = (\gamma_{01}^*, \gamma_{02}^*) = (-0.5, 1)$ and $\gamma_1^* = (\gamma_{11}^*, \gamma_{12}^*) = (-0.3, 0.5)$.
The joint distribution of $(V_j, V_{-j})$ is defined by the FGM copula with the dependence parameter of $\rho^* = 0.7$.
As in Remark \ref{remark:constant}, $D = (0, 0)$ occurs in the region of multiple equilibria if and only if $\epsilon \le \lambda^*$ with $\epsilon \sim \text{Uniform}[0, 1]$ and $\lambda^* = 0.5$.

The potential outcomes are generated by $Y_j^{(d_j, d_{-j})} = \beta_0^{(d_j, d_{-j})} + X_j \beta_1^{(d_j, d_{-j})} + U_j^{(d_j, d_{-j})}$ where $X_j \sim N(0, 1)$ and $U_j^{(d_j, d_{-j})} = \sigma^{(d_j, d_{-j})} \cdot \varphi(V_1, V_2) + \varsigma_j$ with $\sigma^{(1, 1)} = 1.5$, $\sigma^{(1, 0)} = \sigma^{(0, 1)} = 1$, $\sigma^{(0, 0)} = 0.5$, and $\varsigma_j \sim N(0, 0.5^2)$.
For the function $\varphi$, we consider two designs: $\varphi(V_1, V_2) = V_1 + V_2$ in design 1 and $\varphi(V_1, V_2) = \exp(V_1 + V_2)$ in design 2.
The true values of $\beta^{(d_j, d_{-j})} = (\beta_0^{(d_j, d_{-j})}, \beta_1^{(d_j, d_{-j})})$ are set to $\beta^{(1, 0)} = \beta^{(0, 1)} = (2, 1)$ and $\beta^{(1, 1)} = \beta^{(0, 0)} = (1, 2)$.

For the above data-generating processes, we consider two sample sizes $n \in \{1500, 6000\}$ for each.
We evaluate both the feasible and infeasible estimators for $\text{MTE}_{\text{direct}}^{(0)}(x, p_1, p_2)$; the feasible estimator is based on the first-stage ML estimates of $\theta^*$, while the infeasible one treats $\theta^*$ as known. 
We fix $x = 0.5$ and $p_1 = 0.5$ and consider four values of $p_2 \in \{0.2, 0.4, 0.6, 0.8\}$, which results in four MTE values as the parameters of interest, labeled respectively as ``MTE1'', ``MTE2'', and so on.
For the choice of basis function, we employ a bivariate power series and tensor-product B-splines of order 3.
When $n = 1500$, the order of the bivariate power series is set to 3, and the number of inner knots of the tensor-product B-splines is set to 1 in each coordinate.
When $n = 6000$, we consider both 3 and 4 for the orders of the power series and both 1 and 2 for the the number of inner knots of the B-splines.
We estimate the MTR values for $D=(1,0)$ and $D=(0,0)$ in two ways: standard least-squares regression and ridge regression.
For the ridge regression, we set the regularization parameters equal to $n^{-1}$ and $10 \cdot n^{-1}$ for the estimation based on the power series and the B-splines, respectively.\footnote{
	Note that the intercept term should not be penalized (see, e.g., Chapter 3 of \citealpappendix{hastie2009elements}).
}
To estimate the MTR values for $D=(0,0)$, we employ the over-identified estimator introduced in Subsection \ref{subsec:overidentification:est}.
For the smoothed indicator function $\tau_\varpi$, we use
\begin{align*}
	\tau_\varpi(a) = \mathbf{1}\{\varpi \le a < 3 \varpi\} \left( \frac{\exp(a - 2\varpi)}{1 + \exp(a - 2\varpi)} - \frac{\exp(-\varpi)}{1 + \exp(-\varpi)} \right) \left( \frac{2 + \exp(\varpi) + \exp(-\varpi)}{\exp(\varpi) - \exp(-\varpi)} \right) + \mathbf{1}\{ a \ge 3\varpi\},
\end{align*}
with $\varpi = 0.01$.\footnote{
	\phantomsection\label{page:R2-7-c}\Copy{R2-7-c}{
		This trimming function is taken from \citeappendix{buchinsky1998alternative}.
		We observed that about 0.6 percent of the observations were trimmed out on average with this trimming function.
		We also found that altering the trimming constant $\varpi$ to 0.005 or 0.05 had only negligible impacts on our results.
}}
The following results are based on 1,000 Monte Carlo replications. 

Table \ref{table:MC} presents the bias and RMSE for the direct MTE estimation.
Notably, the performances of the feasible and infeasible estimators are almost identical, which is consistent with our theory.
Further, the estimator based on the bivariate power series outperforms that based on the tensor-product B-splines.
The precision for the bivariate power series is satisfactory even with the sample size of 1,500, while the estimator with the tensor-product B-splines is not reliable for small samples.
However, the accuracy of the B-splines-based estimator can be improved rapidly as the sample size increases from 1,500 to 6,000.
In addition, using the ridge regression can further improve the performance of the estimators, particularly for the B-splines-based estimator with the sample size of 1,500.
The performances of the estimators deteriorate when $p_2$ is near the boundary.
Since the realizations are often sparse near the boundary, this observation is reasonable.

Table \ref{table:ML} presents the results for the first-stage treatment decision model.
The initial values for the ML estimation are randomly chosen from uniform distributions with reasonable ranges.
We can observe that both the bias and RMSE are satisfactorily small even for a small sample size.
This also demonstrates the validity of the global identification result established in Subsection \ref{subsec:first}.

\begin{table}[!tbp]
	{\scriptsize
		\caption{Simulation results for the direct MTE estimation\label{table:MC}} 
		\begin{center}
			\begin{tabular}{rrrrrcrrrrcrrrr}
				\hline\hline
				&\multicolumn{4}{c}{\bfseries }&\multicolumn{1}{c}{\bfseries }&\multicolumn{4}{c}{\bfseries Feasible estimator}&\multicolumn{1}{c}{\bfseries }&\multicolumn{4}{c}{\bfseries Infeasible estimator}\tabularnewline
				\cline{7-10} \cline{12-15}
				&\multicolumn{1}{c}{design}&\multicolumn{1}{c}{$n$}&\multicolumn{1}{c}{order/\#knots}&\multicolumn{1}{c}{ridge}&\multicolumn{1}{c}{}&\multicolumn{1}{c}{MTE1}&\multicolumn{1}{c}{MTE2}&\multicolumn{1}{c}{MTE3}&\multicolumn{1}{c}{MTE4}&\multicolumn{1}{c}{}&\multicolumn{1}{c}{MTE1}&\multicolumn{1}{c}{MTE2}&\multicolumn{1}{c}{MTE3}&\multicolumn{1}{c}{MTE4}\tabularnewline
				\hline
				\multicolumn{15}{l}{\bfseries Bias (bivariate power series)} \tabularnewline
				&$1$&$1500$&$3$&$0$&&$-0.073$&$-0.029$&$ 0.020$&$ 0.063$&&$-0.061$&$-0.026$&$ 0.012$&$ 0.045$\tabularnewline
				&$1$&$1500$&$3$&$1$&&$-0.120$&$-0.071$&$-0.023$&$ 0.024$&&$-0.115$&$-0.071$&$-0.027$&$ 0.018$\tabularnewline
				&$1$&$6000$&$3$&$0$&&$-0.012$&$-0.002$&$ 0.010$&$ 0.014$&&$-0.012$&$-0.003$&$ 0.009$&$ 0.013$\tabularnewline
				&$1$&$6000$&$3$&$1$&&$-0.090$&$-0.046$&$-0.004$&$ 0.039$&&$-0.091$&$-0.048$&$-0.005$&$ 0.038$\tabularnewline
				&$1$&$6000$&$4$&$0$&&$-0.050$&$-0.005$&$ 0.017$&$ 0.005$&&$-0.048$&$-0.003$&$ 0.018$&$ 0.007$\tabularnewline
				&$1$&$6000$&$4$&$1$&&$-0.124$&$-0.075$&$-0.014$&$ 0.059$&&$-0.126$&$-0.076$&$-0.014$&$ 0.060$\tabularnewline
				&$2$&$1500$&$3$&$0$&&$-0.170$&$-0.021$&$ 0.085$&$ 0.129$&&$-0.136$&$ 0.000$&$ 0.078$&$ 0.103$\tabularnewline
				&$2$&$1500$&$3$&$1$&&$-0.250$&$-0.086$&$ 0.024$&$ 0.073$&&$-0.240$&$-0.087$&$ 0.016$&$ 0.059$\tabularnewline
				&$2$&$6000$&$3$&$0$&&$-0.076$&$ 0.010$&$ 0.049$&$ 0.020$&&$-0.074$&$ 0.011$&$ 0.050$&$ 0.020$\tabularnewline
				&$2$&$6000$&$3$&$1$&&$-0.179$&$-0.032$&$ 0.063$&$ 0.098$&&$-0.183$&$-0.036$&$ 0.059$&$ 0.093$\tabularnewline
				&$2$&$6000$&$4$&$0$&&$-0.127$&$-0.031$&$ 0.034$&$ 0.045$&&$-0.121$&$-0.026$&$ 0.038$&$ 0.052$\tabularnewline
				&$2$&$6000$&$4$&$1$&&$-0.306$&$-0.162$&$-0.003$&$ 0.160$&&$-0.311$&$-0.165$&$-0.004$&$ 0.161$\tabularnewline
				\hline
				\multicolumn{15}{l}{\bfseries RMSE (bivariate power series)} \tabularnewline
				&$1$&$1500$&$3$&$0$&&$ 0.649$&$ 0.432$&$ 0.481$&$ 0.701$&&$ 0.659$&$ 0.442$&$ 0.487$&$ 0.709$\tabularnewline
				&$1$&$1500$&$3$&$1$&&$ 0.242$&$ 0.188$&$ 0.182$&$ 0.227$&&$ 0.244$&$ 0.194$&$ 0.190$&$ 0.234$\tabularnewline
				&$1$&$6000$&$3$&$0$&&$ 0.284$&$ 0.179$&$ 0.177$&$ 0.273$&&$ 0.282$&$ 0.179$&$ 0.177$&$ 0.270$\tabularnewline
				&$1$&$6000$&$3$&$1$&&$ 0.187$&$ 0.118$&$ 0.100$&$ 0.152$&&$ 0.187$&$ 0.119$&$ 0.100$&$ 0.151$\tabularnewline
				&$1$&$6000$&$4$&$0$&&$ 0.459$&$ 0.388$&$ 0.379$&$ 0.441$&&$ 0.460$&$ 0.393$&$ 0.383$&$ 0.440$\tabularnewline
				&$1$&$6000$&$4$&$1$&&$ 0.217$&$ 0.142$&$ 0.114$&$ 0.207$&&$ 0.217$&$ 0.142$&$ 0.114$&$ 0.205$\tabularnewline
				&$2$&$1500$&$3$&$0$&&$ 1.102$&$ 0.684$&$ 0.713$&$ 1.081$&&$ 1.100$&$ 0.692$&$ 0.729$&$ 1.101$\tabularnewline
				&$2$&$1500$&$3$&$1$&&$ 0.456$&$ 0.322$&$ 0.313$&$ 0.394$&&$ 0.453$&$ 0.328$&$ 0.324$&$ 0.405$\tabularnewline
				&$2$&$6000$&$3$&$0$&&$ 0.469$&$ 0.279$&$ 0.263$&$ 0.409$&&$ 0.468$&$ 0.279$&$ 0.262$&$ 0.404$\tabularnewline
				&$2$&$6000$&$3$&$1$&&$ 0.334$&$ 0.185$&$ 0.174$&$ 0.260$&&$ 0.333$&$ 0.185$&$ 0.173$&$ 0.257$\tabularnewline
				&$2$&$6000$&$4$&$0$&&$ 0.750$&$ 0.577$&$ 0.584$&$ 0.658$&&$ 0.751$&$ 0.585$&$ 0.589$&$ 0.652$\tabularnewline
				&$2$&$6000$&$4$&$1$&&$ 0.431$&$ 0.254$&$ 0.175$&$ 0.351$&&$ 0.432$&$ 0.255$&$ 0.173$&$ 0.344$\tabularnewline
				\hline
				\multicolumn{15}{l}{\bfseries Bias (tensor-product B-splines)} \tabularnewline
				&$1$&$1500$&$1$&$0$&&$-0.937$&$-0.243$&$ 0.354$&$-0.113$&&$-0.932$&$-0.267$&$ 0.356$&$-0.143$\tabularnewline
				&$1$&$1500$&$1$&$1$&&$-0.628$&$-0.298$&$ 0.016$&$-0.074$&&$-0.624$&$-0.288$&$ 0.028$&$-0.079$\tabularnewline
				&$1$&$6000$&$1$&$0$&&$-0.975$&$-0.312$&$ 0.297$&$-0.038$&&$-0.974$&$-0.307$&$ 0.300$&$-0.035$\tabularnewline
				&$1$&$6000$&$1$&$1$&&$-0.780$&$-0.309$&$ 0.152$&$-0.077$&&$-0.780$&$-0.313$&$ 0.149$&$-0.076$\tabularnewline
				&$1$&$6000$&$2$&$0$&&$-0.586$&$ 0.247$&$ 0.105$&$-0.322$&&$-0.557$&$ 0.280$&$ 0.114$&$-0.330$\tabularnewline
				&$1$&$6000$&$2$&$1$&&$-0.777$&$ 0.114$&$ 0.035$&$-0.150$&&$-0.772$&$ 0.118$&$ 0.029$&$-0.158$\tabularnewline
				&$2$&$1500$&$1$&$0$&&$-2.523$&$-0.645$&$ 0.827$&$-0.257$&&$-2.525$&$-0.672$&$ 0.858$&$-0.294$\tabularnewline
				&$2$&$1500$&$1$&$1$&&$-1.606$&$-0.768$&$-0.042$&$-0.132$&&$-1.606$&$-0.740$&$-0.003$&$-0.133$\tabularnewline
				&$2$&$6000$&$1$&$0$&&$-2.587$&$-0.797$&$ 0.676$&$ 0.038$&&$-2.587$&$-0.793$&$ 0.685$&$ 0.049$\tabularnewline
				&$2$&$6000$&$1$&$1$&&$-2.016$&$-0.829$&$ 0.269$&$-0.114$&&$-2.018$&$-0.840$&$ 0.262$&$-0.109$\tabularnewline
				&$2$&$6000$&$2$&$0$&&$-1.336$&$ 0.544$&$ 0.287$&$-0.635$&&$-1.308$&$ 0.630$&$ 0.313$&$-0.644$\tabularnewline
				&$2$&$6000$&$2$&$1$&&$-1.871$&$ 0.186$&$-0.006$&$-0.403$&&$-1.863$&$ 0.196$&$-0.017$&$-0.427$\tabularnewline
				\hline
				\multicolumn{15}{l}{\bfseries RMSE (tensor-product B-splines)} \tabularnewline
				&$1$&$1500$&$1$&$0$&&$ 1.360$&$ 1.247$&$ 1.351$&$ 1.979$&&$ 1.349$&$ 1.268$&$ 1.429$&$ 1.933$\tabularnewline
				&$1$&$1500$&$1$&$1$&&$ 0.663$&$ 0.340$&$ 0.293$&$ 0.331$&&$ 0.660$&$ 0.332$&$ 0.302$&$ 0.342$\tabularnewline
				&$1$&$6000$&$1$&$0$&&$ 1.058$&$ 0.629$&$ 0.544$&$ 0.539$&&$ 1.055$&$ 0.629$&$ 0.544$&$ 0.540$\tabularnewline
				&$1$&$6000$&$1$&$1$&&$ 0.796$&$ 0.331$&$ 0.258$&$ 0.255$&&$ 0.795$&$ 0.334$&$ 0.256$&$ 0.255$\tabularnewline
				&$1$&$6000$&$2$&$0$&&$ 2.154$&$ 1.404$&$ 0.950$&$ 1.871$&&$ 2.074$&$ 1.422$&$ 0.943$&$ 1.834$\tabularnewline
				&$1$&$6000$&$2$&$1$&&$ 0.837$&$ 0.576$&$ 0.350$&$ 0.696$&&$ 0.833$&$ 0.575$&$ 0.345$&$ 0.691$\tabularnewline
				&$2$&$1500$&$1$&$0$&&$ 2.950$&$ 1.935$&$ 2.080$&$ 2.916$&&$ 2.937$&$ 1.978$&$ 2.193$&$ 2.950$\tabularnewline
				&$2$&$1500$&$1$&$1$&&$ 1.644$&$ 0.817$&$ 0.485$&$ 0.536$&&$ 1.644$&$ 0.790$&$ 0.492$&$ 0.550$\tabularnewline
				&$2$&$6000$&$1$&$0$&&$ 2.660$&$ 1.122$&$ 0.958$&$ 0.783$&&$ 2.658$&$ 1.122$&$ 0.959$&$ 0.779$\tabularnewline
				&$2$&$6000$&$1$&$1$&&$ 2.033$&$ 0.849$&$ 0.441$&$ 0.382$&&$ 2.034$&$ 0.859$&$ 0.429$&$ 0.375$\tabularnewline
				&$2$&$6000$&$2$&$0$&&$ 3.418$&$ 2.166$&$ 1.453$&$ 2.759$&&$ 3.284$&$ 2.187$&$ 1.435$&$ 2.740$\tabularnewline
				&$2$&$6000$&$2$&$1$&&$ 1.933$&$ 0.867$&$ 0.539$&$ 1.077$&&$ 1.925$&$ 0.862$&$ 0.532$&$ 1.069$\tabularnewline
				\hline
		\end{tabular}
	\end{center}
	Note: The column labeled ``order/\#knots'' indicates the order of the bivariate power series or the number of inner knots of the univariate B-splines.
	The column labeled ``ridge'' indicates whether the ridge regression is used (1 for ``yes'' and 0 for ``no'').
}
\end{table}

\begin{table}[!tbp]
	{\scriptsize
		\caption{Simulation results for the ML estimation of the treatment decision game\label{table:ML}} 
		\begin{center}
			\begin{tabular}{rrrrrrrr}
				\hline\hline
				&\multicolumn{1}{c}{$n$}&\multicolumn{1}{c}{$\gamma_{01}$}&\multicolumn{1}{c}{$\gamma_{02}$}&\multicolumn{1}{c}{$\gamma_{11}$}&\multicolumn{1}{c}{$\gamma_{12}$}&\multicolumn{1}{c}{$\rho$}&\multicolumn{1}{c}{$\lambda$}\tabularnewline
				\hline
				\multicolumn{8}{l}{\bfseries Bias} \tabularnewline
				&$1500$&$-0.004$&$0.002$&$ 0.003$&$0.001$&$-0.027$&$-0.007$\tabularnewline
				&$6000$&$ 0.001$&$0.000$&$-0.003$&$0.003$&$-0.002$&$ 0.002$\tabularnewline
				\hline
				\multicolumn{8}{l}{\bfseries RMSE} \tabularnewline
				&$1500$&$ 0.060$&$0.038$&$ 0.133$&$0.073$&$ 0.264$&$ 0.219$\tabularnewline
				&$6000$&$ 0.032$&$0.020$&$ 0.073$&$0.037$&$ 0.143$&$ 0.116$\tabularnewline
				\hline
		\end{tabular}
	\end{center}}
\end{table}

\setcounter{table}{0}
\setcounter{figure}{0}

\section{Appendix: Supplementary Material for the Empirical Analysis}\label{sec:supp}

This appendix provides supplementary tables and figures for the empirical illustration in Subsection \ref{subsec:empirical}.

\begingroup
\renewcommand{\arraystretch}{1.2}
\begin{table}[!h]
	\caption{Definitions of the variables}
	\begin{center}
		\scriptsize{\begin{tabular}{l|l}
				\hline \hline
				\multicolumn{1}{c|}{Variables} & \multicolumn{1}{|c}{Definitions}\\
				\hline
				Outcome & log(GPA + 1)\\
				Treatment & 1 if at least one of (Smoke, Drink, Skip, Fight) is larger than 2; 0 otherwise:\\
				\quad Smoke &  How often the respondent smoked cigarettes. (0: never - 6: nearly every day)\\
				\quad Drink &  How often the respondent drank alcohols. (0: never - 6: nearly every day)\\
				\quad Skip & How often the respondent skipped school without excuses. (0: never - 6: nearly everyday)\\
				\quad Fight & How often the respondent got into a physical fight in the past year. (0: never - 4: more than 7 times)\\
				Age & Age\\
				Grade & Grade\\
				White & 1 if the respondent is White; 0 otherwise.\\
				Black & 1 if the respondent is Black or African; 0 otherwise.\\
				Asian & 1 if the respondent is Asian; 0 otherwise.\\
				Mother's education & The respondent's mother's education level in years.\\
				Mother's job (professional) & 1 if the respondent's mother is a worker with expertise or a managerial worker; 0 otherwise.\\
				Mother's job (unemployed) & 1 if the respondent's mother is not employed (except for housewife); 0 otherwise.\\
				Father's education & The respondent's father's education level in years.\\
				Father's job (professional) & 1 if the respondent's father is a worker with expertise or a managerial worker; 0 otherwise.\\
				Father's job (unemployed) & 1 if the respondent's father is not employed (except for househusband); 0 otherwise.\\
				Academic club & 1 if the respondent belongs to an academic club; 0 otherwise.\\
				Sports club & 1 if the respondent belongs to a sport club; 0 otherwise.\\
				\hline
		\end{tabular}}
		\label{table:def}
	\end{center}
\end{table}
\endgroup

\phantomsection\label{page:AE-19-2}\Copy{AE-19-2}{
	Table \ref{table:KMtest} summarizes the results of \citemain{kedagni2020generalized}'s test for the independence of each IV from the potential outcomes.
	The test is based on checking whether an inequality: $\sum_{k = 1}^K \sup_z \mathbb{E}[W_k | Z = z] \le 0$ is true (in their notation), which must hold when the IV $Z$ is actually independent from the potential outcomes.
	In its basic form, the outcome variable is supposed to be binary.
	Thus, we binarize the GPA variable, and consider the following three cases: $\mathbf{1}\{\text{GPA} \le \text{25th percentile point}\}$, $\mathbf{1}\{\text{GPA} \le \text{50th percentile point}\}$, and $\mathbf{1}\{\text{GPA} \le \text{75th percentile point}\}$. 
	Because the testing procedure requires splitting the data into many subsamples, we performed the test by merging the male and female data to maintain the size of each subsample.
	\citemain{kedagni2020generalized} also discuss how to incorporate additional covariates into the testing procedure, but for simplicity of implementation, we ignore additional control variables.
	The numbers reported in Table \ref{table:KMtest} are the lower bounds of the one-sided confidence intervals of $\sum_{k = 1}^K \sup_z \mathbb{E}[W_k | Z = z]$ at each corresponding confidence level.
	The lower bounds are negative for all IV candidates, implying that there is no strong evidence to doubt their independence property.
}

\begingroup
\renewcommand{\arraystretch}{1.2}
\begin{table}[!h]
	\caption{Testing the IV independence assumption}
	\begin{center}
		\scriptsize{\begin{tabular}{l|ccc|ccc|ccc}
			\hline \hline
			\multicolumn{1}{r|}{Dependent variable} & \multicolumn{3}{c|}{$\mathbf{1}\{\text{GPA} \le \text{25th percentile point}\}$} & \multicolumn{3}{c|}{$\mathbf{1}\{\text{GPA} \le \text{50th percentile point}\}$} & \multicolumn{3}{c}{$\mathbf{1}\{\text{GPA} \le \text{75th percentile point}\}$} \\ 
				\multicolumn{1}{r|}{Confidence level} & 90\% & 95\% & 99\% & 90\% & 95\% & 99\% & 90\% & 95\% & 99\% \\ \hline
				log(Mother's education + 1) & -0.4000 & -0.4044 & -0.4143 & -0.4042 & -0.4081 & -0.4152 & -0.3873 & -0.3907 & -0.3978 \\ 
				Mother's job (professional) & -0.4090 & -0.4170 & -0.4253 & -0.3873 & -0.3967 & -0.4084 & -0.3690 & -0.3797 & -0.3905 \\ 
				Mother's job (unemployed) & -0.4166 & -0.4194 & -0.4253 & -0.4052 & -0.4087 & -0.4155 & -0.3971 & -0.4003 & -0.4070 \\ 
				log(Father's education + 1) & -0.3799 & -0.3852 & -0.3956 & -0.4080 & -0.4113 & -0.4183 & -0.3896 & -0.3933 & -0.3997 \\ 
				Father's job (professional) & -0.4109 & -0.4148 & -0.4224 & -0.4034 & -0.4076 & -0.4159 & -0.3976 & -0.4018 & -0.4102 \\ 
				Father's job (unemployed) & -0.3812 & -0.4029 & -0.4183 & -0.4145 & -0.4179 & -0.4251 & -0.3980 & -0.4012 & -0.4075 \\ 
				\hline
			\end{tabular}}
		\label{table:KMtest}
	\end{center}
\end{table}
\endgroup

\begingroup
\renewcommand{\arraystretch}{1.2}
\begin{table}[!h]
	\caption{Estimation results of the treatment decision model}
	\begin{center}
		\footnotesize{\begin{tabular}{llccllcc}
				\hline \hline
				&  \multicolumn{1}{c}{$\gamma_0$} & Estimate & $t$-value &  &   \multicolumn{1}{c}{$\gamma_1$} & Estimate & $t$-value\\
				\hline
				$X$ & Intercept & -5.600  & -8.066  & $X$ & Intercept & -0.450  & -0.061  \\
				 & log(Age) & 2.210  & 6.822  &  & log(Age) & -0.993  & -0.307  \\
				 & Grade  & -0.047  & -2.148  &  & Grade  & 0.064  & 0.301  \\
				 & White & -0.078  & -1.647  &  & White & 1.437  & 2.011  \\
				 & Black & -0.360  & -6.026  &  & Black & -0.062  & -0.075  \\
				 & Asian & -0.360  & -4.698  &  & Asian & 1.906  & 2.294  \\
				 & log(Mother's education + 1) & -0.064  & -2.652  &  & log(Mother's education + 1) & -0.159  & -0.598  \\
				 & Mother's job (professional) & -0.048  & -1.264  &  & Mother's job (professional) & 0.134  & 0.324  \\
				 & Mother's job (unemployed) & -0.038  & -0.508  &  & Mother's job (unemployed) & -0.199  & -0.208  \\
				 & log(Father's education + 1) & -0.088  & -5.658  &  & log(Father's education + 1) & 0.577  & 2.191  \\
				 & Father's job (professional) & -0.044  & -0.874  &  & Father's job (professional) & -0.021  & -0.040  \\
				 & Father's job (unemployed) & -0.048  & -0.588  &  & Father's job (unemployed) & 0.045  & 0.053  \\
				 & Academic club & -0.130  & -3.446  &  & Academic club & -0.950  & -1.669  \\
				 & Sports club & 0.006  & 0.181  &  & Sports club & -1.265  & -2.720  \\
				$Z$ & log(Mother's education + 1) & 0.059  & 2.349  & $Z$ & log(Mother's education + 1) & -0.627  & -1.979  \\
				 & Mother's job (professional) & -0.153  & -2.653  &  & Mother's job (professional) & 1.145  & 1.755  \\
				 & Mother's job (unemployed) & 0.128  & 1.111  &  & Mother's job (unemployed) & -0.180  & -0.121  \\
				 & log(Father's education + 1) & -0.044  & -1.822  &  & log(Father's education + 1) & 0.211  & 0.712  \\
				 & Father's job (professional) & -0.142  & -1.996  &  & Father's job (professional) & 0.150  & 0.197  \\
				 & Father's job (unemployed) & -0.218  & -1.734  &  & Father's job (unemployed) & 0.123  & 0.095  \\
				 &  &  &  &  &  &  &  \\
				 & $\rho$ & -0.810  & -1.698  &  & Sample size & 6053 &  \\
				 & $\lambda$ & 0.761  & 2.723  &  & Log-likelihood & -7425.106  &  \\
				\hline
		\end{tabular}}
		\label{table:1ststage}
	\end{center}
	\footnotesize{
		Note: The above result is based on the following model: $D_j = \mathbf{1}\{ \tilde{W}_j^\top\gamma_0 + D_{-j}\cdot \exp(W_j^\top\gamma_1)^{0.25} \ge \varepsilon_j \}$, where $\tilde{W}_j = (X_j, Z_j, \text{School dummies}_j)$, and $W_j = (X_j, Z_j)$.
		The IVs $(Z_j)$ are variables that are defined by averaging the variables listed above (i.e., log(Mother's education + 1), \ldots, Father's job (unemployed)) over the 4th-5th best friends.
		The school dummy variable is introduced only for large schools where the number of respondents in the sample data is larger than or equal to 120.
		The results for the school dummies are omitted to save space.
		The estimation procedure is the same as in the Monte Carlo experiments.
	}
\end{table}
\endgroup

\clearpage

\begin{figure}[!h]
	\centering
	\includegraphics[width = 16cm]{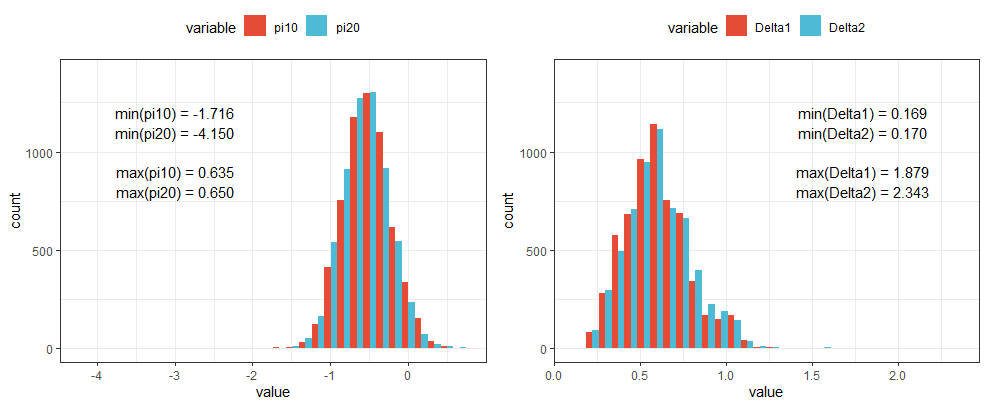}
	\caption{Histograms of the estimated $(\pi_1^0(W_1), \pi_2^0(W_2))$ and $(\Delta_1(W_1), \Delta_2(W_2))$ in the left and right panels, respectively.}
	\label{fig:PiDelhist}
\end{figure}

\begin{figure}[!h]
	\centering
	\includegraphics[width = 16cm]{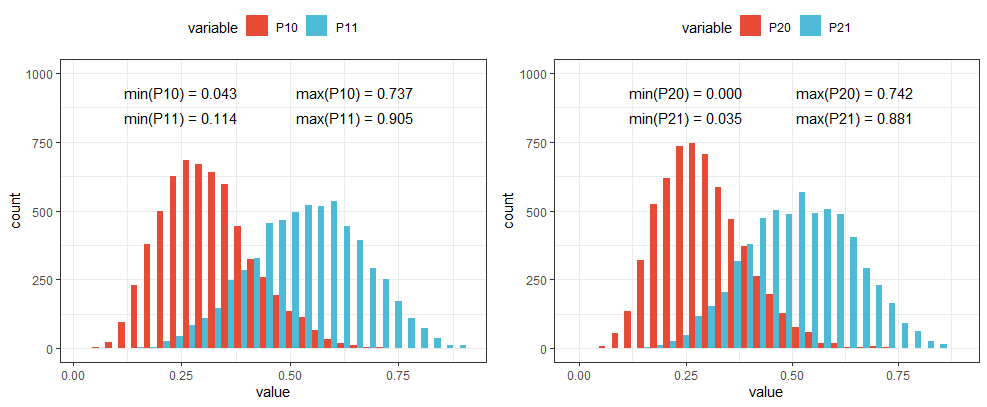}
	\caption{Histograms of $(\hat P_1^0, \hat P_1^1)$ and $(\hat P_2^0, \hat P_2^1)$ in left and right panels, respectively. }
	\label{fig:Phist}
\end{figure}

\clearpage
\begin{center}
	\it \large{References cited in the main text}.
\end{center}
\begin{small}
	\bibliographystylemain{tandfx}
	\bibliographymain{ref}
\end{small}

\clearpage
\begin{center}
	\it \large{References cited in the appendices}.
\end{center}
\begin{small}
	\bibliographystyleappendix{tandfx}
	\bibliographyappendix{ref}
\end{small}

\end{document}